\documentclass[10pt,a4paper]{article}
\usepackage{authblk}

\usepackage{tikz}
\usetikzlibrary{arrows,shapes,automata,backgrounds,petri,positioning}
\DeclareSymbolFont{LTL}{U}{ltlfonts}{m}{n}

\DeclareMathAccent \LTLhat               {\mathord}{LTL}{"02}

\DeclareMathSymbol \LTLsquare            {\mathord}{LTL}{"08}
\DeclareMathSymbol \LTLdiamond           {\mathord}{LTL}{"09}
\DeclareMathSymbol \LTLcircle            {\mathord}{LTL}{"0A}
\DeclareMathSymbol \LTLsquareminus       {\mathord}{LTL}{"0B}
\DeclareMathSymbol \LTLdiamondminus      {\mathord}{LTL}{"0C}
\DeclareMathSymbol \LTLcircleminus       {\mathord}{LTL}{"0D}
\DeclareMathSymbol \LTLcircletilde       {\mathord}{LTL}{"0E}

\DeclareMathSymbol \LTLsquareuc          {\mathord}{LTL}{"10}
\DeclareMathSymbol \LTLdiamonduc         {\mathord}{LTL}{"11}
\DeclareMathSymbol \LTLcircleuc          {\mathord}{LTL}{"12}
\DeclareMathSymbol \LTLsquareminusuc     {\mathord}{LTL}{"13}
\DeclareMathSymbol \LTLdiamondminusuc    {\mathord}{LTL}{"14}
\DeclareMathSymbol \LTLcircleminusuc     {\mathord}{LTL}{"15}
\DeclareMathSymbol \LTLcircletildeuc     {\mathord}{LTL}{"16}

\DeclareMathSymbol \LTLsquarehat         {\mathord}{LTL}{"18}
\DeclareMathSymbol \LTLdiamondhat        {\mathord}{LTL}{"19}
\DeclareMathSymbol \LTLsquareminushat    {\mathord}{LTL}{"1B}
\DeclareMathSymbol \LTLdiamondminushat   {\mathord}{LTL}{"1C}

\DeclareMathSymbol \LTLsquarehatuc       {\mathord}{LTL}{"20}
\DeclareMathSymbol \LTLdiamondhatuc      {\mathord}{LTL}{"21}
\DeclareMathSymbol \LTLsquareminushatuc  {\mathord}{LTL}{"23}
\DeclareMathSymbol \LTLdiamondminushatuc {\mathord}{LTL}{"24}

\RequirePackage{amsmath}
\RequirePackage{amssymb} %
\RequirePackage{comment}
\RequirePackage{dsfont}
\RequirePackage{multidef,xspace}
\RequirePackage{xcolor}
\RequirePackage{graphicx}
\RequirePackage{tikz}
\RequirePackage[disable]{todonotes} %
\RequirePackage{url}
\RequirePackage[hidelinks]{hyperref}
\RequirePackage[capitalize]{cleveref} %
\RequirePackage{stmaryrd} %
\RequirePackage[notion, quotation, silent,paper]{knowledge}
\RequirePackage{mathtools} %
\RequirePackage{enumitem}
\usetikzlibrary{automata, positioning, calc, backgrounds}
\usepackage[appendix=inline,bibliography=common]{apxproof} %
\newtheoremrep{theorem}{Theorem}[section]
\newtheoremrep{lemma}[theorem]{Lemma}
\newtheoremrep{proposition}[theorem]{Proposition}
\newtheoremrep{definition}[theorem]{Definition}
\newtheoremrep{example}[theorem]{Example}
\newtheoremrep{remark}[theorem]{Remark}

\let\mathbb\mathds
\renewcommand{\geq}{\geqslant}
\renewcommand{\leq}{\leqslant}
\renewcommand{\ge}{\geqslant}
\renewcommand{\le}{\leqslant}

\newcommand{\nw}[1]{\hbox to 0pt{\color{blue!60!black}\hss
    \textsuperscript{\bfseries*}}%
    \marginpar{\color{blue!60!black}NW: #1}}

\multidef{\textup{\textsc{#1}}\xspace}{co,coNP,NP,PSPACE,coPSPACE,NPSPACE,PTIME,EXPTIME,EXPSPACE,NEXPTIME, LOGSPACE}

\multidef{\textup{\textsc{#1}}\xspace}{COVER, TARGET, PRP, conjPRP, SAT, AS, LS, PIP, cPIP, cCOVER, ASOF}

\newcommand{\chana}[1]{\todo[color=green!20]{\small #1}}

\newcommand{\nats}{\mathbb{N}}
\newcommand{\set}[1]{\{#1\}}
\newcommand{\nset}[2]{[#1,#2]}
\newcommand{\nsetinfinity}[1]{\mathord{[#1, +\infty[}}
\newcommand{\deff}{:=}
\newcommand{\multisetsof}[1]{\mathcal{M}(#1)}
\newcommand{\ms}{\mu}

\newcommand{\concat}{\,}
\newcommand{\size}[1]{|#1|}
\newcommand{\norm}[1]{||#1||}

\newcommand{\biglor}{\bigvee}

\newcommand{\prot}{\mathcal{P}}
\newcommand{\trace}{\run}
\newcommand{\atrace}{\run}
\newcommand{\accrun}{\alpha}

\newcommand{\fairrunsfrom}[1]{\mathsf{FRuns}(#1)}
\newcommand{\config}{\gamma}
\newcommand{\aconfig}{\config}
\newcommand{\configs}{\Gamma}
\newcommand{\initialstates}{I}
\newcommand{\initialconfigs}{\mathcal{I}}
\newcommand{\preop}{\mathsf{pre}}
\newcommand{\postop}{\mathsf{post}}
\newcommand{\prestar}[1]{\preop^*(#1)}
\newcommand{\poststar}[1]{\postop^*(#1)}

\newcommand{\step}[1]{\xrightarrow{#1}}
\newcommand{\run}{\rho}
\newcommand{\transitions}{\Delta}
\newcommand*{\states}{Q}
\newcommand*{\atrans}{t}
\newcommand*{\transat}[2]{\mathsf{tr}_{#2}(#1)}

\newcommand{\aquantif}{Q}

\newcommand{\satset}[1]{\llbracket {#1} \rrbracket}

\newcommand*{\apoly}{\mathsf{P}}
\newcommand{\varx}{\mathsf{x}}
\newcommand*{\startstate}{\mathsf{start}}
\newcommand*{\statex}[1]{X_{#1}}
\newcommand*{\okaction}{\mathsf{ok}}
\newcommand*{\okstate}{q_f}
\newcommand{\statey}{Y}
\newcommand{\maxexp}{\delta}
\newcommand{\reservoirstate}{R}

\newcommand{\PP}{\mathcal{P}}                  %
\newcommand{\trans}[1]{\xrightarrow{#1}}       %
\newcommand{\initial}{\mathcal{I}}   %

\newcommand{\cSet}{\mathcal{S}}  %

\newcommand*{\buchiaut}{\mathcal{A}}
\newcommand*{\ltlaut}[1]{\buchiaut_{#1}}
\newcommand{\bstates}{\mathcal{L}}
\newcommand{\btransitions}{T}
\newcommand*{\abuchistate}{\ell}
\newcommand*{\controlloc}{\abuchistate}

\newcommand*{\bstate}{\abuchistate}
\newcommand*{\binitstate}{\bstate_0}
\newcommand*{\bwinning}{\mathcal{W}}

\newcommand{\prodsystem}{\mathcal{PS}}
\newcommand{\prodconfigs}{\mathcal{C}}
\newcommand{\prodconfigsproj}[1]{\prodconfigs_{#1}}
\newcommand{\controlsize}[1]{|#1|_{\mathrm{cont}}}
\newcommand{\protocolsize}[1]{|#1|_{\mathrm{prot}}}
\newcommand{\prodconfig}{c}
\newcommand{\accstep}[1]{\xrightarrow[\mathsf{acc}]{#1}}

\newcommand{\prodinitialconfigs}{\prodconfigs_0}
\newcommand{\structuralbound}{\mathsf{B}}

\newcommand*{\setcomplement}[1]{\overline{#1}}

\newcommand*{\lewqo}{\preceq}
\let\olduparrow\uparrow
\renewcommand*{\uparrow}{{\,\olduparrow}}
\let\olddownarrow\downarrow
\renewcommand*{\downarrow}{{\,\olddownarrow}}
\newcommand*{\downwardclosure}[1]{\downarrow #1}
\newcommand*{\upwardclosure}[1]{\uparrow #1}

\newcommand*{\dummysymb}{\#}
\newcommand*{\natsanddummy}{\nats_{\dummysymb}}
\newcommand*{\transferflows}{\mathcal{F}}
\newcommand*{\tfweight}[1]{\mathsf{weight}(#1)}
\newcommand*{\minwqo}[1]{\mathsf{basis}(#1)}
\newcommand*{\transfersof}[1]{F[#1]}
\newcommand*{\atf}{\mathsf{tf}}
\newcommand*{\tfset}{F}
\newcommand*{\tftimes}{\otimes}
\newcommand*{\flowstep}[1]{\xhookrightarrow{#1}}
\newcommand*{\letf}{\preceq}
\newcommand*{\emptysequence}{\epsilon}
\newcommand*{\omeganats}{\nats_{\omega}}
\newcommand*{\preflows}[2]{\preop_{#1}(#2)}
\newcommand*{\zerovectors}{V_0}

\tikzset{controlpart/.style = {color = blue!50!black, draw, shape = rectangle}}
\tikzset{transferedge/.style = {-stealth, thick}}
\tikzset{zerotransfer/.style = {-stealth, thick, dashed}}
\tikzset{flowstate/.style = {color = black, draw, shape = circle}}
\tikzset{sepline/.style = {densely dotted, thick}}
\definecolor{color1}{HTML}{648FFF}
\definecolor{color4}{HTML}{785EF0}
\definecolor{color3}{HTML}{DC267F}
\definecolor{color2}{HTML}{FE6100}
\definecolor{color5}{HTML}{FFB000}

\newcommand*{\confpair}[2]{(#1,#2)}
\newcommand*{\ams}{\ms}
\newcommand{\Tleq}[1]{\transfersof{\transitions^{\leq #1}}}
\newcommand{\Tpow}[1]{\transfersof{\transitions^{#1}}}
\newcommand{\T}{\transfersof{\transitions}}
\newcommand*{\neutraltransferflows}{\transfersof{\emptysequence}}

\newcommand*{\leprod}{\leq_{\times}}
\newcommand*{\idealnorm}[1]{||#1||}
\newcommand*{\finitemapping}{\theta}
\newcommand*{\mapencoding}{\chi}
\newcommand*{\relevantvectors}{\mapencoding(\transferflows)}
\newcommand*{\indexof}{\mathsf{index}}
\newcommand*{\tftrans}{\mathsf{t}}
\newcommand*{\Tmin}{\minwqo{\T}}
\newcommand*{\maxweight}{2}

\newcommand{\DefOR}{\ensuremath{\hspace{0.6em}\big|\hspace{0.6em}}}

\newcommand{\ltlnext}{\mathsf{X}}
\newcommand{\ltlglobally}{\mathsf{G}}
\newcommand{\ltlfinally}{\mathsf{F}}
\newcommand{\ltluntil}{\mathbin \mathcal{U}}

\newcommand{\cheatstate}{q_\bot}

\newcommand{\U}{\ltluntil}

\newcommand{\Next}{\ltlnext}

\newcommand{\modelsLTL}{\ensuremath{\models}}
\newcommand{\modelsHLTL}{\ensuremath{\models}}
\newcommand{\modelsforall}{\models^{\forall}}
\newcommand{\modelsexists}{\models^{\exists}}

\renewcommand{\And}{\mathrel{\wedge}}
\newcommand{\Or}{\mathrel{\vee}}

\newcommand{\true}{\mathit{true}}
\newcommand{\false}{\mathit{false}}

\newcommand{\IntoP}{\mathrel{\rightharpoonup}}
\newcommand{\partTo}{\IntoP}
\newcommand{\V}{\mathcal{V}}

\newcommand{\Vars}{\ensuremath{\mathsf{Vars}}\xspace}

\newcommand{\Dom}{\ensuremath{\textit{Dom}}}

\newcommand{\machine}{\mathcal{M}}
\newcommand{\incr}{\mathsf{inc}}
\newcommand{\decr}{\mathsf{dec}}
\newcommand{\zero}{\mathsf{test_0}}
\newcommand{\halt}{\mathsf{halt}}
\newcommand{\halttrans}{\mathsf{h}}
\newcommand{\res}{\mathsf{res}} %
\newcommand{\interm}{\mathsf{int}}
\newcommand{\aux}{\mathsf{aux}}
\newcommand{\bad}{\mathcal{B}}

\newcommand{\pr}[2]{\mathsf{Pr}[#1 \models #2]}

\newcommand{\prodnet}[2]{N(#1, \ltlaut{#2})}
\newcommand{\protproj}[1]{\mathsf{pr}(#1)}
\newcommand{\configout}{\config_t}
\newcommand{\bstateout}{\bstate_t}
\newcommand{\Cout}{C_t}
\newcommand{\prodout}{c_t}

\newcommand{\stateone}{q_{\mathsf{one}}}
\newcommand{\staterest}{q_{\mathsf{rest}}}
\newcommand{\tgood}{\mathsf{good}}
\newcommand{\tbad}[1]{\mathsf{bad}_{#1}}
\newcommand{\goodconfigs}{\configs_{\mathsf{good}}}
\newcommand{\ltlgre}{\cSet_\bwinning}
\newcommand{\hyperneg}[1]{\mathsf{neg}(#1)}
\newcommand{\simplify}[2]{#1[#2]} 

\title{Temporal Hyperproperties for Population Protocols}

\author[1]{Nicolas Waldburger}
\author[2]{Chana Weil-Kennedy}
\author[2]{Pierre Ganty}
\author[2]{C\'esar S\'anchez}

\affil[1]{Université de Rennes, IRISA, INRIA, France}
\affil[2]{IMDEA Software Institute, Pozuelo de Alarc\'{o}n, Spain}

\date{}

\begin{document}

\maketitle

\begin{abstract}
	Hyperproperties are properties over sets of traces (or runs) of a system, as opposed to properties of just one trace.
	They were introduced in 2010 and have been much studied since, in particular via an extension of the temporal logic LTL called HyperLTL.
	Most verification efforts for HyperLTL are restricted to finite-state systems, usually defined as Kripke structures.
	In this paper we study hyperproperties for an important class of infinite-state systems.
	We consider population protocols, a popular distributed computing model in which arbitrarily many identical finite-state agents interact in pairs.
	Population protocols are a good candidate for studying hyperproperties because the main decidable verification problem, well-specification, is a hyperproperty.
	We first show that even for simple (monadic) formulas, HyperLTL verification for population protocols is undecidable.
	We then turn our attention to immediate observation population protocols, a simpler and well-studied subclass of population protocols.
	We show that verification of monadic HyperLTL formulas without the next operator is decidable in 2-\EXPSPACE, but that all extensions make the problem undecidable.
\end{abstract}

%

%
%
%
\section{Introduction}
\label{sec:introduction}

Hyperproperties are properties that allow to relate multiple traces (also called runs) of a system simultaneously~\cite{clarkson10hyperproperties}.
They generalize regular run properties to properties of sets of runs, and formalize a wide range of important properties such as information-flow security policies like noninterference~\cite{goguen82security,mclean96general} and observational determinism~\cite{zdancewic03observational}, consistency models in concurrent computing~\cite{bonakdarpour18monitoring}, and robustness models in cyber-physical systems~\cite{wang19statistical,bonakdarpour20model}.

HyperLTL~\cite{clarkson14temporal} was introduced as an extension of LTL (linear temporal logic) with quantification over runs which can then be related across time.
HyperLTL enjoys a decidable model-checking problem for finite-state systems, expressed as Kripke structures.
Other logics for hyperproperties were later introduced, like HyperCTL\(^*\)~\cite{finkbeiner13atemporal}, HyperQPTL~\cite{rabe16temporal,coenen19hierarchy}, and HyperPDL-\(\Delta\)~\cite{gutsfeld20propositional} which extend CTL\(^*\), QPTL~\cite{sistla87complementation}, and PDL~\cite{fischer79propositional} respectively.
These logics also enjoy decidable model-checking problems for finite-state systems.

Most algorithmic verification results for verifying hyperproperties of temporal logics are restricted to finite-state systems.
In the case of software verification, which is inherently infinite-state, the analysis of hyperproperties~\cite{beutner22software,farzan19automated,shemer19property,sousa16cartesian,unno21constraint} has been limited to the class of \(k\)-safety properties --- which only allow to establish the absence of a bad interaction between any \(k\) runs --- and do not extend to a temporal logic for hyperproperties.
A notable exception is~\cite{beutner22software}, but the logic used (OHyperLTL) is a simple asynchronous logic for hyperproperties and it requires restrictions on the underlying theories of the data used in the program.

In this paper we focus on the verification of HyperLTL for an important class of infinite-state systems.
We consider population protocols (PP) \cite{First-Pop-Prot}, an extensively studied (see e.g. \cite{AlistarhG18,ElsasserR18,Esparza21}) model of distributed computation in which anonymous finite-state agents interact pairwise to change their states, following a common protocol.
In a \emph{well-specified} PP, the agents compute a predicate: the input is the initial configuration of the agents’ states, and the agents interact in pairs to eventually reach a consensus opinion corresponding to the evaluation of the predicate (for \emph{any} number of agents).
Interactions are selected at random, which is modelled by considering only \emph{fair} runs.
LTL verification has been investigated for PPs in \cite{EGLM16}.
The authors consider LTL over actions, where formulas are evaluated over fair runs.
They show that it is decidable, given a PP and an LTL formula, to check if all fair runs from initial configurations of the protocol verify the formula.
Another related work on LTL verification for infinite-state systems is \cite{FortinMW17}, where the authors consider stuttering-invariant LTL verification over shared-memory pushdown systems.
\chana{stutt-inv then "defined" below}

We consider PPs because, though they are infinite-state, they enjoy several decidable problems.
In particular, the central verification problem checking whether a protocol is well-specified is decidable \cite{EGLM17journal} and has a hyperproperty ``flavor''.
A PP is well-specified if for every initial configuration \(\config_0\), every fair run starting in \(\config_0\) stabilizes to the same opinion.
A run stabilizes to an opinion \(b \in \set{0,1}\) if from some position onwards it visits no configuration with an agent whose opinion is \(1-b\).
With \(\initial\) the set of initial configurations and \(\fairrunsfrom{\config}\) the set of fair runs starting in \(\config\), well-specification can be expressed as: \[\forall \config_0 \in \initial, \ \fairrunsfrom{\config_0} \modelsHLTL \forall \run_1.
	\forall \run_2 \ldotp \textstyle{\bigvee_{b \in \set{0,1}}} (\ltlfinally \ltlglobally (\run_1 \text{ sees } b) \land \ltlfinally \ltlglobally(\run_2 \text{ sees } b))\]
where ``\(\run \text{ sees } b\)'' means that the run takes a transition
that puts agents into states with opinion \(b\).
Then \(\ltlfinally \ltlglobally (\run_i \text{ sees } b)\) ensures that \(\run_i\) converges to \(b\).

We show that for the general PP model, HyperLTL verification is already undecidable for simple (\emph{monadic}) formulas which can be decomposed into formulas referring to only one run each (Section \ref{sec:undecidability}).
We turn our attention to \emph{immediate observation population protocols} (IOPP), a subclass of PP \cite{Comp-Power-Pop-Prot}.
We show that HyperLTL verification over IOPP is a problem decidable in 2-\EXPSPACE when the formula is monadic and does not use the temporal operator \(\ltlnext\) (the formula is then \emph{stuttering-invariant}).
This result delineates the decidability frontier for verification in PP: non-monadic or non-stuttering-invariant HyperLTL verification over IOPP is undecidable (Section \ref{sec:iopp-verif}).
The decidability result for HyperLTL verification of IOPP is the most technical result of the paper.
In particular, the technical results of Section \ref{sec:structural-bounds} reason on the flow of agents in runs of an IOPP in conjunction with reading the transitions in a Rabin automaton.

%

%
%
%
%
%
%
%
%
%
%
%
%
%
%
%
%

%
%
%
\section{Preliminaries}

A \emph{finite multiset} over a finite set \(S\) is a mapping \(\ams \colon S \rightarrow \nats\) such that for each \(s\in S\), \(\ams(s)\) denotes the number of occurrences of element \(s\) in \(\ams\).
Given a set \(S\), \(\multisetsof{S}\) denotes the set of finite multisets over \(S\).
Given \(s \in S\), we denote by \(\vec{s}\) the multiset \(\ams\) such that \(\ams(s) = 1\) and \(\ams(s') = 0\) for all \(s' \ne s\).
Given \(\ams, \ams' \in \multisetsof{S}\), the multiset \(\ams + \ams'\) is defined by \((\ams + \ams')(s) = \ams(s) + \ams'(s)\) for all \(s \in S\).
We let \(\ams \leq \ams'\) when \(\ams(s) \leq \ams'(s)\) for all \(s \in S\).
When \(\ams' \leq \ams\), we let \(\ams - \ams'\) be the multiset such that \((\ams- \ams')(s) = \ams(s) - \ams'(s)\) for all \(s \in S\).
We call \(|\ams| =\sum_{s\in S} \ams(s)\) the \emph{size} of \(\ams\).
A set \(\cSet \subseteq \multisetsof{S}\) is ""Presburger"" if it can be written as a formula in Presburger arithmetic, \emph{i.e.}, in \(FO(\nats,+)\).

A \emph{strongly connected component} (SCC) in a graph is a non-empty maximal set of mutually reachable vertices.
A SCC is \emph{bottom} if no path leaves it.

\subsection{Population Protocols}

A \emph{population protocol} (PP) is a tuple \(\PP=(Q,\transitions,I)\) where \(Q\) is a finite set of \emph{states}, \(\transitions \subseteq Q^2 \times Q^2\) is a set of \emph{transitions} and \(I \subseteq Q\) is the set of \emph{initial states}.
A transition \(\atrans=\big( (q_1, q_2), (q_3,q_4) \big) \in \transitions\) is denoted \((q_1, q_2)\trans{\atrans} (q_3, q_4)\).
We let \(\size{\prot} \deff \size{\states} + \size{\transitions}\) denote the \emph{size} of \(\prot\).
A ""configuration"" of \(\PP\) is a multiset over \(Q\).
We denote by \(\configs \deff \set{\ams \in \multisetsof{Q} \mid \size{\ams}>2}\) the set of configurations; configurations must have at least \(2\) agents.
We note \(\initial \deff \set{\config \in \configs \mid \forall q \notin I, \config(q) = 0}\) the set of \emph{initial configurations}.
Given \(\config,\config' \in \configs\) and \((q_1, q_2)\trans{\atrans} (q_3, q_4)\in \transitions\), there is a ""step"" \(\config \trans{t} \config'\) if \(\config \geq \vec{q_1} + \vec{q_2}\) and \(\config' = \config - \vec{q_1} - \vec{q_2} + \vec{q_3} + \vec{q_4}\).
A transition \((q_1, q_2)\trans{} (q_3, q_4)\) is ""activated"" at \(\config\) if \(\config \ge \vec{q_1}+\vec{q_2}\), \emph{i.e.}, if there is an agent in \(q_1\) and an agent in \(q_2\) (or two agents in \(q_1\) if \(q_1=q_2\)).
Henceforth, we assume that for every \(q_1,q_2 \in Q\), there exist \(q_3,q_4 \in Q\) such that \((q_1, q_2)\trans{} (q_3, q_4) \in \transitions\), so that there is always an activated transition.
This can be done by adding self-loops \((q_1, q_2)\trans{}(q_1, q_2)\).

A \emph{finite run} is a sequence \(\config_0, \atrans_0, \config_1,\dots, \atrans_{k-1}, \config_k\) where \(\config_i \step{\atrans_{i}} \config_{i+1}\) for all \(i\leq k-1\); we say \(\atrans_{i}\) is ""fired"" at \(\config_i\).
We write \(\config \trans{*} \config'\) if there exists a finite run from \(\config\) to \(\config'\), and we say \(\config'\) is ""reachable"" from \(\config\).
Given \(\cSet \subseteq \configs\), let \(\poststar{\cSet}\) be the set of configurations reachable from \(\cSet\), \emph{i.e.}, \( \poststar{\cSet} \deff \set{ \config \mid \exists \config' \in \cSet \, .
	\, \config' \trans{*} \config}\).
Similarly, let \(\prestar{\cSet} \deff \set{\config \mid \exists \config' \in \cSet \, .
	\, \config \trans{*} \config'}\).

An \emph{infinite run} is an infinite sequence \(\run = \config_0,\atrans_0,\config_1,\atrans_1,\dots\) with \(\config_i \step{\atrans_{i}} \config_{i+1}\) for all \(i\in \nats\).
A configuration \(\config\) is \emph{visited} in \(\run\) where there is \(i\) such that \(\config_i = \config\); it is \emph{visited infinitely often} when there are infinitely many such \(i\).
Similarly, \(t \in \transitions\) is \emph{fired infinitely often} in \(\run\) where there are infinitely many \(i\) such that \(\atrans_i = t\).
A finite run \(\config'_0, \atrans_0', \config_1',\dots,\atrans_{k-1}', \config_k'\) \emph{appears infinitely often} in \(\run\) when there are infinitely many \(i\) such that \(\config_{i+j} = \config_j'\) for all \(j \in \nset{0}{k}\) and \(\atrans_{i+j} = \atrans'_j\) for all \(j \in \nset{0}{k-1}\).
Also, \(\run\) is \emph{strongly fair} when, for every finite run \(\run'\), by letting \(\config_0'\) the first configuration in \(\run'\), if \(\config_0'\) is visited infinitely often in \(\run\) then \(\run'\) appears infinitely often in \(\run\).
Given a configuration \(\config_0\), the set of strongly fair runs from \(\config_0\) is denoted \(\fairrunsfrom{\config_0}\).
Note that this notion of fairness differs from the one usually used for PPs.
We will discuss this choice in \cref{why-strong-fairness}.

\subsection{LTL and HyperLTL}
Linear temporal logic~\cite{pnueli77temporal} (LTL) extends propositional logic with modalities to relate different positions in a run, allowing to define temporal properties of systems.
HyperLTL~\cite{clarkson14temporal} is an extension of LTL for hyperproperties, with explicit quantification over runs.
We here define LTL and HyperLTL for population protocols.
Let \(\prot = (\states, \transitions, \initialstates)\) be a PP.
Our atomic propositions are the transitions of the run(s); we discuss this choice at the end of this section.

\paragraph*{LTL.}
The syntax of LTL over \(\prot\) is:
\begin{alignat*}{6}
	                     & \varphi   ::=   \atrans &                                                                                  & \DefOR \varphi \lor \varphi &  & \DefOR \neg \varphi &  &
	\DefOR \Next \varphi &                         & \DefOR \varphi \U\varphi \qquad \text{ where } \atrans \in \transitions\enspace.
\end{alignat*}
The operators \(\Next\) (next) and \(\U\) (until) are the temporal modalities.
We use the usual additional operators: \(\true = t \Or \neg t\), \(\false =\neg\true\), \(\varphi\And\varphi = \neg(\neg\varphi\Or\neg\varphi)\), \(\ltlfinally\varphi= \true\U\varphi\) and \(\ltlglobally\varphi = \neg\ltlfinally\neg\varphi\).
The \emph{size} \(\size{\varphi}\) of an LTL formula \(\varphi\) is the number of (temporal and Boolean) operators of \(\varphi\).
The semantics of LTL is defined over runs in the usual way (\emph{e.g.}, \cite{BK08}) over \(\transitions^\omega\).
An infinite run \(\run = \config_0,\atrans_0, \config_1, \atrans_1,\dots\) \emph{satisfies} an LTL formula \(\varphi\), denoted \(\run \modelsLTL \varphi\), when \(w \modelsLTL \varphi\) where \(w = \atrans_0 \atrans_1 \atrans_2 \dots \in \transitions^\omega\).
A configuration \(\config\) satisfies an LTL formula \(\varphi\), denoted \(\config \modelsLTL \varphi\), when \(\run \modelsLTL \varphi\) for all \(\run \in \fairrunsfrom{\config}\), \emph{i.e.}, when \emph{all} strongly fair runs starting from \(\config\) satisfy \(\varphi\).
\begin{toappendix}
	Formally, the semantics of LTL is defined as follows.
	For all \(\run\) and \(i \in \nats\):
	\[
		\begin{array}{@{}rl@{\hspace{2em}}c@{\hspace{2em}}l@{}}
			(\run,i) & \modelsLTL t                       & \text{iff } & \atrans \in \transat{\run}{i}                                               \\
			(\run,i) & \modelsLTL \varphi_1 \Or \varphi_2 & \text{iff } & (\run,i)\modelsLTL \varphi_1 \text{ or } (\run,i)\modelsLTL \varphi_2       \\
			(\run,i) & \modelsLTL \neg \varphi            & \text{iff } & (\run,i) \not\modelsLTL \varphi                                             \\
			(\run,i) & \modelsLTL \Next \varphi           & \text{iff } & (\run,i+1)\modelsLTL \varphi                                                \\
			(\run,i) & \modelsLTL \varphi_1 \U \varphi_2  & \text{iff } & \text{for some } j\geq 0\;\; (\run,i+j)\modelsLTL \varphi_2                 \\
			         &                                    &             & \hspace{1em}\text{and for all \(0\leq k<j\),} (\run,i+k)\modelsLTL\varphi_1
		\end{array}
	\]
	where \(\transat{\run}{i}\) denotes the \((i+1)\)-th transition in \(\run\).
	A run \(\run\) \emph{satisfies} a property \(\varphi\), denoted \(\run\modelsLTL\varphi\), whenever \((\run,0)\modelsLTL\varphi\).
\end{toappendix}

\paragraph*{HyperLTL.}
The syntax of HyperLTL over \(\prot\) is:
\begin{alignat*}{9}
	                                        & \psi   ::= \exists \run. \psi &                             & \DefOR \forall \run.\psi &                     & \DefOR \varphi &   &
	\qquad \qquad \varphi  ::= \atrans_\run &                               & \DefOR \varphi \lor \varphi &                          & \DefOR \neg \varphi &                &
	\DefOR \Next \varphi                    &                               & \DefOR \varphi \U\varphi
\end{alignat*}
where \(\atrans \in \transitions\) and \(\run\) is a \emph{run
	variable}.
Note that \(\varphi\) is an LTL formula with, as atomic propositions, the transitions of the run variables.
A HyperLTL formula \(\psi\) must additionally be well-formed: all appearing variables are quantified and no variable is quantified twice.
The size \(\size{\psi}\) of an HyperLTL formula \(\psi\) is the number of (temporal and Boolean) operators and quantifiers of \(\psi\).
HyperLTL formulas are interpreted over strongly fair runs starting from a configuration as follows: a configuration \(\config\) satisfies a HyperLTL formula \(\psi\), denoted \(\config \modelsHLTL \psi\), whenever \(\fairrunsfrom{\config} \modelsHLTL \psi\).
A formal definition of the semantics is given below.
Notice that, given a configuration \(\config\) and an LTL formula \(\varphi\), \(\config \modelsLTL \varphi\) if and only if \(\config \modelsHLTL \forall \run.
\varphi_\run\) where \(\varphi_\run\) is equal to \(\varphi\) where \(\atrans\) is replaced by \(\atrans_\run\) for all \(\atrans \in \transitions\).

\begin{example}
	Suppose that \(\transitions = \set{s,t}\).
	Let \(\psi \deff \forall \run_1.
	\exists \run_2. \ltlfinally \ltlglobally ((s_{\run_1} \land t_{\run_2}) \lor (t_{\run_1} \land s_{\run_2}))\).
	Given \(\config \in \configs\), we have that \(\config \models \psi\) if and only if, for every strongly fair run \(\run_1 \in \fairrunsfrom{\config}\) from \(\config\), there is a strongly fair run \(\run_2 \in \fairrunsfrom{\config}\) from \(\config\) such that, always after some point, \(\run_1\) fires \(s\) whenever \(\run_2\) fires \(t\) and vice versa.
\end{example}

A HyperLTL formula \(\psi:\aquantif_1\run_1\ldots \aquantif_k\run_k.
\varphi\) is
""monadic"" if \(\varphi\) has a \emph{decomposition} as a
Boolean combination of temporal formulas \(\varphi_1,\ldots,\varphi_n\),
each of which
refer to exactly one run variable. %
We assume that a monadic formula is always given by its decomposition, \emph{i.e.}, by giving \(\varphi_1\) to \(\varphi_n\) and the Boolean combination.

\begin{toappendix}
	\paragraph*{Semantics of HyperLTL.}
	To define the semantics of HyperLTL, we will have to consider HyperLTL formulas that are not well-formed because they have free run variables.
	We denote by \(\V\) the set of runs variables, which we assume to be infinite.
	Given a formula \(\psi\), we use \(\Vars(\psi)\) for the set of free run variables in \(\psi\) (those that appear in \(\psi\) but are not quantified in \(\psi\)).
	Given a set of runs \(R\), the semantics of a HyperLTL formula \(\psi\) is defined in terms of run assignments, which is a (partial) map from run variables to indexed runs \(\Pi:\Vars(\psi)\partTo R\).
	The run assignment with empty domain is denoted by \(\Pi_\emptyset\).
	We use \(\Dom(\Pi)\) for the subset of \(\Vars(\psi)\) for which \(\Pi\) is defined.
	Given a run assignment \(\Pi\), a run variable \(\run\) and a run \(\sigma\), we denote by \(\Pi[\run \mapsto \sigma]\) the assignment that coincides with \(\Pi\) for every run variable except for \(\run\), which is mapped to \(\sigma\).
	A pointed run assignment \((\Pi,i)\) consists of a run assignment \(\Pi\) with a pointer \(i\).
	The semantics of HyperLTL assign pointed run assignments to formulas
	as follows:
	\[
		\begin{array}{@{}rl@{\hspace{2em}}c@{\hspace{2em}}l@{}}
			(\Pi,0)                         &
			\modelsHLTL_R \exists \pi. \psi & \text{iff } & \text{for some
				\(\sigma\in{}
			R\), } (\Pi[\pi\mapsto \sigma],0)\modelsHLTL_R\psi             \\ (\Pi,0) & \modelsHLTL_R \forall \pi.
			\psi                            & \text{iff } & \text{for all
				   \(\sigma\in{}
			R\), } (\Pi[\pi\mapsto \sigma],0)\models_R\psi                 \\ (\Pi,0) &\modelsHLTL_R \varphi & \text{iff } & (\Pi,0)\modelsHLTL \varphi \\ (\Pi,i) &\modelsHLTL a_\pi & \text{iff } & (\sigma,i) \modelsLTL a, \text{ where } \sigma=\Pi(\pi) \\ (\Pi,i) &\modelsHLTL \varphi_1 \Or \varphi_2 & \text{iff } & (\Pi,i)\modelsHLTL \varphi_1 \text{ or } (\Pi,i)\modelsHLTL \varphi_2 \\ (\Pi,i) &\modelsHLTL \neg \varphi & \text{iff } & (\Pi,i) \not\modelsHLTL \varphi \\ (\Pi,i) &\modelsHLTL \Next \varphi & \text{iff } & (\Pi,i+1)\modelsHLTL \varphi\\ (\Pi,i) &\modelsHLTL \varphi_1 \U \varphi_2 & \text{iff } & \text{for some } j\geq 0\;\; (\Pi,i+j)\modelsHLTL \varphi_2\\ &&& \hspace{1em}\text{and for all \(0\leq k<j\),} (\Pi,i+k)\models\varphi_1\end{array} \] Note that quantifiers assign runs to run variables and set the pointer to the initial position \(0\).
	%
	We say that a set of runs \(R\) is a model of a HyperLTL formula \(\psi\), denoted \(R\models\psi\), whenever \(\Pi_\emptyset\models_R\psi\).
\end{toappendix}

\paragraph*{Verification Problems.}
Given a PP \(\PP=(Q,\transitions,I)\) and an LTL formula \(\varphi\) (resp.
\ a HyperLTL formula \(\psi\)), we denote \(\PP \modelsforall \varphi\) when \(\config_0 \modelsLTL \varphi\) (resp. \(\config_0 \modelsLTL \psi\)) for all \(\config_0 \in \initial\).
Dually, we let \(\PP \modelsexists \varphi\) (resp.
\(\PP \modelsexists \psi\)) when there is \(\config_0 \in \initial\) such that \(\config_0 \modelsLTL \varphi\) (resp. \(\config_0 \modelsHLTL \psi\)).

The ""LTL verification problem for population protocols"" consists on determining, given \(\prot\) and an LTL formula \(\varphi\), whether \(\prot \modelsforall \varphi\), \emph{i.e.}, whether all strongly fair runs from all initial configurations satisfy \(\varphi\).
We also consider a variant problem, the \emph{existential LTL verification problem}, that asks whether \(\prot \modelsexists \varphi\), \emph{i.e.}, whether there is an initial configuration from which all strongly fair runs satisfy \(\varphi\).
Given a HyperLTL formula \(\psi\), the ""HyperLTL verification problem for population protocols"" consists on determining whether \(\prot \modelsforall \psi\); again, the existential variant consists in asking whether \(\prot \modelsexists \psi\).

\begin{example}
	A PP \(\PP=(Q,\transitions,I)\) equipped with an ""opinion"" function \(O\colon Q \rightarrow \set{0,1}\) is ""well-specified"" if for every \(\config_0 \in\initial\), every run in \(\fairrunsfrom{\config_0}\) eventually visits only configurations where either all agents are in states \(O^{-1}(0)\) or all agents are in states \(O^{-1}(1)\).
	Let \(\transitions_b\) be the set of transitions \((q_1, q_2)\trans{} (q_3, q_4)\in\transitions\) such that \(O(q_3)=O(q_4)=b\).
	Well-specification of \(\PP=(Q,\transitions,I)\) with opinion function \(O\) corresponds to the HyperLTL verification problem over a monadic formula: \[ \PP \modelsforall \forall \run_1, \run_2 \ldotp \textstyle{\bigvee_{b \in \set{0,1}}} \ltlfinally \ltlglobally (\bigvee_{t\in\transitions_b} t_{\run_1}) \land \ltlfinally \ltlglobally(\bigvee_{t\in\transitions_b} t_{\run_2})\enspace .
	\]
\end{example}

\paragraph*{LTL over transitions and LTL over states.}
Our LTL formulas are \emph{over transitions}, \emph{i.e.}, their atomic propositions are the transitions of the run.
In \cite[Theorems 9 and 10]{EGLM16}, the LTL verification problem defined above is proven to be decidable, although as hard as reachability for Petri nets and therefore Ackermann-complete \cite{Leroux2022,Czerwinski2022}.
The authors of \cite{EGLM16} also show that \emph{LTL over states}, where the atomic predicates indicate whether or not a state contains at least one agent, is undecidable.
A slight difference between their model and ours is that their initial configurations are given by a Presburger set; however, their undecidability proof, which relies on \(2\)-counter machines, can easily be translated to our setting.
In the rest of the paper we consider only (Hyper)LTL over transitions.

\begin{propositionrep}%
	\label{prop:presence-undecidable}
	The LTL over states verification problem for PP is undecidable.
\end{propositionrep}
\begin{appendixproof}
	The proof is a small modification of the proof of Proposition 13 of \cite{EGLM16}, which shows that verification of LTL over states (which they call LTL over presence) with initial configurations \emph{encoded into a Presburger formula} is undecidable.
	In LTL over states, the atomic propositions are of the form \(q_{\ge 1}\), for \(q\) a state of a population protocol.
	Given a population protocol and a run \(\sigma\), we have \((\sigma,i) \modelsLTL q_{\ge 1}\) when the \(i\)-th configuration of \(\sigma\) contains one or more agents in \(q\).

	Proposition 13 of \cite{EGLM16} (which uses Lemma 11 of \cite{EGLM16}) shows undecidability by a reduction from the halting problem for counter machines.
	Given a counter machine \(\machine\), it constructs a population protocol \(\PP\) which simulates the counter machine, coupled with an LTL formula \(\varphi\) which is satisfied when the simulation is correct (it does not ``cheat'') and the \(\halt\) state is reached.
	For every instruction \(l\) of \(\machine\), there is a state \(l\) in \(\PP\).
	We count \(\halt\) as an instruction.
	An agent in \(l\) in \(\PP\) intuitively means that the simulated counter machine is on instruction \(l\).
	The initial configurations of \(\PP\) put exactly one agent in the first instruction \(l_1\), exactly one agent in a dummy state \(D\), and an unbounded number of agents in a reservoir state \(Store\).

	With our definition, we cannot express that initial configurations put ``exactly one agent'' somewhere.
	Instead, we modify \(\PP\) by having as unique initial state \(Store\)
	and adding the following transitions:
	\begin{align*}
		 & (Store, Store)\trans{inst_1} (l_1, Store) \text{   and   } (Store, Store)\trans{dummy} (D, Store)
	\end{align*}
	We (conjunctively) add \(inst_1 \land \ltlnext dummy \land \ltlnext\ltlnext (\neg \ltlfinally inst_1 \land \neg \ltlfinally dummy)\) to the formula \(\varphi\) which enforces that the first two transitions put an agent in \(l_1\) and \(D\) respectively, and then never again, thus enforcing a correct simulation.
	\chana{commented out is version without \(\ltlnext\) and with fairness}
\end{appendixproof}

\subsection{Rabin Automata and LTL}
Let \(\Sigma\) be a finite set.
The set of finite words (resp.
\ infinite words) over \(\Sigma\) is denoted \(\Sigma^*\) (resp. \(\Sigma^\omega\)).
A ""deterministic Rabin automaton"" over \(\Sigma\) is a tuple \(\buchiaut = (\bstates, \btransitions, \binitstate, \bwinning)\), where \(\bstates\) is a finite set of states, \(\binitstate \in \bstates\) is the initial state, \(\btransitions : \bstates \times \Sigma \rightarrow \bstates\) is the transition function and \(\bwinning \subseteq 2^\bstates \times 2^\bstates\) is a finite set of ""Rabin pairs"".
An infinite word \(w \in \Sigma^\omega\) is ""accepted"" if there exists \((F, G)\in \bwinning\) such that the run of \(\buchiaut\) reading \(w\) visits \(F\) finitely often and \(G\) infinitely often.

\begin{theorem}[\cite{EKS18}]
	\label{ltl-rabin}
	Given \(\Sigma\) a finite set and \(\varphi\) an "LTL formula" over \(\Sigma\), one can compute, in time doubly-exponential in \(\size{\varphi}\), a "deterministic Rabin automaton" \(\ltlaut{\varphi}\) over \(\Sigma\), of doubly-exponential size, that recognizes (the language of) \(\varphi\).
\end{theorem}

\subsection{Why Strong Fairness?}
\label{why-strong-fairness}

Usually, fairness in population protocols is either of the form ``all configurations reachable infinitely often are reached infinitely often'' \cite{Comp-Power-Pop-Prot,EGLM17journal}, or ``all steps possible infinitely often are taken infinitely often'' \cite{EGLM16}.
Our notion of fairness, dubbed strong fairness, is more restrictive.
A sanity check is that a (reasonable) stochastic scheduler yields a strongly fair run with probability \(1\).
This alone does not justify using a new notion of fairness different from the literature and in particular from the prior work on LTL verification \cite{EGLM16}.
The authors motivate their choice of fairness by claiming that there is a fair run satisfying an LTL formula \(\varphi\) if and only if, under a stochastic scheduler, \(\varphi\) is satisfied with non-zero probability \cite[Proposition~7]{EGLM16}.
However, we show that this claim is incorrect.

\begin{example}
\label{ex:fair-vs-strong-fair}
The intuition is that a (not strongly) fair run may exhibit infinite regular patterns.
Consider three configurations \(\config_1,\config_2,\config_3\) and three transitions \(a,b,c\) such that \(\config_1 \trans{a} \config_2\), \(\config_2 \trans{b} \config_1\), \(\config_1 \trans{c} \config_3\) and \(\config_3 \trans{d} \config_1\), and these are the only steps possible from each of the configurations.
Consider \(\varphi = \neg \ltlfinally(a \land (\ltlnext b) \land (\ltlnext^2 a) \land (\ltlnext^3 b))\), which expresses that the sequence of transitions \(abab\) does not appear.
Under a stochastic scheduler, \(\varphi\) is satisfied with probability \(0\) from \(\config_1\).
However, the run which repeats sequence \(abcd\) satisfies \(\varphi\), and it is fair.
\end{example}

\begin{toappendix}
	\paragraph*{Mistake related to fairness in \cite{EGLM16}.}
	We explain the mistake in Proposition 7 from \cite{EGLM16},
	using their notation.
	In \cite{EGLM16}, an infinite run \(\run\) is \emph{fair} if for every \(\config\) appearing infinitely often in \(\run\), for every possible step \(\config \trans{t} \config'\), step \(\config \trans{t} \config'\) appears infinitely often in \(\run\).
	The product system \(N(\mathcal{A}, \mathcal{R}_\varphi)\) is composed of the population protocol \(\mathcal{A}\) put side by side with a deterministic Rabin automaton \(\mathcal{R}_\varphi\) that recognizes the same language as \(\varphi\).
	From a run \(\run\) of the population protocol \(\mathcal{A}\), one can easily build a run \(\run'\) of \(N(\mathcal{A}, \mathcal{R}_\varphi)\) whose projection on \(\mathcal{A}\) is \(\run\).
	The run \(\run'\) simply corresponds to performing \(\run\) in \(\mathcal{A}\) while the Rabin automaton moves accordingly to the sequence of transitions of \(\run\).
	However, the fact that \(\run\) is fair in \(\mathcal{A}\) does not imply that \(\run'\) is fair in \(N(\mathcal{A}, \mathcal{R}_\varphi)\).
	Indeed, it could be that a configuration \(t\) is activated from a configuration \(C\) of \(\mathcal{A}\), but that \(\run'\) visits infinitely often both \((C+\mathbf{q_1})\) and \((C + \mathbf{q_2})\) and that \(t\) is only fired from \((C+\mathbf{q_1})\) but never from \((C + \mathbf{q_2})\).
	For this reason, it does not hold that, when \(\run\) is a fair run, \(\run'\) is eventually in a bottom SCC of \(N(\mathcal{A}, \mathcal{R}_\varphi)\), nor that \(\run'\) visits infinitely often all configurations in the bottom SCC assuming that it ends up in one.
\end{toappendix}

The fair run described in \cref{ex:fair-vs-strong-fair} is not strongly fair.
We show that strong fairness does in fact allow the desired equivalence with stochastic schedulers.
As in \cite{EGLM16}, fix a stochastic scheduler, assumed to be memoryless and guaranteeing non-zero probability for every activated transition; \(\pr{\config}{\varphi}\) denotes the probability that a run from \(\config\) satisfies \(\varphi\).
\begin{propositionrep}
	\label{strong-fairness-stochastic-scheduler}
	Given an LTL formula \(\varphi\) and \(\config_0 \in \configs\), \(\pr{\config_0}{\varphi} = 1\) if and only if, for all \(\run \in \fairrunsfrom{\config_0}\), \(\run \models \varphi\).
\end{propositionrep}

\begin{appendixproof}
	We follow the same proof strategy as \cite[Proposition 7]{EGLM16}, but we circumvent the issue with fairness by relying on strong fairness instead.
	We build a Rabin automaton \(\ltlaut{\varphi} = (\bstates, \btransitions, \binitstate, \bwinning)\) that recognizes \(\varphi\) (see \cref{ltl-rabin}).
	We build the (conservative) Petri net \(\prodnet{\prot}{\varphi}\) obtained by making \(\ltlaut{\varphi}\) read the transitions performed in \(\prot\).
	We only consider configurations of \(\prodnet{\prot}{\varphi}\) with one agent in \(\ltlaut{\varphi}\), which are denoted \((\config, \bstate)\) with \(\config \in \multisetsof{\states}\) and \(\bstate \in \bstates\).
	Given a run \(\run\) of \(\prodnet{\prot}{\varphi}\), we denote by \(\protproj{\run}\) the corresponding run of \(\prot\).
	We call a run \(\run\) of \(\prodnet{\prot}{\varphi}\) \emph{protocol-fair} when \(\protproj{\run}\) is strongly fair.
	Fix \(\config_0 \in \multisetsof{\states}\) and let \(c_0 \deff (\config_0, \binitstate)\).
	Let \(\mathcal{G}\) be the graph of configurations reachable from \(c_0\) in \(\prodnet{\prot}{\varphi}\), with an edge between \(c_1\) and \(c_2\) when \(c_1 \step{} c_2\) in \(\prodnet{\prot}{\varphi}\).
	We call an SCC \(S\) of \(\mathcal{G}\) \emph{winning} when there is a winning pair \((F,G) \in \bwinning\) such that \(S\) contains some configuration with Rabin state in \(G\) but none with Rabin state in \(F\).
	By \cite[Proposition 6]{EGLM16}, we have \(\pr{\config_0}{\varphi}= 1\) if and only if all bottom SCC reachable from \(c_0\) are winning.

	Suppose there is a bottom SCC \(S\) reachable from \(c_0\) that is not winning.
	Then there is a protocol-fair run \(\run\) of \(\prodnet{\prot}{\varphi}\) that does not satisfy the Rabin winning condition and therefore does not satisfy \(\varphi\): it suffices to consider a run that goes to \(S\), then chooses transitions in a randomized fashion (uniformly at random among all possible transitions, regardless of the past).
	Almost-surely, the run obtained is protocol-fair and visits all configurations in \(S\) infinitely often; this proves the existence of the desired run.

	Now suppose that all bottom SCC reachable from \(c_0\) are winning.
	We show that for all \(\run \in \fairrunsfrom{\config_0}\), \(\run \models \varphi\), by proving that every protocol-fair run ends in a bottom SCC and visits all configurations in this bottom SCC infinitely often.
	Let \(\run\) be a protocol-fair run of \(\prodnet{\prot}{\varphi}\); let \(S\) denote the SCC of \(\mathcal{G}\) visited infinitely often in \(\run\).
	Suppose by contradiction that \(S\) is not bottom.
	There is \(\atrans \in \transitions\) and \((\configout,\bstateout) \in S\) from which firing \(\atrans\) takes us out of \(S\).
	Let \(\Cout \deff S \cap (\set{\configout} \times \bstates)\) and let \(\Cout' \subseteq \Cout\) be the set of such configurations from which firing \(\atrans\) yields a configuration in \(S\); we denote \(\Cout' = \set{(\configout, \bstate_1),\dots, (\configout,\bstate_m)}\) with \(m = \size{\Cout'}\).
	Whenever \(\atrans\) is fired from \(\configout\) in \(\protproj{\run}\), it is fired in \(\run\) from a configuration in \(\Cout'\), as \(\run\) does not leave \(S\).
	By strong fairness, \(\run\) fires \(\atrans\) from \(\configout\) infinitely often, so that \(\run\) fires \(\atrans\) infinitely often from some configuration in \(\Cout'\).
	This implies in particular that \(\Cout' \ne \emptyset\) and that \(m \geq 1\).

	By definition of \(\Cout'\), for all \(i \in \nset{1}{m}\), there exists \(c_{t,i} \in S\) such that \((\configout, \bstate_i) \step{t} c_{t,i}\).
	Because \(S\) is strongly connected, there is \(w_i \in \transitions^*\) such that \(c_{t,i} \step{w_i} (\configout, \bstateout)\).
	In \(\prot\), we have \(\configout \step{t \concat w_i} \configout\) for all \(i\).
	We build words \(\sigma_i \in \transitions^*\) for each \(i \in \nset{0}{m}\) by induction on \(i\) as follows.
	First, let \(\sigma_0 \deff \epsilon\).
	Suppose that \(\sigma_i\) is constructed.
	Let \(c_{i+1}\) denote the configuration obtained by firing \(\sigma_i\) from \((\configout, \bstate_{i+1})\); \(c_{i+1}\) is in \(\set{\configout} \times \bstates\).
	If \(c_{i+1} \notin \Cout'\) then we let \(\sigma_{i+1} \deff \sigma_i\).
	Otherwise, there is \(j\) such that \(c_{i+1} = (\configout, \bstate_j)\), and we let \(\sigma_{i+1} \deff \sigma_i \concat t \concat w_j\).
	We finally let \(\sigma \deff \sigma_m \concat t\).

	We have that, from each configuration in \(\Cout'\), firing \(\sigma\) takes us out of \(S\).
	Indeed, let \((\config, \bstate_i) \in \Cout'\).
	Consider the configuration \(c_{i}\) obtained by firing \(\sigma_{i-1}\) from \((\configout, \bstate_i)\).
	If \(c_{i} \notin S\) then we are done; if \(c_{i} \notin \Cout'\) then firing \(\sigma_{i-1} \concat t\) from \((\configout, \bstate_i)\) makes us leave \(S\).
	If \(c_{i} = (\configout, \bstate_j) \in \Cout'\) then firing \(\sigma_i = \sigma_{i-1} \concat t \concat w_j\) from \((\configout, \bstate_i)\) takes us to \((\configout,\bstateout)\), from where firing \(t\) makes us leave \(S\).

	The sequence of transitions \(\sigma\) is available infinitely often from \(\configout\) in \(\protproj{\run}\) and thus fired infinitely often by strong fairness.
	Therefore it is fired infinitely often from \(\Cout'\) in \(\run\).
	However, firing \(\sigma\) from \(\Cout'\) makes us leave \(S\) and \(\Cout' \subseteq S\), a contradiction.
	We have proven that any protocol-fair \(\run\) visits a non-bottom SCC finitely many times, which implies that it ends in a bottom SCC \(S\).
	We now prove that such a run \(\run\) visits all configurations in \(S\) infinitely often.
	It suffices to prove that, if we have \(\prodout = (\configout, \bstateout),\prodout' \in S\) and \(t \in \transitions\) such that \(\prodout\) is visited infinitely often and \(\prodout \step{t} \prodout'\), then \(\prodout'\) is visited infinitely often by \(\run\).
	The proof is very similar to the one above, but simpler because \(\Cout = \Cout'\).
	Indeed, all configurations \(c\) in \(\Cout =(\set{\configout} \times \bstates) \cap S\) are such that, when firing \(t\) from \(c\), one remains in \(S\) so that there is a path to \(\prodout\).
	We therefore build the sequence of transitions \(\sigma\) as above, except that the case \(c_{i+1} \notin \Cout'\) cannot occur.
	With the same proof technique, we prove that firing \(\sigma\) from any configuration in \(\Cout\) makes us visit \(\prodout'\).
	With the strong fairness of \(\protproj{\run}\), this allows us to conclude that \(\prodout'\) is visited infinitely often.
\end{appendixproof}

This therefore justifies our choice to consider strong fairness for LTL verification.
In particular, all results from \cite{EGLM16} hold if strong fairness is considered instead of the usual fairness.
An alternative to strong fairness for (non-Hyper)LTL verification would be to work directly with a stochastic scheduler.
However, HyperLTL requires quantification over a subset of the set of runs; we make the choice to consider, for this subset, the set of strongly fair runs.

%
%

%
%
%
%
%
%
%
\section{Undecidability of HyperLTL}
\label{sec:undecidability}

One can show that verification of HyperLTL over transitions is undecidable for PP, using a proof with counter machines similar to the one for undecidability of LTL over states \cite{EGLM16}.
Intuitively, HyperLTL can be used to express whether a transition is activated at some point in the run, and hence encode zero-tests\footnote{ See the proof of \cref{undec-hyperltl-iopp} for an illustration of this.
}.
We show an even stronger undecidability result: verification of monadic HyperLTL formulas over two runs using only \(\ltlfinally \ltlglobally \) as temporal operator is undecidable.
\begin{theorem}
	\label{monadic-undec-pp}
	Verification of monadic HyperLTL for PP is undecidable.
	If fact, it is already undecidable for formulas of the form: \[\forall \atrace_1.
		\, \exists \atrace_2. \, \neg (\ltlfinally \ltlglobally \,a_{\atrace_1}) \lor (\ltlfinally \ltlglobally \,b_{\atrace_2})  \qquad \text{where } a, b \in \transitions\enspace. \]
\end{theorem}

This verification problem asks whether, for all \(\config_0 \in \initialconfigs\), for all \(\run_1 \in \fairrunsfrom{\config_0}\), there is \(\run_2 \in \fairrunsfrom{\config_0}\) such that if \(\run_1\) fires \(a\) infinitely often then \(\run_2\) fires \(b\) infinitely often.
We first observe that the \(\forall\)-\(\exists\) sequence of quantifiers is reminiscent of inclusion problems.
Since the population protocol model is close to Petri nets, it is natural to look for undecidable inclusion-like problems for that model.
Indeed, undecidability was shown multiple times \cite{baker1973rabin,Hack76} for the problem asking whether the set of reachable markings of a Petri net is included in the set of reachable marking of another Petri net with equally many places.
We call this problem the \emph{reachability set inclusion problem}.
Our attempts at reducing the reachability set inclusion problem to the above problem faced a major obstacle: Petri nets allow the creation/destruction of tokens while in PPs the number of agents remains the same.
We sidestepped this obstacle by looking at a particular proof of undecidability for the reachability set inclusion problem which leverages Hilbert's Tenth Problem (shown to be undecidable by Matijasevic in the seventies).
We thus obtain a reduction from Hilbert's Tenth Problem to the above problem for PPs.
Our reduction uses PPs to ``compute'' the value of polynomials while keeping the number of agents constant during the computation.

The statement of the variant of Hilbert's Tenth Problem which we use is:

\begin{proposition}[\cite{Hack76}]
	\label{hilbert-variant-undec}
	The following problem is undecidable: \\ \textbf{Input}: two polynomials \(\apoly_1(\varx_1, \dots, \varx_r), \apoly_2(\varx_1, \dots, \varx_r)\) with natural coefficients \\ \textbf{Question}: Does it hold that, for all \(x_1, \dots, x_r \in \nats\), \(\apoly_1(x_1, \dots, x_r) \leq \apoly_2(x_1, \dots, x_r)\)?
\end{proposition}

\begin{toappendix}
	We will proceed by reduction from the problem in \cref{hilbert-variant-undec}.
	Let \(\apoly(\varx_1, \dots, \varx_r)\) be a multivariate polynomial with positive integer coefficient; let \(\maxexp\) denote the degree of \(\apoly\).
	Given a population protocol \(\prot\) with states including \(\set{\startstate, \statex{1}, \dots, \statex{r}, \reservoirstate, \statey}\) and a special transition \(\okaction\), \(\prot\) weakly computes \(\apoly\) if, for every initial configuration \(\config_0\), there exists a run \(\trace \in \fairrunsfrom{\config_0}\) firing \(\okaction\) infinitely often if and only if:
	\begin{itemize}[noitemsep,topsep=0pt,parsep=0pt,partopsep=0pt]
		\item \(\config_0(\startstate) = 1\),
		\item \(\config_0(\reservoirstate) \geq 1 + \sum_i (\maxexp-1)\; \config_0(\statex{i})\),
		\item \(\config_0(\statey) \leq \apoly(\config_0(\statex{1}), \dots, \config_0(\statex{r}))\).
	\end{itemize}

	\begin{figure}[h]
		\centering
		\includegraphics[page=1,scale=.7]{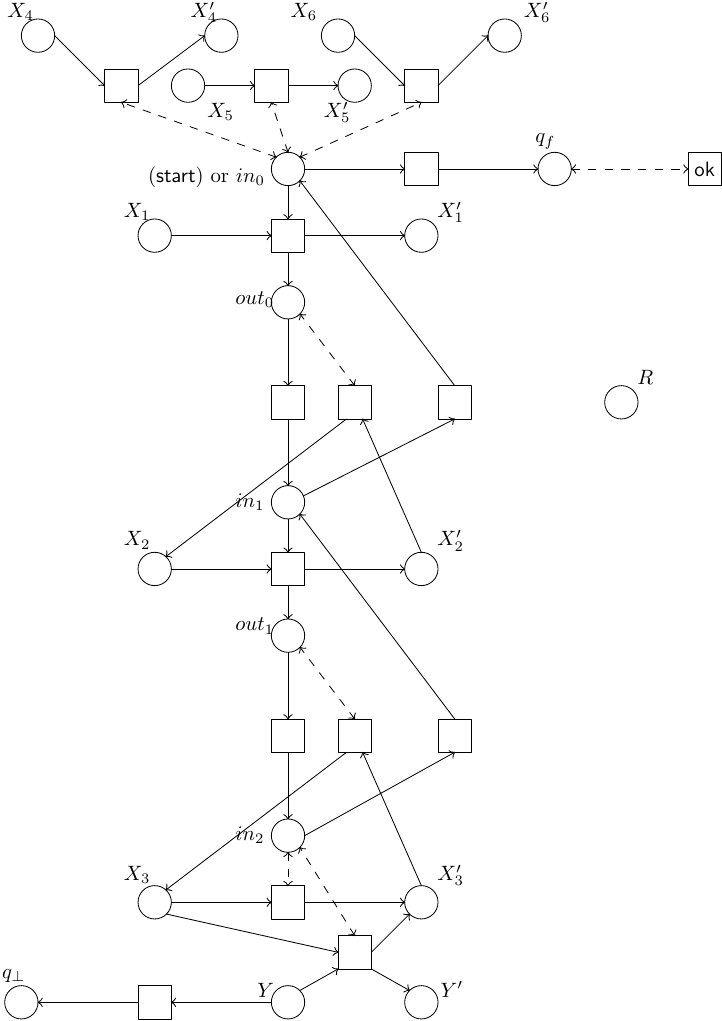}
		\caption{Gadget for the monomial \(X_1 X_2 X_3\) for variables \(\{X_1,X_2,X_3,X_4,X_5,X_6\}\).}
		\label{fig:multiplier}
	\end{figure}

	We will transform the problem of \cref{hilbert-variant-undec} to make it easier to encode with population protocols.
	First, we assume, without loss of generality, that constant terms in \(\apoly_1\) and \(\apoly_2\) are non-negative.
	With that in mind, let us now turn to the encoding of multivariate polynomials.
	Given a multivariate polynomial, the first transformation replaces multiple occurrences of the same variable in each monomial by assigning to each repeated occurrence its own variable.
	For instance, the monomial \(123 x^2 y^3 z\) is replaced by \(123 x^{(0)} x^{(1)} y^{(0)} y^{(1)} y^{(2)} z^{(0)}\).
	When the variables \(x^{(0)}\) and \(x^{(1)}\) are given the same value as \(x\), \(y^{(0)} y^{(1)} y^{(2)}\) are given the same value as \(y\) and \(z^{(0)}\) the same value as \(z\), we find that the polynomials before and after this transformation evaluate to the same value.

	\begin{figure}[t]
		\centering
		\includegraphics[page=2]{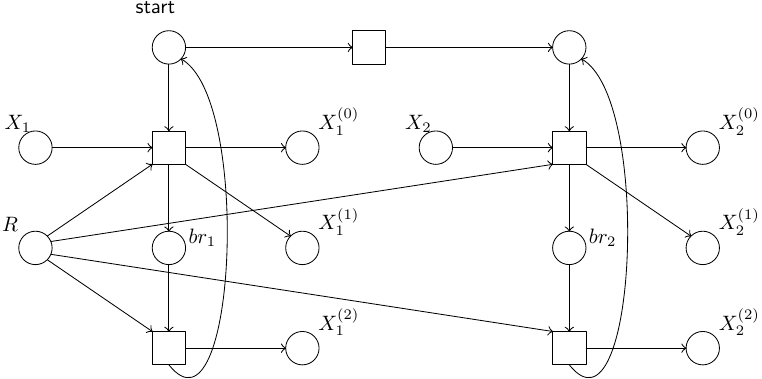}
		\caption{Gadget initializing three copies of variables \(X_1\) and \(X_2\).}
		\label{fig:initializer}
	\end{figure}

	The next transformation gets rid of coefficient using their unary expansion.
	Such transformation replaces \( 5 xyz \) by \( xyz + xyz + xyz + xyz + xyz \).
	Clearly the polynomials before and after the second transformation evaluate to the same values when fed the same arguments.

	Let us assume that our polynomials have been transformed as explained above and let us turn to the encoding of a monomial using population protocols.
	We use Petri net-inspired notation in our figures: circles are states, squares are transitions and a dashed line between a transition and a state is shorthand for having an arrow in both directions.
	For instance, the top left transition in \figurename~\ref{fig:multiplier} adjacent to states \(X_4, X'_4\) and \(in_0\) corresponds to the population protocol transition \( (X_4,in_0) \trans{} (X'_4,in_0) \).
	To keep \figurename~\ref{fig:multiplier} and \figurename~\ref{fig:initializer} lightweight, we also assume that transitions which have a single input arrow and single output arrow have an (undrawn) dashed line with state \(R\).
	For instance, the \textsf{ok} transition in \figurename~\ref{fig:multiplier} corresponds to the population protocol transition \( (q_f,R) \trans{\textsf{ok}} (q_f,R) \).
	Finally, transitions with more than two input and output arrows, such as the top left transition in \figurename~\ref{fig:initializer}, are actually encoded using a gadget of population protocol transitions following the construction explained by Blondin et al.~\cite[Lemma~3]{BlondinEJ18} \chana{I cannot get this cite to compile.
		..}
	to encode \(k\)-way transitions into the standard \(2\)-way transitions.

	We return to the encoding of monomials: we use the nesting of loops to weakly compute monomials such that \(d\) nested loops compute a monomial of degree \(d\).
	The idea is that the innermost loop iterates at most as many times as the product of the values in the monomial.

	\begin{lemma}
		\label{prot-weakly-computes-product}
		Let \(r \geq 1\), \(\apoly(\varx_1, \dots, \varx_r) := \prod_{i\in R} \varx_i\) where \(R\subseteq \{1,\ldots,r\}\).
		There is a protocol that weakly computes \(\apoly\).
	\end{lemma}
	\begin{proof}
		\(\prot\) has a special state \(\cheatstate\) that is attracting, \emph{i.e.}, if an agent is in \(\cheatstate\) then, by fairness, all agents will eventually come to \(\cheatstate\) and the run does not fire \(\okaction\) infinitely often.
		\(\prot\) has a \emph{leader part} in which there should only be one agent; if two agents lie in this part of the protocol, they may interact and be sent to \(\cheatstate\).
		The state \(\startstate\) is in the leader part, as well as the state \(\okstate\), from which the transition \(\okaction\) is fired at will.
		Therefore, in order to fire \(\okaction\) infinitely often with non-zero probability, exactly one agent must lie in the leader part.
		We henceforth assume that there is exactly one agent in this part, acting as a \emph{leader}.
		We explain the rest of the protocol by means of an example that is depicted in \figurename~\ref{fig:multiplier}.
		For the protocol depicted in the figure there is a run that moves exactly \(\gamma_0(X_1) \times \gamma_0(X_2) \times \gamma_0(X_3)\) agents from \(Y\) to \(Y'\) while also moving \(\gamma_0(X_i)\) agents to \(X'_i\) for \(i=1,2,3,4,5,6\).
		This run starts from a configuration \(\gamma_0\) that has no agents in the primed variables (i.e. \(X'_i\), \(i=1,\ldots,6\) and \(Y'\)), no agents in \(out_0\), \(out_1\), \(in_1\), \(in_2\) and \(q_f\), and exactly \(1\) agent (the leader) in \(in_0\) (which we also refer to as \(\startstate\)).
		The leader part in this example is given by the \(in_i\) and \(out_j\) (\(i=0,1,2\), \(j=0,1\)) states together with \(q_f\).
		Observe that no run moves more than \(\gamma_0(X_i)\) agents to \(X'_i\) for \(i=1,\ldots,6\).
		Also no run moves more than \(\gamma_0(X_1) \times \gamma_0(X_2) \times \gamma_0(X_3)\) agents from \(Y\) to \(Y'\).

		It is an easy exercise to generalize the above construction to a product with more than 3 variables.

		In the figure, when the leader is in the state \(in_0\), a transition moves it to state \(q_f\) enabling transition \(\okaction\) to fire at will.
		Incidentally, if, in a run, no agent enters \(\cheatstate\) then we find that, with probability \(1\), we end up with no agent in the leader part.
		Moreover, state \(Y\) has a transition to \(\cheatstate\) that is enabled at every moment for agents in \(Y\), so that fair runs where no agent enters \(\cheatstate\) must eventually empty \(Y\), which is possible if and only if \(\config_0(\statey) \leq \apoly(\config_0(\statex{1}), \dots, \config_0(\statex{r}))\).

	\end{proof}

	\begin{lemma}
		\label{prot-weakly-computes-poly}
		Given a multivariate polynomial \(\apoly\), one can compute a population protocol that weakly computes \(\apoly\).
	\end{lemma}
	\begin{proof}
		Let \(\apoly(\varx_1, \dots, \varx_r)\) be a multivariate polynomial with positive integer coefficients and degree \(\maxexp\).
		As explained above, we may assume that all monomials in \(\apoly\) have coefficient \(1\).
		As given in \cref{prot-weakly-computes-product}, the protocol has a leader part in which only one agent evolves; if several agents are in the leader part, then eventually some of them will be sent to \(\cheatstate\).

		The protocol is composed of layers, each of which encodes a monomial using the construction from \cref{prot-weakly-computes-product}.
		Again, as explained above, we assume that each variable contributes at most linearly to each monomial.
		The agents in the copies of the variables will be transmitted from layer to layer.
		This is the role played by the primed variables in \figurename~\ref{fig:multiplier}.
		A final layer encodes the constant term \(c\) by moving as many as \(c\) agents from \(Y\) to \(Y'\).

		We now explain how enough agents are moved into the copies of the variables in the first layer.
		To do so, the leader, which initially is in the \(\startstate\) state, starts by taking tokens from \(\reservoirstate\) to fill up all copies of the first layer with the right number of agents.
		\figurename~\ref{fig:initializer} depicts an example of a gadget filling the first layer by using \(\reservoirstate\) to create three copies of \(X_1\) and \(X_2\).
		For each \(i\), there are exactly \(x_i := \config_0(\statex{i})\) agents in state \(\statex{i}\); these \(x_i\) agents may be used in the first layer, therefore the leader must fill up \(\maxexp\) other states with exactly \(x_i\) agents.
		If \(\reservoirstate\) does not have enough agents to do so, \emph{i.e.}, if \(\config_0(\reservoirstate) < \sum_i (\maxexp-1) x_i\), then the leader might get stuck in \(br_1\) or \(br_2\) and \(\okaction\) can never be fired.
		Also there are transitions moving agents to \(q_\bot\) from any pair of agents in \(X_i\) and \(X_j\).
		This will guarantee that no agent stays forever in a state \(X_i\) for \(i=1,\ldots,r\), hence that the gadget has performed the copies as specified.

		Once the leader has been through every layer, at most \(\apoly(\config_0(\statex{1}), \dots, \config_0(\statex{r}))\) agents have be moved from \(\statey\) to \(\statey'\), and there is a run that indeed moves this many agents from \(\statey\).
		There is, as in \cref{prot-weakly-computes-product}, a transition from \(\statey\) to \(\cheatstate\), so that a fair run that fires \(\okaction\) infinitely often eventually has no agent in \(\statey\).
		This proves that, from a given \(\config_0\) with a single agent in \(\startstate\) and enough agents in \(\reservoirstate\), there is a fair run firing \(\okaction\) infinitely often if and only if \(\config_0(\statey) \leq \apoly(\config_0(\statex{1}), \dots, \config_0(\statex{r}))\).
	\end{proof}

	This allows us to prove \cref{monadic-undec-pp} by reduction from \cref{hilbert-variant-undec}.
	Let \(\PP_1\), \(\PP_2\) obtained by applying \cref{prot-weakly-computes-poly} on \(\apoly_1\) and \(\apoly_2\) respectively, with winning transitions \(\okaction_1\) and \(\okaction_2\).
	Without loss of generality, assume that the value of \(\maxexp\) is the same in \(\prot_1\) and \(\prot_2\).
	Our protocol \(\prot\) has a leader part with initial state \(\startstate\), from which a process, the leader, may go to either \(\startstate_1\) and \(\startstate_2\).
	Again, runs starting from initial configurations with more than one leader agent in \(\startstate\) will be sent to \(\cheatstate\) and cannot fire \(\okaction\) infinitely often.
	For all \(i \in \set{1,2}\), when the leader goes to \(\startstate_i\), it will launch the weak computation of \(\apoly_i\).
	Therefore, we obtain that the following two assertions are equivalent:
	\begin{itemize}[noitemsep,topsep=0pt,parsep=0pt,partopsep=0pt]
		\item for all \(x_1, \dots, x_r \in \nats\), \(\apoly_1(x_1, \dots, x_r) \leq \apoly_2(x_1, \dots, x_r)\);
		\item \(\forall \config_0 \in \initialconfigs, \, \forall \atrace_1 \in \fairrunsfrom{\config_0}, \, \exists \atrace_2 \in \fairrunsfrom{\config_0}, \, \neg (\ltlfinally \ltlglobally \,\okaction_1(\atrace_1)) \lor (\ltlfinally \ltlglobally \,\okaction_2(\atrace_2))\).
	\end{itemize}
	First, note that the second statement above can be rephrased as: for a given \(\config_0 \in \initialconfigs\), if there is a fair run \(\run_1\) from \(\config_0\) that fires \(\okaction_1\) infinitely often, then there is a fair run \(\run_2\) from \(\config_0\) that fires \(\okaction_2\) infinitely often.

	We now prove this equivalence.
	Assume first that, for all \(x_1, \dots, x_r \in \nats\), \(\apoly_1(x_1, \dots, x_r) \leq \apoly_2(x_1, \dots, x_r)\).
	Because \(\prot_1\) weakly computes \(\apoly_1\), for all \(\config_0\), if there is a fair run that fires \(\okaction_1\) infinitely often, then we have that:
	\begin{itemize}[noitemsep,topsep=0pt,parsep=0pt,partopsep=0pt]
		\item \(\config_0(\startstate) = 1\),
		\item \(\config_0(\reservoirstate) \geq 1 + \sum_i (\maxexp-1) \config_0(\statex{i})\),
		\item \(\config_0(\statey) \leq \apoly_1(\config_0(\statex{1}), \dots, \config_0(\statex{r}))\).
	\end{itemize}
	Therefore, we also have \(\config_0(\statey) \leq \apoly_1(\config_0(\statex{1}), \dots, \config_0(\statex{r}))\), hence, since \(\prot_2\) weakly computes \(\apoly_2\), there is a fair run from \(\aconfig_0\) that fires \(\okaction_2\) infinitely often.

	Assume now that the second statement is true.
	Let \(x_1, \dots, x_r \in \nats\); and let \(\config_0\) the initial configuration such that:
	\begin{itemize}[noitemsep,topsep=0pt,parsep=0pt,partopsep=0pt]
		\item \(\config_0(\startstate) = 1\),
		\item \(\config_0(\reservoirstate) =  1 + \sum_i (\maxexp-1) \config_0(\statex{i})\),
		\item for all \(i \in \nset{1}{r}\), \(\config_0(\statex{i}) = x_i\),
		\item \(\config_0(\statey) = \apoly_1(x_1, \dots, x_r)\).
	\end{itemize}
	Because \(\prot_1\) weakly computes \(\apoly_1\), we know that there is a fair run from \(\config_0\) that fires \(\okaction_1\) infinitely often; we deduce that there also is a fair run from \(\config_0\) that fires \(\okaction_2\) infinitely often, but \(\prot_2\) weakly computes \(\apoly_2\) therefore this proves that \(\apoly_1(x_1, \dots, x_r) = \config_0(\statey) \leq \apoly_2(x_1, \dots, x_r)\).
	This being true for every \(x_1, \dots, x_r\), we have proven the equivalence.
	This concludes the proof of \cref{monadic-undec-pp}.
\end{toappendix}

%
%
%
%
\section{Verification of HyperLTL for IOPP}
\label{sec:iopp-verif}
\cref{sec:undecidability} showed that verification of HyperLTL in PPs is undecidable, even when the formulas are monadic and have a simple shape.
We thus turn to a subclass of PPs called \emph{immediate observation population protocols} (IOPP) \cite{Comp-Power-Pop-Prot} that has been studied extensively (see e.g. \cite{ERW19,JancarV22,BlondinL23,BGKMWW24}).

\subsection{Immediate Observation PP and Preliminary Results}

\begin{definition}
	An \emph{immediate observation population protocol} (""IOPP"") is a population protocol where all transitions are of the form \((q_1, q_2)\trans{} (q_3, q_2)\).
\end{definition}

We denote a transition \((q_1, q_2)\trans{} (q_3, q_2)\) as \(q_1\trans{q_2} q_3\).
Intuitively, when two agents interact, one remains in its state, as if it was observed by the other agent.

\chana{put an example?}

The IOPP model tends to be simpler to verify than standard PP \cite{ERW19}, notably because it enjoys a convenient monotonicity property: whenever an agent observes an agent in \(q_3\) and goes from \(q_1\) to \(q_2\), another agent in \(q_1\) may do the same ``for free''.
This property is however broken by the \(\ltlnext\) operator of LTL.
In fact, under LTL, IOPP has similar power to regular PP.
Indeed, consider a PP transition \(t: (q_1,q_2) \trans{} (q_3,q_4)\).
One may split this transition into immediate observation transitions \(t_1: q_1 \trans{q_3} q_2\) and \(t_2: q_3 \trans{q_2} q_4\).
Using an LTL formula with the \(\ltlnext\) operator, one can enforce that, whenever \(t_1\) is fired, \(t_2\) must be fired directly after.
Verification of LTL for IOPP is as hard as its counterpart for PP:

\begin{propositionrep}
	\label{prop:iopp-ltl-ack}
	Verification of LTL for IOPP is Ackermann-complete.
\end{propositionrep}
\begin{appendixproof}
	By \cite{EGLM16}, verification of LTL for standard PP is inter-reducible to reachability in Petri nets, an Ackermann-complete problem \cite{Leroux2022,Czerwinski2022}.
	Trivially, verification of LTL for IOPP reduces to verification of LTL for PP, giving decidability in Ackermannian time.
	We now prove Ackermann-hardness.
	We use the following Ackermann-hard problem from the proof of Ackermann-hardness of LTL verification for PP in  \cite{EGLM16} (in fact in the appendix of the long version \cite{EGLM16longversion}):
	\begin{quote}
		{\textbf{Input}: A population protocol \(\PP= (\states,\transitions,\initialstates)\) where \(\states\) contains two special states \(\stateone\) and \(\staterest\) such that all transitions \((q_1,q_2) \trans{} (q_3,q_4) \in \transitions\) are such that \(q_1,q_2 \notin \set{\stateone, \staterest}\).   \\
			\textbf{Question}: Does there exist \(\config_0 \in \initialconfigs\), \(\config \in \configs\) and a finite run \(\run: \config_0 \step{*} \config\) such that \(\config(\stateone) = 1\) and \(\config(q) = 0\) for all \(q \notin \set{\stateone, \staterest}\)?}
	\end{quote}
	We reduce the above problem to (the complement of) the verification problem of LTL for IOPP.
	Let \(\PP = (\states, \transitions, \initialstates)\) be a PP, with \(\stateone, \staterest \in \states\) and such that all transitions \((q_1,q_2) \trans{} (q_3,q_4) \in \transitions\) are such that \(q_1,q_2 \notin \set{\stateone, \staterest}\).
	We assume there is some transition which sends an agent to \(q_1\), else the problem is trivial.
	We construct an IOPP \(\PP' = (\states', \transitions', \initialstates')\) and an LTL formula \(\varphi\) as follows.
	First, we let \(\states' \deff \states\) and \(\initialstates' \deff \initialstates\).
	Let \(t:(q_1,q_2)\trans{}(q_3,q_4) \in \transitions\); we add to \(\transitions'\) transitions \(f_t: q_1\trans{q_2}q_3\) and \(g_t: q_2\trans{q_3}q_4\).
	Our aim is to enforce that, when \(f_t\) is fired, \(g_t\) must be fired immediately after.

	We denote \(\psi_f \deff \biglor_{t \in \transitions} f_t\) and \(\psi_g \deff \biglor_{t \in \transitions} g_t\).
	We let: \[ \varphi_1 \deff (\psi_f \lor \psi_g) \land (\biglor_{t \in \transitions} (f_t \implies \ltlnext g_t) \land (\psi_g \implies \neg \ltlnext \psi_g) \]

	Let \(F \deff \bigcup_{t \in \transitions} f_t\) and \(G \deff \bigcup_{t \in \transitions} g_t\).
	The formula \(\neg \psi_g \land \ltlglobally \varphi_1\) guarantees that the run alternates transitions in \(F\) and in \(G\), starting with a transition in \(F\), and that, whenever \(f_t\) is fired for some \(t \in \transitions\), \(g_t\) follows immediately after.
	This is how we implement PP transitions.
	There is however an issue with \(\neg \psi_g \land \ltlglobally \varphi_1\): this formula would not be satisfied by fair runs.
	For this reason, we only enforce \(\varphi_1\) in a finite initial phase using the \(\ltluntil\) operator.

	We therefore also add to \(\transitions'\) transitions \(\tgood: \stateone \trans{\staterest} \stateone\) and \(\tbad{q}: q \trans{\stateone} q\), for all \(q \neq \staterest\).
	Let \(\goodconfigs \deff \set{\config \in \configs \mid \config(\stateone) =1 \land \forall q \notin \set{\stateone, \staterest}, \config(q) = 0}\).
	We claim that there is a strongly fair run \(\run=\tgood^\omega\) from some \(\config \in \configs\) if and only \(\config \in \goodconfigs\).
	First, suppose that \(\config \in \goodconfigs\).
	We have that, for all \(q \ne \staterest\), transition \(\tbad{q}\) is disabled; \(\tbad{\stateone}\), in particular, is disabled because it requires two agents in \(\stateone\).
	All transitions \(f_t\) and \(g_t\), for \(t \in \transitions\), are also disabled because, by hypothesis, the source states of \(t\) cannot be in \(\set{\stateone, \staterest}\).
	Hence, from \(\config\), all transitions are disabled except \(\tgood\) -- the run that only fires \(\tgood\) is strongly fair.
	Conversely, if there is \(\run \in \fairrunsfrom{\config}\) that fires \(\tgood\) only, then it only visits \(\config\); it must therefore be that \(\tbad{q}\) is disabled from \(\config\), which implies that \(\config(q) = 0\) for all \(q \notin \set{\stateone, \staterest}\) and that \(\config(\stateone) =1\) (if \(\config(\stateone) > 1\) then \(\tbad{\stateone}\) is enabled) so that \(\config \in \goodconfigs\).

	We let \[\varphi \deff \neg \psi_g \land (\varphi_1 \ltluntil (\ltlglobally \, \tgood)).
	\] We prove that \(\PP\) is a positive instance of the problem
	iff \(\PP' \nvDash^{\forall} \neg \varphi\),
	i.e. iff there exists a \(\config_0 \in \initial\) and a run \(\run \in \fairrunsfrom{\config_0}\) such that \(\run \models \varphi\).
	First, suppose that there is \(\config \in \goodconfigs\) that is reachable from \(\config_0 \in \initialconfigs\) in \(\PP\); let \(\config_0, \atrans_1, \config_1,\dots, \atrans_m,\config_m = \config\) denote the corresponding finite run.
	There is a run from \(\config_0\) to \(\config_m\) in \(\PP'\) with sequence of transitions \(f_{\atrans_1}, g_{\atrans_1}, f_{\atrans_2}, g_{\atrans_2}, \dots, f_{\atrans_m}, g_{\atrans_m}\).
	Consider the infinite run from \(\config_0\) with sequence of transitions \(\run \deff f_{\atrans_1}, g_{\atrans_1}, f_{\atrans_2}, g_{\atrans_2}, \dots, f_{\atrans_m}, g_{\atrans_m}, \tgood, \tgood,\dots\) This is a run because \(\config \step{\tgood} \config\), and it is strongly fair by the reasoning above.
	Also, \(\run \models \varphi\).

	Conversely, suppose that there is \(\config_0 \in \initialconfigs'\) and an infinite run \(\run\) of \(\PP'\) such that \(\run \models \varphi\).
	By construction of \(\varphi\), \(\run\) can be split into two phases; a finite part where \(\varphi_1\) holds, so that the sequence of transitions is of the form \(f_{\atrans_1}, g_{\atrans_1},\dots,f_{\atrans_m}, g_{\atrans_m}\), and an infinite part where \(\tgood\) is the only transition fired.
	Let \(\config\) denote the configuration in between the two parts.
	Because the run is strongly fair from \(\config\), we have \(\config \in \goodconfigs\).
	It is easy to prove that the sequence of transitions \(\atrans_1, \dots, \atrans_m\) yields a valid finite run of \(\PP\) from \(\config_0\) to \(\config\), which concludes the proof.
\end{appendixproof}

\begin{remark}
	\label{rmk:stutter-inv}
	The fragment of LTL with no \(\ltlnext\) operator is equivalent to stutter-invariant LTL \cite{PeledW97,Etessami00}.
	Let \(\varphi\) be an LTL\(\setminus \ltlnext\) formula \(\varphi\), let \(t_1,t_2,\ldots \in \transitions\) and \(k_1, k_2, \ldots \geq 1\).
	This means that we have \(t_1^{k_1} \concat t_2^{k_2} \concat \ldots \models \varphi\) if and only if \(t_1 \concat t_2 \concat \ldots \models \varphi\).
\end{remark}

Below, we consider the fragment LTL\(\setminus \ltlnext\) as done in prior work \cite{FortinMW17} in which the systems under study feature monotonicity due to non-atomic writes: stuttering-invariance is a natural choice for systems with monotonicity properties.
We show that, even then, verification of HyperLTL\(\setminus \ltlnext\) formulas for IOPP is undecidable.

\begin{theoremrep}%
	\label{undec-hyperltl-iopp}
	Verification of HyperLTL\(\setminus \ltlnext\) is undecidable for IOPP.
\end{theoremrep}
%
\begin{appendixproof}
	The proof is by reduction from the halting problem for \(2\)-counter machines with zero-tests, an undecidable problem \cite{Minsky67}.
	A ""\(2\)-counter machine"" consists in two counters \(c_1,c_2\) plus a list of instructions \(l_1, \ldots, l_n\) and one instruction \(\halt\).
	Instructions \(l_1, \ldots, l_n\) are of the form: \(\incr(c_i)\) which increments counter \(c_i\), \(\decr(c_i)\) which decrements counter \(c_i\), and \(\zero(c_i,j)\) which moves to instruction \(l_j\) if \(c_i =0\).
	A configuration of a \(2\)-counter machine is \((l_k, c_1, c_2)\), the current value of the counters as well as the current (not yet executed) instruction.
	The initial configuration of the machine is \((l_1, 0, 0)\).
	If \(l_k\) is an increment or a decrement, configuration \((l_k, c_1, c_2)\) moves to configuration \((l_{k+1}, c_1', c_2')\), updating the counters accordingly.
	If \(l_k=\zero(c_i,j)\) then \((l_k, c_1, c_2)\) moves to \((l_j, c_1, c_2)\) if the zero test is successful, and \((l_{k+1}, c_1, c_2)\) otherwise.
	A \(2\)-counter machine ""halts"" if it reaches the \(\halt\) instruction from the initial configuration.

	Fix a \(2\)-counter machine \(\machine\) with instructions \(l_1, \ldots, l_n,\halt\).
	We build an IOPP \(\PP\), with the goal of simulating executions of \(\machine\) faithfully using runs of \(\PP\).
	For each instruction \(l_j\) in \(\machine\) there is a corresponding state \(l_j\) in \(\PP\); there are two states \(c_1,c_2\) that represent the counters of \(\machine\); there is a reservoir state \(\res\) (which will intuitively contain a large amount of agents) which is also the only initial state of \(\PP\), and a sink state \(\bot\).

	A \emph{faithful} run in \(\PP\) will have exactly one agent in the instruction states (except in the first configuration of the run), and the number of agents in state \(c_i\) will symbolize the value of the counter \(c_i\).
	We want a faithful run to simulate the instructions of \(\machine\) correctly, updating the counters and moving to the correct instruction.
	We will use the gadgets illustrated in Figures \ref{fig:incr-app} and \ref{fig:zerotest-app} to simulate the instructions.
	To ensure that the gadgets are used correctly and that there is exactly one agent in the instructions states, we will add \emph{bad} transitions (in yellow in the figures).
	A faithful run is then a run in which no bad transition is ever activated, and this will be enforced by a HyperLTL\(\setminus \ltlnext\) formula.

	\begin{figure}[ht]
		\centering
		%
		%

		\begin{tikzpicture}[->, node distance=1.5cm, auto, thick, font = \small]
			\node[place] (res) {$\res$};
			\node[transition] (t1) [right of=res] {$t_1$};
			\node[place] (int) [right of=t1] {$\interm_k$};
			\node[transition, fill=yellow] (bad) [above=0.5 of int] {$b_k$};
			\node[transition] (t2) [right of=int] {$t_2$};
			\node[place] (c) [right of=t2] {$c_i$};
			\node[place] (l1) [below of=t1] {$l_k$};
			\node[transition] (t3) [below of=int] {$t_3$};
			\node[place] (aux) [below of=t2] {$\aux_k$};
			\node[transition] (t4) [below of=c] {$t_4$};
			\node[place] (l2) [right of=t4] {$l_{k+1}$};
			\node[place] (bot) [right of=bad] {$\bot$};

			\path[->]
				(res) edge node {} (t1)
				(t1)  edge node {} (int)
				(int) edge node {} (t2)
				(int) edge node {} (bad)
				(t2)  edge node {} (c)
				(l1)  edge node {} (t3)
				(t3)  edge node {} (aux)
				(aux) edge node {} (t4)
				(t4)  edge node {} (l2)
				(bad) edge node {} (bot)
			;

			\path[<->, dashed]
				(l1)  edge node {} (t1)
				(t3)  edge node {} (int)
				(aux) edge node {} (t2)
				(int) edge [bend left=40] node {} (bad)
			;

		\end{tikzpicture}
		\caption{Gadget simulating instruction \(l_k:\incr(c_i)\).
			The notation is Petri net-inspired: circles are states, squares are transitions and a dashed line is an observation.
		}
		\label{fig:incr-app}
	\end{figure}
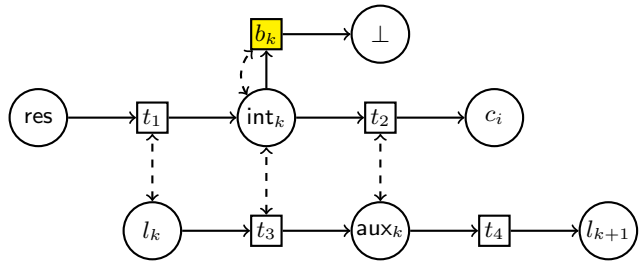
	In the figures illustrating the gadgets, for ease of representation, some transitions \(t\) are missing an observation state, i.e. a state \(q_3\) such that \(t: q_1 \trans{q_3} q_2\).
	For these transitions, the (undepicted) observation state is \(\res\).
	Figure \ref{fig:incr-app} illustrates the gadget used to simulate an instruction \(l_k:\incr(c_i)\) in \(\PP\).
	An agent in the reservoir state observes the instruction agent in \(l_k\) and moves to an intermediary state.
	The instruction agent moves to an auxiliary state upon observation of this previous agent, which then moves to \(c_i\), thus ``incrementing'' the counter.
	Finally, the instruction agent moves to the next instruction.
	The simulation of the instruction could be faulty if more than one agent in \(\res\) moves to the intermediary state upon observing the instruction agent in \(l_k\).
	These agents can then move to \(c_i\) after observation of the agent in the auxiliary state, thus incrementing the counter by more than one.
	The role of transition \(b_k\) (b for bad) is to detect this faulty behavior.
	If \(\interm_k\) contains more than one agent, \(b_k\) is activated.

	Decrements are simulated similarly to increments, by swapping \(\res\) and the counter state.
	\begin{figure}[ht]
		\centering
		%
		%

		\begin{tikzpicture}[->, node distance=1.5cm, auto, thick, font = \small]
			\node[place] (l1) {$l_k$};
			\node[transition] (t1) [right of=l1] {$t_2$};
			\node[place] (aux) [right of=t1] {$\aux_k$};
			\node[transition] (t2) [right of=aux] {$t_3$};
			\node[place] (l2) [right of=t2] {$l_{j}$};
			\node[transition] (t3) [left of=l1] {$t_1$};
			\node[place] (l3) [left of=t3] {$l_{k+1}$};
			\node[place] (c) [above of=t3] {$c_i$};
			\node[transition, fill=yellow] (bad) [above=0.93 of aux] {$b_k$};
			\node[place] (bot) [right of=bad] {$\bot$};

			\path[->]
				(l1)  edge node {} (t1)
				(t1)  edge node {} (aux)
				(aux) edge node {} (t2)
				(aux) edge node {} (bad)
				(t2)  edge node {} (l2)
				(l1)  edge node {} (t3)
				(t3)  edge node {} (l3)
				(bad) edge node {} (bot)
			;

			\path[<->, dashed]
				(t3) edge node {} (c)
				(c)  edge node {} (bad)
			;

		\end{tikzpicture}
		\caption{Gadget simulating instruction \(l_k:\zero(c_i,j)\).}
		\label{fig:zerotest-app}
	\end{figure}
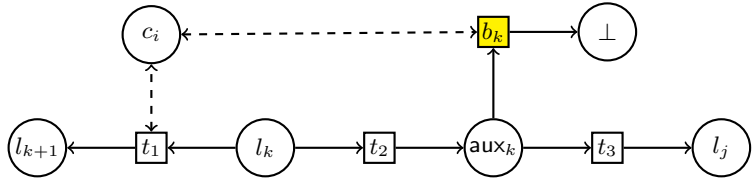
	Figure \ref{fig:zerotest-app} illustrates the gadget used to simulate a zero-test instruction \(l_k:\zero(c_i,j)\).
	The instruction agent can observe another agent in \(c_i\), guaranteeing that the counter value is non-zero, and move to \(l_{k+1}\).
	The instruction agent can also move to an auxiliary state, and then to \(l_j\).
	We want this to happen only if \(c_i\) is zero: if the instruction agent moves to the auxiliary state while \(c_i\) contains at least one agent, then transition \(b_k\) is activated.

	To start the simulation, we add a transition from \(\res \trans{\res} l_1\).
	To ensure that there is exactly one agent in the instruction states, we add \(n^2 / 2\) bad transitions \(l_i \trans{l_j} \bot\) for every \(1 \le i \le j \le n\).
	These are activated if there are two or more agents in the instruction states.
	We add some transitions intended to end the simulation:
	\begin{itemize}[noitemsep,topsep=0pt,parsep=0pt,partopsep=0pt]
		\item
		      a transition \(\halttrans: \halt \trans{\res} \halt\) that can be taken only if instruction state \(\halt\) is reached,
		\item
		      for every state \(q\neq \halt\), a transition \(\halttrans_q: q \trans{\res} \halt\), and
		\item
		      a transition \(\halttrans_{\res}: \res \trans{\halt} \halt\).
	\end{itemize}
	Intuitively, these transitions ensure that once \(\halt\) is reached, all agents will eventually end up in \(\halt\).

	The final element we need is our HyperLTL\(\setminus \ltlnext\) formula.
	It will be satisfied in \(\PP\) if and only if the \(2\)-counter machine \(\machine\) \emph{does not} halt.
	Let \(\bad\) be the set of all bad transitions.
	Given a run \(\run \in \fairrunsfrom{\config_0}\), we define \(\psi_{\bad}(\run)\) to express that some bad transition is activated in \(\run\): \[ \psi_{\bad}(\run)= \exists \run'.
		(\bigvee_{t\in \transitions} t_\run \land t_{\run'}) \ltluntil
		(\bigvee_{b\in \bad}  b_{\run'})\enspace .
	\]
	The formula expresses that there exists another run
	which takes all the same transitions as \(\run\) until it takes a bad transition \(b\).
	If \(\run\) and \(\run'\) start in the same initial configuration and take the same transitions until \(b\), then \(b\) is activated in \(\run\) too.
	We take the following as our final HyperLTL\(\setminus \ltlnext\) formula \(\psi\): \[ \forall \run.
		\neg( \ltlfinally \halttrans_\run) \lor
		\psi_{\bad}(\run)\enspace.
	\]
	Machine \(\machine\) does not halt if and only if \(\PP\modelsforall\psi\): if \(\machine\) does not halt, then by construction there is no run of \(\PP\) that can ever take \(\halttrans\), so \(\PP\modelsforall\psi\).
	Suppose \(\machine\) halts.
	There exists a finite faithful run \(\sigma_1\) in \(\PP\) that puts the instruction agent in state \(\halt\).
	Extend \(\sigma_2\) with a finite run \(\sigma_2\) which uses the \(\halttrans_q\) and \(\halttrans_{\res}\) to bring all agents to state \(\halt\).
	There exists an initial configuration \(\config_0\) with a large enough number of agents in \(\res\) such that \(\sigma_1 \sigma_2 (\halttrans)^\omega\) can be taken.
	This run is strongly fair and thus \(\PP\nvDash^{\forall}\psi\).
\end{appendixproof}

%
%

However, we will show that the monadic HyperLTL\(\setminus \ltlnext\) case is decidable.
\subsection{Product Systems}
Our approach consists, as in the proof of \cref{strong-fairness-stochastic-scheduler}, to define \emph{product systems} that combine the IOPP with a Rabin automaton recognizing an LTL formula.

\begin{definition}
	A ""product system"" is a pair \(\prodsystem = (\prot, \buchiaut)\) where \(k \in \nats\) and:
	\begin{itemize}[noitemsep,topsep=0pt,parsep=0pt,partopsep=0pt]
		\item \(\prot = (\states, \transitions, \initialstates)\) is an "IOPP",
		\item \(\buchiaut = (\bstates, \btransitions, \binitstate, \bwinning)\) is a "deterministic Rabin automaton" over \(\transitions\).
	\end{itemize}
\end{definition}
We refer to the part with the Rabin automaton as the ""control part"".
There are two distinct notions of size for a "product system": the ""protocol size"" \(\protocolsize{\prodsystem} \deff \size{\states}\) and the ""control size"" \(\controlsize{\prodsystem} \deff \size{\bstates}\).
The reason for this distinction is that the "control size" is typically exponential in the size of the LTL formulas, so that keeping track of the two sizes separately will later improve our complexity analysis.

\paragraph*{Semantics of Product Systems.}
A "configuration" of \(\prodsystem\) is an element of \(\prodconfigs \deff \multisetsof{\states} \times \bstates\).
Given a set \(S \subseteq \bstates\), let \(\prodconfigsproj{S} \deff \set{(\config, \abuchistate) \mid \abuchistate \in S}\).
Moreover, we let \(\prodinitialconfigs \deff \set{(\config, \binitstate) \mid \config \in \initial}\) be the set of initial configurations of the "product system".

In product systems, unlike in the proof of \cref{strong-fairness-stochastic-scheduler}, the semantics in the PP is modified to match the monotonicity properties of the system.
More precisely, we rely on \emph{accelerated semantics} for the IOPP: in \(\prot\), there is an \emph{accelerated step} from \(\config\) to \(\config'\) with transition \(\atrans \in \transitions\) when there is \(k \geq 1\) such that \(\config \step{\atrans^k} \config'\).
Given two configurations \(\prodconfig = (\config,\bstate), \prodconfig' = (\config', \bstate') \in \prodconfigs\) and \(t \in \transitions\), we let \(\prodconfig \step{\atrans} \prodconfig'\) when there is \(k \geq 1\) such that \(\config \step{\atrans^k} \config'\) in \(\prot\) and \(\transitions(\abuchistate, \atrans) = \abuchistate'\).
A step in the "product system" corresponds to an accelerated step in \(\prot\) whose transition is read by \(\buchiaut\).
Note that there is no communication from the "control part" to the "IOPP".
In "product systems" runs and operators \(\prestar{\cdot}\), \(\poststar{\cdot}\) are defined as expected.

%
%

%
\subsection{Satisfiability as a Reachability Problem}

We fix \(\PP\) an IOPP, \(\varphi\) an LTL\(\setminus \ltlnext\) formula, \(\buchiaut = (\bstates, \btransitions, \binitstate, \bwinning) \) a deterministic Rabin automaton recognizing \(\varphi\) obtained using \cref{ltl-rabin} and we let \(\prodsystem = (\prot, \buchiaut)\).

Recall that, in \(\prot\), there is an \emph{accelerated step} from \(\config\) to \(\config'\) using \(\atrans\) when there are \(k \geq 1\) and \(t \in \transitions\) such that \(\config \step{t^k} \config'\).
A (finite) \emph{accelerated run} is a sequence \(\config_0, \atrans_1, \config_1, \dots, \atrans_m\) such that, for all \(i \in \nset{1}{m}\), there is an accelerated step from \(\config_{i-1}\) to \(\config_i\) using \(\atrans_i\).
We similarly define infinite accelerated runs.
We extend the notion of strong fairness: an infinite accelerated run \(\accrun\) is \emph{strongly fair} when, for every finite accelerated run \(\accrun'\), if the first configuration of \(\accrun'\) is visited infinitely often in \(\accrun\) then \(\accrun'\) appears infinitely often in \(\accrun\).
A run \(\run\) of \(\prodsystem\) can be projected onto \(\prot\) to obtain an accelerated run of \(\prot\), denoted \(\protproj{\run}\); \(\run\) is called \emph{protocol-fair} when the accelerated run \(\protproj{\run}\) is strongly fair.
Given an accelerated run \(\accrun = \config_0, \atrans_1, \config_1, \atrans_2, \dots\), we let \(\accrun \models \varphi\) when \(\atrans_1 \atrans_2 \dots \models \varphi\).
An accelerated infinite run \(\accrun= \config_0, \atrans_1, \config_1, \atrans_2, \dots\) is an \emph{acceleration} of an infinite run \(\run\) when there are \(k_1, k_2,\ldots \geq 1\) such that \(\run\) is of the form
\(\config_0, \atrans_1^{k_1}, \config_1, \atrans_2^{k_2}, \config_2, \ldots\)
\begin{lemmarep}
	\label{strongly-fair-acceleration}
	Given a strongly fair accelerated run \(\accrun\), there is a strongly fair run \(\run\) such that \(\accrun\) is an acceleration of \(\run\).
	Conversely, given a strongly fair run \(\run\), there is a strongly fair acceleration \(\accrun\) of \(\run\).
\end{lemmarep}
\begin{appendixproof}
	We start with the first statement.
	Let \(\accrun = \config_0, \atrans_1, \config_1, \atrans_2, \dots\) be a strongly fair accelerated run.
	Let \(\run\) be the infinite non-accelerated run equal to \(\config_0 \step{t_1^{k_1}} \config_1 \step{t_2^{k_2}} \config_2 \dots\) where, for all \(i \geq 1\), \(k_i\) is the minimal integer \(k \geq 1\) such that \(\config_{i-1} \step{t_i^k} \config_i\).
	Note that all \(k_i\) exist because \(\accrun\) is an accelerated run.
	Clearly, \(\accrun\) is an acceleration of \(\run\).
	We now claim that \(\run\) is strongly fair.
	Let \(\run' = \config_0', \atrans_1', \config_1', \dots, \atrans_m', \config_m'\) where \(\config_0'\) is visited infinitely often in \(\run\).
	We claim that \(\config_0'\) appears infinitely often in \(\accrun\).
	Trivially, there is a configuration \(\config \in \configs\) that is visited infinitely often in \(\accrun\).
	Both \(\config\) and \(\config_0'\) are visited infinitely often in \(\run\), therefore there is a finite run from \(\config\) to \(\config_0'\), and hence there is an accelerated finite run from \(\config\) to \(\config_0'\).
	Because \(\accrun\) is strongly fair, this finite accelerated run appears infinitely often in \(\accrun\) so that \(\config_0'\) is visited infinitely often in \(\accrun\).
	Let \(\accrun'\) be the accelerated run equal to \(\run'\), but seen as an accelerated run.
	By strong fairness, \(\accrun'\) appears infinitely often in \(\accrun\).
	For each \(i \in \nset{1}{m}\), we have \(\config_{i-1}' \step{\atrans_i'} \config_i'\).
	Therefore, whenever \(\accrun'\) appears in \(\accrun\), all the corresponding \(k_i\) are equal to \(1\) by minimality.
	This proves that, for each occurrence of \(\accrun'\) in \(\accrun\), there is an occurrence of \(\run'\) in \(\run\).
	We conclude that \(\run'\) appears infinitely often in \(\run\) and that \(\run\) is strongly fair.

	We now prove the second statement.
	Let us fix a probability distribution \(f: \nats \to [0,1]\) such that \(f(n) > 0\) for all \(n\) (\emph{e.g.}, a geometric distribution).
	We first define a random variable \(R\) that takes value over the set of infinite accelerated runs.
	We build \(R\) as follows.
	We proceed (accelerated) step by (accelerated) step by grouping consecutive steps of \(\run\) with the same transition.
	Suppose that the acceleration has been built until the \(i\)-th configuration of \(\run\); let \(\config\) denote this configuration, and let \(t\) denote the next transition in \(\run\) (the \(i\)-th transition of \(\run\), which is fired from \(\config\)).
	We pick an integer \(m \in \nats\) according to \(f\), independently from the past.
	If steps \(i\) to \(i+m-1\) of \(\run\) use transition \(t\) then we accelerated all those steps into one accelerated step from the \(i\)-th configuration of \(\run\) to the \(i+m\)-th configuration of \(\run\), and we repeat the procedure from the \((i+m)\)-th configuration of \(\run\).
	Otherwise, we define the next accelerated step as equal to the step from the \(i\)-th configuration to the \((i+1)\)-th configuration of \(\run\) (the next step is not grouped with other steps), and we repeat the procedure from the \((i+1)\)-th configuration of \(\run\).

	By repeating this construction, we obtained an infinite accelerated run \(R\).
	Trivially, \(R\) is an acceleration of \(\run\).
	We claim that \(R\) is protocol-fair with probability \(1\).
	Let \(\config_0 \in \configs\) and let \(\accrun = \config_0, \atrans_1, \config_1, \atrans_2, \dots, \atrans_m \config_m\) be an accelerated finite run from \(\config_0\).
	There are \(k_1, \dots, k_m\) such that \(\config_{i-1} \step{\atrans_i^{k_i}} \config_i\) for all \(i \in \nset{1}{m}\).
	Whenever \(\config\) appears in \(R\), there is probability at least \(\prod_{i=1}^{m} f(k_i)>0\) that the next \(m\) accelerated steps are the same as in \(\sigma\).
	This proves that there is probability \(0\) that \(\config\) is visited infinitely often in \(R\) but that \(\accrun\) appears finitely often.
	Because the set of configurations and the set of finite accelerated runs are countable, this proves that there is probability zero that there are \(\config\) and \(\accrun\) disproving strong fairness.
	Hence, \(R\) is strongly fair with probability one.
	This in particular implies the existence of a strongly fair acceleration of \(\run\).
\end{appendixproof}

For \(L \subseteq \bstates\), we write \(\prodconfigsproj{L} \deff \configs \times L \subseteq \prodconfigs\); also, for \(\cSet \subseteq \prodconfigs\), \(\setcomplement{\cSet} \deff \prodconfigs \setminus \cSet\).
We let \(\satset{\exists \run .
	\, \varphi} \deff \set{\config \in \configs \mid \exists \run \in \fairrunsfrom{\config}, \, \run \models \varphi}\).
Similarly, we let \(\satset{\forall \run .
	\, \varphi} \deff \set{\config \in \configs \mid \forall \run \in \fairrunsfrom{\config}, \, \run \models \varphi} = \configs \setminus {\satset{\exists \run . \, \neg \varphi}}\).
We give a characterization of these sets.

\begin{theorem}
	\label{lm:characterization}
	A configuration \(\config\) of \(\PP\) is in \(\satset{\exists \run .
		\, \varphi}\) if and only if \((\config,\binitstate)\) is in
	\[
		\ltlgre \deff \preop^* \left( \textstyle{\bigcup_{(F, G) \in \bwinning}} \setcomplement{\prestar{\prodconfigsproj{F}}} \cap \setcomplement{\prestar{\setcomplement{\prestar{\prodconfigsproj{G}}}}} \right)
	\]
\end{theorem}

\begin{proof}
	Let \(\config \in \configs\).
	By \cref{strongly-fair-acceleration} and \cref{rmk:stutter-inv}, \(\config \in \satset{\exists \run .
		\, \varphi}\) if and only if there is a strongly fair accelerated run \(\accrun\) from \(\config\) such that \(\accrun \models \varphi\).
	Let \(G\) denote the graph whose vertices are the configurations of the product system reachable from \((\config, \binitstate)\) and where there is an edge from \(\prodconfig\) to \(\prodconfig'\) whenever \(\prodconfig \step{t} \prodconfig'\) for some \(t \in \transitions\).
	We claim that there is a strongly fair accelerated run \(\accrun\) from \(\config\) such that \(\accrun \models \varphi\) if and only if there is a bottom SCC \(S\) of \(G\) reachable from \((\config, \binitstate)\) that is \emph{winning}, \emph{i.e.}, such that there is \((F,G) \in \bwinning\) for which \(S \cap \prodconfigsproj{G} \ne \emptyset\) but \(S \cap \prodconfigsproj{F} = \emptyset\).

	The arguments are the same as in the proof of \cref{strong-fairness-stochastic-scheduler}, but with accelerated semantics in \(\prot\).
	If we have such an SCC \(S\), it is easy to build a protocol-fair run \(\run\) of \(\prodsystem\) that goes to \(S\) and visits all configurations in \(S\) infinitely often.
	We let \(\accrun \deff \protproj{\run}\); \(\accrun\) is strongly fair and, because \(S\) is winning, \(\accrun \models \varphi\).
	Suppose now that we have a strongly fair accelerated run \(\accrun\) such that \(\accrun \models \varphi\).
	Let \(\run\) be the run of \(\prodsystem\) such that \(\protproj{\run} = \accrun\); \(\run\) is protocol-fair.
	Let \(S\) be the SCC visited infinitely often in \(\run\); \(S\) is bottom and \(\run\) visits infinitely often all configurations in \(S\).
	Indeed, the same arguments as in the proof of \cref{strong-fairness-stochastic-scheduler} apply, except that we rely on strong fairness of the accelerated run, which makes no difference since strong fairness is defined the same for accelerated and non-accelerated runs.

	It remains to prove that there is a winning bottom SCC \(S\) reachable from \((\config, \binitstate)\) if and only if \((\config,\binitstate) \in \ltlgre\).
	Suppose first that there is such an SCC \(S\); let \(\prodconfig \in S\) and let \((F,G) \in \bwinning\) such that \(S \cap \prodconfigsproj{G} \ne \emptyset\) and \(S \cap \prodconfigsproj{F} = \emptyset\).
	We have \((\config, \binitstate) \in \prestar{\prodconfig}\).
	Since \(S\) is bottom and \(S \cap \prodconfigsproj{F} = \emptyset\), we have \(\poststar{\prodconfig} \cap \prodconfigsproj{F} = \emptyset\) and so \(\prodconfig \in \setcomplement{\prestar{\prodconfigsproj{F}}}\).
	We also have \(S = \poststar{S}\), and because \(S \cap \prodconfigsproj{G} \ne \emptyset\), we have \(\poststar{S}\subseteq \prestar{\prodconfigsproj{G}}\); therefore \(S \cap \prestar{\setcomplement{\prestar{\prodconfigsproj{G}}}} = \emptyset\).
	This proves that \(\prodconfig \in \setcomplement{\prestar{\prodconfigsproj{F}}} \cap \setcomplement{\prestar{\setcomplement{\prestar{\prodconfigsproj{G}}}}}\); therefore \((\config, \binitstate) \in \ltlgre\).
	Suppose now that \((\config,\binitstate) \in \ltlgre\).
	Let \((F,G) \in \bwinning\), \(\prodconfig \in \poststar{(\config, \binitstate)}\) such that \(\prodconfig \in \setcomplement{\prestar{\prodconfigsproj{F}}} \cap \setcomplement{\prestar{\setcomplement{\prestar{\prodconfigsproj{G}}}}}\).
	Let \(S\) be an SCC reachable from \(\prodconfig\).
	We claim that \(S\) is winning.
	Because \(S \subseteq \poststar{\prodconfig}\), we have \(S \cap \prodconfigsproj{F} = \emptyset\).
	Also, if we had \(S \cap \prodconfigsproj{G} \ne \emptyset\) then any configuration \(\prodconfig_S \in S\) would be in \(\setcomplement{\prestar{\prodconfigsproj{G}}}\), so that \(\prodconfig\) would be in \(\prestar{\setcomplement{\prestar{\prodconfigsproj{G}}}}\), a contradiction.
\end{proof}

\subsection{\texorpdfstring{\(K\)}{K}-blind Sets}
Let \(K \in \nats\).
A set \(S \subseteq \configs\) of configurations of \(\prot\) is \emph{\(K\)-blind} when, for all \(\config \in \configs\) and \(q \in \states\) such that \(\config(q) \ge K\), \(\config \in S\) if and only if \(\config+ \vec{q} \in S\).
Similarly, a set \(\cSet \subseteq \prodconfigs\) of configurations of \(\prodsystem\) is \emph{\(K\)-blind} when, for all \((\config,\bstate) \in \prodconfigs\) and \(q \in \states\) such that \(\config(q) \ge K\), \((\config,\bstate) \in \cSet\) if and only if \((\config+ \vec{q}, \bstate) \in \cSet\).

\begin{example}
	\label{ex:initial-and-winning}
	The set \(\initialconfigs\) is \(1\)-blind, because \(\config \in \initialconfigs\) if and only if \(\config(q)\) is non-zero when \(q\in I\) and zero otherwise.
	For the same reason, the set \(\prodinitialconfigs \subseteq \prodconfigs\) is \(1\)-blind.
	Also, for all \(L \subseteq \bstates\), the set \(\prodconfigsproj{L}\) is \(0\)-blind.
\end{example}

\begin{lemmarep}
	\label{lm:k-blind-boolean}
	Let \(\cSet_1\) a \(K_1\)-blind set and \(\cSet_2\) a \(K_2\)-blind set of \(\prodsystem\).
	Then \(\cSet_1 \star \cSet_2\) is a \(\max(K_1,K_2)\)-blind set for \(\star \in \set{\cup, \cap}\).
	Additionally, \(\setcomplement{\cSet_1}\) is a \(K_1\)-blind set.
\end{lemmarep}
\begin{appendixproof}
	Let \(\cSet_1\) a \(K_1\)-blind set and \(\cSet_2\) a \(K_2\)-blind set.
	Let \((\config,\bstate)\) be a configuration such that \(\config(q) \ge \max(K_1,K_2)\) for some state \(q\).
	Suppose \((\config,\bstate)\) is in \(\cSet_1 \cup \cSet_2\).
	Thus \((\config,\bstate)\) is in \(\cSet_i\) for an \(i \in \set{1,2}\).
	By \(K_i\)-blindness of \(\cSet_i\), \((\config+\vec{q},\bstate)\) is in \(\cSet_i\) and thus in \(\cSet_1 \cup \cSet_2\).
	Conversely if \((\config+\vec{q},\bstate)\) is in \(\cSet_i\) then \((\config,\bstate)\) is in \(\cSet_i\) and thus in \(\cSet_1 \cup \cSet_2\).
	The proof is similar for \(\cSet_1 \cap \cSet_2\).
	Let \((\config,\bstate)\) such that \(\config(q) \ge K_1\) for some state \(q\).
	Since \(\cSet_1\) a \(K_1\)-blind set, \((\config,\bstate) \notin \cSet_1\) if and only if \((\config+\vec{q},\bstate) \notin \cSet_1\).
	Thus \((\config,\bstate) \in \overline{\cSet_1}\) if and only if \((\config+\vec{q},\bstate) \in \overline{\cSet_1}\).
\end{appendixproof}

Next we find that \(K\)-blind sets are closed under reachability if we enlarge \(K\).

\begin{theorem}%
	\label{thm:k-blind-post}
	Let \(\cSet\) be a \(K'\)-blind set of \(\prodsystem\).
	Then \(\poststar{\cSet}\) and \(\prestar{\cSet}\) are \(K\)-blind sets for \(K := \size{\states}^2 \max(K', 2 B)\) where \(B = \size{\bstates}^{3^{\size{\states}^2+2} \cdot 2(\log(\size{\states}^2+2) +1) \size{\states}^2}\).
\end{theorem}

This theorem crucially relies on the immediate observation assumption, its proof is technical and presented in \cref{sec:structural-bounds}.
Note that \(K\) is doubly-exponential in \(\size{\states}\) but polynomial in \(\size{\bstates}\) and in \(K'\), so that this bound is doubly-exponential in \(\size{\varphi}\) if we let \(\buchiaut = \ltlaut{\varphi}\) using \cref{ltl-rabin}.
Let us apply this result to \(\satset{\exists \run .
	\, \varphi}\):

\begin{lemma}
	\label{thm:ltl-k-blind}
	Set \(\satset{\exists \run .
		\, \varphi}\) is \(K\)-blind with
	\(K\) doubly-exponential in \(\size{\prot}\)  and \(\size{\varphi}\).
\end{lemma}
\begin{proof}
	By \cref{lm:characterization} we find that \(\satset{\exists \run .
		\, \varphi} \times \set{\binitstate} = \ltlgre\).
	\(\prodconfigsproj{F}\) and \(\prodconfigsproj{G}\) are \(0\)-blind for each pair \((F, G) \in \bwinning\).
	The result follows by iterative applications of \cref{thm:k-blind-post} and \cref{lm:k-blind-boolean}.
\end{proof}

\subsection{LTL and HyperLTL Verification}
We now apply the results from the previous sections to verification of LTL\(\setminus \ltlnext\) and verification of monadic HyperLTL\(\setminus \ltlnext\) for IOPP; we prove that both problems are decidable and in 2-\EXPSPACE.
For LTL\(\setminus \ltlnext\), \cref{thm:ltl-k-blind} shows that we only need to check emptiness of a \(K\)-blind set for \(K\) bounded doubly-exponentially.

\begin{theorem}
	\label{thm:ltl-iopp}
	Verification of LTL\(\setminus \ltlnext\) for IOPP is in 2-\EXPSPACE, and the same is true for its existential variant.
\end{theorem}

\begin{proof}
	By Savitch's Theorem, we can present a non-deterministic procedure.
	Let \(\varphi\) be an LTL\(\setminus \ltlnext\) formula, and \(\PP\) an IOPP.
	We construct \(\ltlaut{\varphi}\) using \cref{ltl-rabin}; for this, we pay a doubly-exponential cost in \(\size{\varphi}\), which is the most costly part of the procedure.
	We work in the product system \(\prodsystem \deff (\prot, \ltlaut{\varphi})\).

	Observe that \(\PP \modelsforall \varphi\) if and only if \(\initialconfigs \cap \satset{\exists \run .
		\, \neg \varphi} = \emptyset\), so that it suffices to consider the existential variant.
	We therefore want to decide whether \linebreak \(\satset{\exists \run.
		\,\varphi} \cap \initialconfigs \ne \emptyset\).
	The set \(\initialconfigs\) is \(1\)-blind; by \cref{thm:ltl-k-blind} and \cref{lm:k-blind-boolean}, \(\initialconfigs \cap \satset{\exists \run .
		\, \varphi}\) is \(K\)-blind for \(K\) doubly-exponential in the size of \(\prot\) and in the size of \(\varphi\).

	Hence, \(\initial \cap \satset{\exists \run .
		\, \varphi} \ne \emptyset\) if and only if it contains \(\config_0\) such that
	\(\config_0(q) \le K\) for all \(q \in \states\).
	We guess such a configuration \(\config_0\).
	We can write \(\config_0\) in binary, and thus in exponential space.
	Checking if \(\config_0 \in \initial\) is immediate.
	By \cref{lm:characterization}, we can check if \(\config_0 \in \satset{\exists \run .
		\, \varphi}\) by checking whether, in the product system \(\prodsystem = (\prot, \ltlaut{\varphi})\), \linebreak
	\((\config_0, \binitstate) \in \ltlgre = \bigcup_{(F, G) \in \bwinning} \preop^* \left(
	\setcomplement{\prestar{\prodconfigsproj{F}}} \cap \setcomplement{\prestar{\setcomplement{\prestar{\prodconfigsproj{G}}}}} \right)\).

	We guess a Rabin pair \((F, G) \in \bwinning\).
	We only need to consider configurations in \(\prodconfigs_{\config_0} \deff \set{(\config, \bstate) \in \prodconfigs \mid |\config|=|\config_0| }\).
	Given a set \(\cSet \subseteq \prodconfigs\) whose membership can be checked in 2-\EXPSPACE for configurations in \(\prodconfigs_{\config_0}\), checking whether a configuration \(\prodconfig \in \prodconfigs_{\config_0}\) is in \(\prestar{\cSet}\) can also be done in 2-\EXPSPACE: guess a run starting at \(\prodconfig\), step by step.
	After each step, check if the current configuration \(\prodconfig'\) is in \(\cSet\).
	We only remember the previous configuration and the current one; checking the step can be done in 2-\EXPSPACE because we have constructed \(\ltlaut{\varphi}\) and because, in the protocol, a step corresponds to simple arithmetic operations.
	For each \(H\in\set{F,G}\), checking whether a configuration \(\prodconfig \in \prodconfigs_{\config_0}\) is in \(\prodconfigsproj{H}\) is easy.
	Therefore, checking whether \(\prodconfig \in \prodconfigs_{\config_0}\) is in \(\prestar{\prodconfigsproj{H}}\) can be done in 2-\EXPSPACE.
	By iterating this technique and treating Boolean operations in a natural manner, we check whether \((\config_0, \binitstate) \in \preop^* ( \setcomplement{\prestar{\prodconfigsproj{F}}} \cap \setcomplement{\prestar{\setcomplement{\prestar{\prodconfigsproj{G}}}}})\).

\end{proof}

Let \(\psi\) be a HyperLTL formula over \(\transitions\), we write \(\satset{\psi} \deff \set{\config \in \configs \mid \config \models \psi}\).

\begin{lemma}
	\label{thm:full-k-blind}
	Let \(\psi= \aquantif_1 \run_1.
	\ldots \aquantif_k \run_k. \varphi\) be a monadic HyperLTL\(\setminus \ltlnext\) formula.
	Set \(\satset{\psi}\) is \(K\)-blind for \(K\) doubly-exponential in \(\size{\prot}\) and \(\size{\varphi}\).
\end{lemma}
\begin{proof}
	We show \(K\)-blindness where \(K\) is the bound obtained when applying \cref{thm:ltl-k-blind} on \(\prot\) and on a formula of size linear in \(\size{\varphi}\).
	Hence, the bound does not depend on the number of quantifiers of \(\psi\).
	We proceed by induction on the number of quantifiers \(k \geq 1\).
	The base case \(k=1\) is proved by \cref{thm:ltl-k-blind}.
	Let \(k \geq 2\); suppose that the result holds for any monadic HyperLTL formula with \(k-1\) quantifiers.
	Let \(\psi = \aquantif_1 \run_1.
	\aquantif_2 \run_2. \ldots \aquantif_k \run_k. \varphi\) with \(\varphi\) described as a Boolean combination of \(\varphi_1\) to \(\varphi_n\), each referring to a single run variable.
	Note that \(\satset{\psi} = \configs \setminus \satset{\hyperneg{\psi}}\), where \(\hyperneg{\psi}\) is the formula obtained from \(\psi\) by transforming \(\forall\) quantifiers into \(\exists\) and vice versa, and by replacing the inner formula \(\varphi\) by \(\neg \varphi\).
	Therefore, we may assume that \(\aquantif_1 = \exists\).

	Suppose w.l.o.g.\ that \(\varphi_1\) to \(\varphi_m\) are the formulas that refer to \(\run_1\).
	For every valuation \(\nu: \nset{1}{m} \rightarrow \set{\true,\false}\), let \(\mathsf{Ev}_{\nu} \deff \bigwedge_{i=1}^m \varphi_i(\atrace) \Leftrightarrow \nu(i)\); note that \(\mathsf{Ev}_{\nu}(\run)\) only has run variable \(\run\).
	Let \(\simplify{\varphi}{\nu}\) denote the formula \(\varphi\) simplified assuming that, for all \(i \in \nset{1}{m}\), \(\varphi_i\) has truth value \(\nu(i)\).
	Note that \(\run_1\) does not appear in \(\simplify{\varphi}{\nu}\).
	Let \(\psi_\nu \deff \aquantif_2 \atrace_2 \, \ldots \aquantif_k \atrace_k .
	\, \simplify{\varphi}{\nu}\).
	Let \(\config \in \configs\); \(\config \in \satset{\exists \run_1 .
		\, \mathsf{Ev}_{\nu}}\) is equivalent to the existence of \(\run_1 \in \fairrunsfrom{\config}\) such that, for all \(i \in \nset{1}{m}\), \(\run_1 \models \varphi_i\) iff \(\nu(i)\) is true.
	In words, \(\config \in \satset{\exists \run_1 .
		\, \mathsf{Ev}_{\nu}}\) whenever there is \(\run_1 \in \fairrunsfrom{\config}\) that yields valuation \(\nu\).
	Also, \(\psi_\nu\) corresponds to \(\psi\) simplified under the assumption that run variable \(\run_1\) yields valuation \(\nu\); run variable \(\run_1\) does not appear in \(\psi_\nu\) and \(\psi_\nu\) does not need quantifier \(\aquantif_1\).
	We deduce that \(\satset{\psi} = \bigcup_{\nu: \nset{1}{m} \rightarrow \set{\true,\false}} \satset{\exists \run_1 .
		\, \mathsf{Ev}_{\nu}} \cap
	\satset{\psi_\nu}.\)

	For every \(\nu\), \(\psi_\nu\) only has \(k-1\) quantifiers; by induction hypothesis, \(\satset{\psi_\nu}\) is \(K\)-blind.
	This also holds for \(\satset{\exists \run_1.
		\, \mathsf{Ev}_{\nu}}\) because \(\mathsf{Ev}_{\nu}\) has size at most linear in \( \size{\varphi}\).
	Thanks to \cref{lm:k-blind-boolean}, we obtain that \(\satset{\psi}\) is \(K\)-blind.

\end{proof}

\begin{theorem}
	\label{thm:hyperltl-iopp}
	Verification of monadic HyperLTL\(\setminus \ltlnext\) for immediate observation population protocols is in 2-\EXPSPACE.
\end{theorem}
\begin{proof}
	Again, we present a non-deterministic procedure.
	Let \(\psi\) be a HyperLTL formula; as in the proof of \cref{thm:ltl-iopp}, we may consider the existential case only, where one asks whether \(\satset{\psi} \cap \initialconfigs \ne \emptyset\).
	By \cref{thm:full-k-blind} and \cref{lm:k-blind-boolean}, \(\satset{\psi} \cap \initialconfigs\) is \(K\)-blind for some doubly-exponential \(K\), so that \(\satset{\psi} \cap \initialconfigs \ne \emptyset\) if and only if there is \(\config \in \satset{\psi} \cap \initialconfigs \cap \configs_{\leq K}\) where \(\configs_{\leq K} \deff \set{\config \mid \forall q, \, \config(q) \leq K}\).
	We guess such a \(\config \in \configs_{\leq K}\).
	We can write \(\config\) in binary, and thus in exponential space.
	It is easy to check that \(\config \in \initial\).
	Let \(\psi = \aquantif_1 \run_1.
	\ldots \aquantif_k \run_k. \varphi\) with \(\varphi\)  described as a Boolean combination of \(\varphi_1\) to \(\varphi_n\), each referring to a single run variable.
	For each \(j \in \nset{1}{k}\), let \(\ell_j\) be the number of \(\varphi_i\) that refer to run variable \(\run_j\).
	From the proof of \cref{thm:ltl-iopp}, we can compute a \emph{simple expression} for \(\satset{\psi}\) in the form of a Boolean combination of \emph{elementary sets} of the form \(\satset{\exists \run.
		\, \varphi'}\).
	Moreover, with a straightforward induction, this simple expression is composed of at most \(O(2^{\ell_1 + \dots + \ell_k})\) elementary sets, because the union over the possible valuations has \(2^{\ell_j}\) disjuncts during induction step \(j\); also, each elementary set formula has size linear in \( \size{\varphi}\).
	We compute, in exponential time, this simple expression.
	We check if \(\config \in \satset{\psi}\) by evaluating membership of \(\config\) in each elementary set with \cref{thm:ltl-iopp} using doubly-exponential space, and then evaluating the simple expression.

\end{proof}

%
%
\section{A Structural Bound in Product Systems}
\label{sec:structural-bounds}

This section is devoted to proving \cref{thm:k-blind-post}.
We rely on the theory of well-quasi-orders (see, \emph{e.g.}, \cite{WQOlecturenotes}).
A \AP ""quasi-order"" is a set equipped with a transitive and symmetric relation.
In a "quasi-order" \((E,\lewqo)\), a set \(S \subseteq E\) is \AP ""upward-closed"" (resp.
""downward-closed"") when, for all \(s \in S\), for all \(t \in E\), if \(s \lewqo t\) then \(t \in S\) (resp.\ if \(t \lewqo s\) then \(s \in S\)); also, \(\upwardclosure{S} \deff \set{t \in E \mid \exists s \in S, s \lewqo t}\) is its ""upward-closure"" and \(\downwardclosure{S} \deff \set{t \in E \mid \exists s \in S, t \lewqo s}\) its  ""downward-closure"".
A \AP ""well-quasi-order"" is a "quasi-order" \((E,\lewqo)\) such that, for every infinite sequence \((x_i)_{i \in \nats}\) of elements of \(E\), there is \(i <j\) such that \(x_i \lewqo x_j\).
In a "well-quasi-order" \((E, \lewqo)\), any "upward-closed" set \(S\) has a finite set of minimal elements \(\minwqo{S}\), and \(S = \upwardclosure{\minwqo{S}}\).
\subsection{Transfer Flows}
We fix a product system \(\prodsystem = (\prot, \buchiaut)\) with \(\prot =: (\states, \transitions, \initialstates)\) and \(\buchiaut =: (\bstates, \btransitions, \binitstate, \bwinning)\).
We prove \cref{thm:k-blind-post} using \emph{transfer flows}, an abstraction representing the possibilities offered by sequences of transitions.
Let \(\natsanddummy \deff \nats \cup \set{\dummysymb}\); we extend \((\nats,\leq)\) to \((\natsanddummy,\leq)\) where \(\dummysymb\) is incomparable with integers: for all \(x \in \natsanddummy\), \(x \sim \dummysymb\) iff \(x = \dummysymb\) for \(\sim \, \in \set{\leq,\geq}\).
We extend addition by \(\dummysymb + x = x\) for all \(x \in \natsanddummy\).

\begin{definition}
	A ""transfer flow"" is a triplet \(\atf = (f,\abuchistate,\abuchistate')\) where \(f : \states^2 \to \natsanddummy\) and \(\abuchistate, \abuchistate' \in \bstates\).
	We denote by \(\transferflows\) the set of all "transfer flows".
\end{definition}

Intuitively, \((f,\abuchistate,\abuchistate')\) represents possible finite runs of \(\prodsystem\), with \(f\) the transfer of agents in \(\PP\) and \(\abuchistate,\abuchistate'\) the start and end states in \(\buchiaut\).
Having \(f(q_1, q_2) = \dummysymb\) represents the impossibility to send agents from \(q_1\) to \(q_2\), while \(f(q_1,q_2) = n\) represents the need to send at least \(n\) agents from \(q_1\) to \(q_2\); in this case, any number in \(\nsetinfinity{n}\) can be sent.
The values \(\abuchistate,\abuchistate'\) are called the ""control part@@tf"" of \(\atf\), while the function \(f\) is called the ""agent part@@tf"" of \(\atf\).
Given a transfer flow \(\atf =(f,\abuchistate,\abuchistate') \in \transferflows\), we define its \emph{weight} by \(\tfweight{\atf} \deff \sum_{q,q'} f(q,q')\).

We define a partial order \(\letf\) on \(\transferflows\) as follows.
For \(\atf_1 = (f_1,\bstate_1, \bstate_1')\) and \(\atf_2 = (f_2,\bstate_2,\bstate_2')\), we let \(\atf_1 \letf \atf_2\) when \(\bstate_1 = \bstate_1'\), \(\bstate_2 = \bstate_2'\) and, for all \(q,q'\), \(f_1(q,q') \leq f_2(q,q')\).
In particular, this requires that, for all \(q,q'\), \(f_1(q,q') = \dummysymb\) if and only if \(f_2(q,q') = \dummysymb\).
It is easy to see that \((\transferflows,\letf)\) is a well-quasi-order.
We highlight the following rule of thumb: \emph{smaller transfer flows are more powerful.
}
Indeed, when \(\atf_1 \letf \atf_2\), for \(q,q'\) such that \(f_1(q,q'), f_2(q,q') \ne\dummysymb\), \(f_1(q,q') \leq f_2(q,q')\): \(\atf_1\) allows to send from \(q\) to \(q'\) any number of agents in \(\nsetinfinity{f_1(q,q')}\) while \(\atf_2\) allows to send from \(q\) to \(q'\) any number of agents in \(\nsetinfinity{f_2(q,q')} \subseteq \nsetinfinity{f_1(q,q')}\).

\begin{definition}
	Given \(\prodconfig_1 = (\config_1, \bstate_1), \prodconfig_2 = (\config_2, \bstate_2)\in \prodconfigs\) and \(\atf = (f,\bstate,\bstate') \in \transferflows\), we let \(\prodconfig_1 \flowstep{\atf} \prodconfig_2\) when \(\bstate_1 = \bstate\), \(\bstate_2 = \bstate'\) and there is a \emph{step witness} \(g: \states^2 \to \natsanddummy\) such that \(f(q,q') \leq g(q,q')\) for all \(q,q'\in\states\), \( \config_1(q) = \sum_{q'} g(q,q')\) for all \(q\in\states\) and \( \config_2(q) = \sum_{q'} g(q',q)\) for all \(q\in\states\).
\end{definition}
Note that if \(\prodconfig_1 \flowstep{\atf} \prodconfig_2\), then \(\prodconfig_1 \flowstep{\atf'} \prodconfig_2\) for all \(\atf' \letf \atf\): again, smaller "transfer flows" are more powerful.
Intuitively, \(g\) corresponds to a transfer of agents in \(\prodsystem\) concretizing \(\prodconfig_1 \flowstep{\atf} \prodconfig_2\).
We now build transfer flows corresponding to transitions of \(\prodsystem\).
For each \(\atrans = (q_1, q_2) \trans{} (q_1, q_3) \in \transitions\), we define the set \(\transfersof{\atrans} \subseteq \transferflows\) that contains all "transfer flows" \((f,\bstate, \bstate')\) such that \(\btransitions(\bstate, \atrans) = \bstate'\) and:
\begin{itemize}[noitemsep,topsep=0pt,parsep=0pt,partopsep=0pt]
	\item if \(q_1 \ne q_2\) or \(q_1 \ne q_3\) then \(f(q_1,q_1) \geq 1, f(q_2,q_3) \geq 1\);
	\item if \(q_1 = q_2 = q_3\) then \(f(q_1,q_1) \geq 2\);
	\item for all \(q \ne q_1\) such that \((q,q) \ne (q_2,q_3)\), \(f(q,q) \geq 0\);
	\item for all \(q\ne q'\) such that \((q,q') \ne (q_2,q_3)\), \(f(q,q') = \dummysymb\).
\end{itemize}
That is, at least one agent is in \(q_1\), some agents are sent from \(q_2\) to \(q_3\) and the control part is changed according to \(t\).
The set \(\transfersof{\atrans}\) is upward-closed with respect to \(\mathord{\letf}\): the number of agents going from \(q_2\) to \(q_3\) can be arbitrarily large, which corresponds to an accelerated step of \(\prot\) using transition \(t\).

\begin{lemmarep}
	\label{flow-one-trans}
	For all \(\prodconfig, \prodconfig' \in \prodconfigs\), \(\atrans \in \transitions\), \(\prodconfig \step{\atrans} \prodconfig'\) iff there is \(\atf \in \transfersof{\atrans}\) s.t.
	\(\prodconfig \flowstep{\atf} \prodconfig'\).
\end{lemmarep}
\begin{appendixproof}
	Let \((q_1,q_2) \trans{\atrans} (q_1,q_3)\) denote transition \(\atrans\).
	Also, let \(\prodconfig=: (\config, \bstate)\) and \(\prodconfig' =: (\config, \bstate')\).
	First, observe that, if \(\btransitions(\bstate, \atrans) \ne \bstate'\) then both statements are false; we now consider that \(\btransitions(\bstate, \atrans) = \bstate'\).
	We start by treating the case \(q_2 = q_3\).
	In this case, we have \(\config = \config'\), and \(\prodconfig \accstep{\atrans} \prodconfig'\) if and only if \(\btransitions(\bstate, \atrans) = \bstate'\), \(\config(q_1) \geq 1\) and \(\config(q_2) \geq 1\) (\(\config(q_1) \geq 2\) if \(q_1 = q_2\)), which is equivalent to \(\prodconfig \flowstep{\atf} \prodconfig'\) with \(\atf = (f,\bstate,\bstate')\) the minimal element of \(\transfersof{\atrans}\), \emph{i.e.}, the one such that \(f(q_1,q_1) = 1\) and \(f(q_2,q_2) = 1\) (\(f(q_1,q_1) = 2\) if \(q_1 = q_2\)).

	We now assume that \(q_2 \ne q_3\).
	First, assume that \(\prodconfig \accstep{\atrans} \prodconfig'\); by definition of the semantics of the product system, there exists \(k \geq 1\) such that \(\prodconfig \step{\atrans^k} \prodconfig'\).
	Because \(q_2 \ne q_3\), we have \(k \leq \config(q_2)\).
	Let \(n \deff \size{\config} = \size{\config'}\).
	We define \(f : \states^2 \to \natsanddummy\) as follows.
	We let \(f(q_2,q_3) := k\), \(f(q_2,q_2) := \config(q_2) -k\), \(f(q,q) := \config(q)\) for all \(q \ne q_2\) and \(f(q,q') := \dummysymb\) otherwise.
	We have \(\atf := (f,\bstate, \bstate') \in \transfersof{\atrans}\), indeed: \(k \geq 1\); \(\config(q_1) \geq 1\) so that \(f(q_1,q_1) \geq 1\); \(\config(q_2) \geq k\) so that \(f(q_2,q_2) \geq 0\); if \(q_1 =q_2\), \(\config(q_1) \geq k+1\) so that \(f(q_1,q_1) \geq 1\).
	Moreover, we have \(\prodconfig \flowstep{\atf} \prodconfig'\), as it suffices to consider \(g = f\) as witness (and the control parts match).

	Conversely, assume that there is \(\atf = (f,\bstate, \bstate') \in \transfersof{\atrans}\) such that \(\prodconfig \flowstep{\atrans} \prodconfig'\).
	Let \(g \geq f\) be a witness that \(\prodconfig \flowstep{\atrans} \prodconfig'\); let \(k \deff g(q_2,q_3) \geq f(q_2,q_3) \geq 1\).
	We claim that \(\prodconfig \step{\atrans^k} \prodconfig'\).
	We have that, for all \(q,q'\) such that \(q\ne q'\) and \((q,q') \ne (q_2,q_3)\), \(f(q,q') = \dummysymb\) hence \(g(q,q') = \dummysymb\), so that \(\config(q) = \config'(q)\) for all \(q \notin \set{q_2,q_3}\).
	Also, we have \(\config'(q_1) \geq 1\) because \(g(q_1,q_1) \geq f(q_1,q_1) \geq 1\).
	Moreover, if \(q_1 =q_2\) then \(\config'(q_1) \geq g(q_1,q_1) + g(q_2,q_3) \geq k+1\), so that firing \(t\) the first \(k-1\) times from \(\config\) leaves at least two agents on \(q_1\) which allows to fire \(\atrans\) once more.
	Finally, we have \(\config'(q_3) - \config(q_3) = f(q_2,q_3) = k\) and \(\config(q_2) - \config'(q_2) = f(q_2,q_3) = k\).
	This proves that \(\config \step{\atrans^k} \config'\) in \(\prot\); because the control parts match, we conclude that \(\prodconfig \step{\atrans} \prodconfig'\) in the product system.
\end{appendixproof}

We define the product set \(\atf_1 \tftimes \atf_2 \subseteq \transferflows\) of two "transfer flows".
This set is meant to encode the possibilities given by using \(\atf_1\) followed by \(\atf_2\).
Let \(\atf_1 = (f_1,\bstate_1,\bstate_1'), \atf_2 = (f_2,\bstate_2,\bstate_2') \in \transferflows\).
If \(\bstate_1' \ne \bstate_2\), then we set \(\atf_1 \tftimes \atf_2 = \emptyset\) .
Assume now \(\bstate_1' = \bstate_2\).
The set \(\atf1 \tftimes \atf_2\) contains all "transfer flows" of the form \((h,\bstate_1,\bstate_2')\) for which there is a \emph{product witness}
\(H : \states^3 \to \natsanddummy\) such that:
\begin{enumerate}[label = (prod.\roman*), leftmargin = 2cm,topsep=3pt]
	\item \label{composition-condition1} for all \((q_1,q_3)\), \(\sum_{q_2}
	      H(q_1,q_2,q_3) = h(q_1,q_3)\); \item \label{composition-condition2} for all \((q_1, q_2)\), \(\sum_{q_3} H(q_1,q_2,q_3) \geq f_1(q_1,q_2)\); \item \label{composition-condition3} for all \((q_2,q_3)\), \(\sum_{q_1} H(q_1,q_2,q_3) \geq f_2(q_2,q_3)\).
\end{enumerate}
In particular, for all \(q_1,q_2\), \(f_1(q_1,q_2) = \dummysymb\) if and only if, for all \(q_3\), \(H(q_1,q_2,q_3) = \dummysymb\).
Similarly, \(f_2(q_2,q_3) = \dummysymb\) if and only if, for all \(q_1\), \(H(q_1,q_2,q_3) = \dummysymb\).
We extend \(\tftimes\) to sets of "transfer flows": for \(F, F' \subseteq \transferflows\), \(\tfset \tftimes \tfset' \deff \bigcup_{\atf \in \tfset, \atf' \in \tfset'} \atf \tftimes \atf'\).

\begin{lemmarep}%
	\label{prod-upward-closed}
	Let \(\atf_1, \atf_2, \atf_3 \in \transferflows\).
	We have the following properties:
	\begin{enumerate}[label = "\upshape(\ref{prod-upward-closed}.\roman*)", leftmargin = 1cm,topsep=3pt]
		\item \label{basic1} the set \(\atf_1 \tftimes \atf_2\) is upward-closed with respect to \(\letf\);
		\item \label{basic2} for all \(\atf_1' \letf \atf_1\) and \(\atf_2' \letf \atf_2\), \(\atf_1 \tftimes \atf_2 \subseteq \atf_1' \tftimes \atf_2'\);
		\item \label{basic3} \(\tftimes\) is associative: \((\atf_1 \tftimes \atf_2) \tftimes \atf_3 = \atf_1 \tftimes (\atf_2 \tftimes \atf_3)\);
		\item \label{basic4} for every \(\atf \in \minwqo{\atf_1 \tftimes \atf_2}\), \(\tfweight{\atf} \leq \tfweight{\atf_1} + \tfweight{\atf_2}\).
	\end{enumerate}
\end{lemmarep}

\begin{appendixproof}
	\textbf{Proof of \ref{basic1}.}
	Let \(\atf = (h,\controlloc, \controlloc') \in \atf_1 \tftimes \atf_2\).
	We apply the definition to obtain a product witness \(H : \states^3 \to \natsanddummy\).
	Let \(\atf' = (h', \controlloc, \controlloc')\in \transferflows\) such that \(\atf \letf \atf'\).
	This implies that \(h \leq h'\).
	We define \(H' : \states^3 \to \natsanddummy\) as follows.
	Let \(q_1, q_3 \in \states\).
	If we have \(H(q_1,q_2,q_3) = \dummysymb\) for all \(q_2\) then we set \(H'(q_1,q_2,q_3) \deff \dummysymb\) for all \(q_2\).
	Suppose now that there is \(\tilde{q_2}\) such that \(H(q_1,\tilde{q_2},q_3) \ne \dummysymb\).
	This in particular implies, by \ref{composition-condition1}, that \(h(q_1,q_3) \ne \dummysymb\) therefore \(h'(q_1,q_3) \ne \dummysymb\).
	We set \(H'(q_1,\tilde{q_2},q_3) \deff H(q_1,\tilde{q_2},q_3) + h'(q_1,q_3) - h(q_1,q_3)\), and we set \(H'(q_1,q_2,q_3) = H(q_1,q_2,q_3)\) for every \(q_2 \ne \tilde{q_2}\).
	We claim that \(H'\) is a product witness that \((h',\controlloc,\controlloc') \in \atf_1 \tftimes \atf_2\).
	First, \(H'\) has the same \(\dummysymb\) values as \(H\), so that we have \(H' \geq H\) by construction.
	Note that \(H' \geq H\) requires that they have the same \(\dummysymb\) values (\(\dummysymb\) is incomparable with all integers), which would not hold if we had set \(H'(q_1,q_2,q_3) = H(q_1,q_2,q_3) + h'(q_1,q_3) - h(q_1,q_3)\) for some \(q_1, q_2\) and \(q_3\) such that \(h'(q_1,q_3) \ne \dummysymb\) but \(H(q_1,q_2,q_3) = \dummysymb\).
	We therefore have \(H'\) satisfies \ref{composition-condition2} and \ref{composition-condition3}.
	Also, for all \(q_1,q_3\), if \(h'(q_1,q_3) = \dummysymb\) then \(\sum_{q_2} H'(q_1,q_2,q_3) = \sum_{q_2} H(q_1,q_2,q_3) = h(q_1,q_3) = \dummysymb\).
	If \(h(q_1,q_3) \ne \dummysymb\) then \(\sum_{q_2} H'(q_1,q_2,q_3) = \sum_{q_2} H(q_1,q_2,q_3) + h'(q_1,q_3) - h(q_1,q_3) = h(q_1,q_3) + h'(q_1,q_3) - h(q_1,q_3) = h'(q_1,q_3)\).
	We have proved that \(H'\) satisfies \cref{composition-condition1} for \(h'\), so that \((h', \controlloc, \controlloc') \in \atf_1 \tftimes \atf_2\).

	\textbf{Proof of \ref{basic2}.}
	Let \(\atf = (h,\controlloc, \controlloc') \in \atf_1 \tftimes \atf_2\).
	We apply the definition to obtain a product witness \(H : \states^3 \to \natsanddummy\).
	Assume that we have \(h' : \states^2 \to \natsanddummy\) such that \(h \leq h'\).
	We increase the values of \(H\) to obtain \(H'\) such that, for all \((q_1,q_3)\), \(\sum_{q_2} H'(q_1,q_2,q_3) = h'(q_1,q_3)\).
	Because \(h\) and \(h'\) have the same \(\dummysymb\) component, we set \(H'\) to have the same \(\dummysymb\) components as \(H\).
	Let \(q_1,q_3 \in \states\) such that \(h(q_1,q_3) < h'(q_1,q_3) \ne \dummysymb\).
	There is \(q_2\) such that \(H(q_1,q_2,q_3) \ne \dummysymb\); we arbitrarily select such a state \(q_2\).
	We simply increase \(H(q_1,q_2,q_3)\) by \(h'(q_1,q_3) - h(q_1,q_3)\).
	Increasing the values will not violate the conditions about \(f\) and \(g\).
	We apply this operation with every pair \((q_1,q_3)\) and end up with a witness that \((h, \bstate_1,\bstate_2') \in \atf_1 \tftimes \atf_2\).

	\textbf{Proof of \ref{basic3}.}
	We now prove associativity of \(\tftimes\).
	Let \(\atf_i =: (f_i, \controlloc_i, \controlloc_i')\) for all \(i \in \set{1,2,3}\).
	If we have \(\controlloc_1' \ne \controlloc_2\) or \(\controlloc_2' \ne \controlloc_3\) then \((\atf_1 \tftimes \atf_2) \tftimes \atf_3 = \atf_1 \tftimes (\atf_2 \tftimes \atf_3) = \emptyset\).
	Suppose now that \(\controlloc_1' = \controlloc_2\) and \(\controlloc_2' = \controlloc_3\).

	Let \(T_{1,2,3}\subseteq \transferflows\) denote the set of "transfer flows" \(\atf = (f, \controlloc_1, \controlloc_3')\) for which there exists a function \(H : \states^4 \to \natsanddummy\) that satisfies the following properties:
	\begin{enumerate}[topsep=3pt]
		\item \label{associtivity-1} for all \(q_1, q_4\), \(\sum_{q_2,q_3}
		      H(q_1,q_2,q_3,q_4) = f(q_1,q_4)\); \item \label{associtivity-2} for all \(q_1,q_2\), \(\sum_{q_3,q_4} H(q_1,q_2,q_3,q_4) \geq f_1(q_1,q_2)\); \item \label{associtivity-3} for all \(q_2,q_3\), \(\sum_{q_1,q_4} H(q_1,q_2,q_3,q_4) \geq f_2(q_2,q_3)\); \item \label{associtivity-4} for all \(q_3,q_4\), \(\sum_{q_1,q_2} H(q_1,q_2,q_3,q_4) \geq f_3(q_3,q_4)\).
	\end{enumerate}
	We claim that \((\atf_1 \tftimes \atf_2) \tftimes \atf_3 = T_{1,2,3} = \atf_1 \tftimes (\atf_2 \tftimes \atf_3)\).

	We first prove that \((\atf_1 \tftimes \atf_2) \tftimes \atf_3 \subseteq T_{1,2,3}\).
	Let \(\atf = (f, \controlloc_1, \controlloc_3') \in (\atf_1 \tftimes \atf_2) \tftimes \atf_3\).
	Let \(\atf_{1,2} = (f_{1,2}, \controlloc_1, \controlloc_2')\in \atf_1 \tftimes \atf_2\) such that \(\atf \in \atf_{1,2} \tftimes \atf_3\); let \(G : \states^3 \to \natsanddummy\) be a product witness of that.
	We have \(\sum_{q_4} G(q_1,q_3,q_4) \geq f_{1,2}(q_1,q_3)\) for all \(q_1,q_3\), \(\sum_{q_1} G(q_1,q_3,q_4) \geq f_3(q_3,q_4)\) for all \(q_3,q_4\) and \(\sum_{q_3} G(q_1,q_3,q_4) = f(q_1,q_4)\) for all \(q_1,q_4\).
	Let \(F_{1,2}\) be a product witness that \(\atf_{1,2} \in \atf_1 \tftimes \atf_2\), \emph{i.e.}, \(\sum_{q_2} F_{1,2}(q_1,q_2,q_3) = f_{1,2}(q_1,q_3)\) for all \(q_1,q_3\), \(\sum_{q_3} F_{1,2}(q_1,q_2,q_3) \geq f_1(q_1,q_2)\) for all \(q_1,q_2\) and \(\sum_{q_1} F_{1,2}(q_1,q_2,q_3) \geq f_2(q_2,q_3)\) for all \(q_2,q_3\).
	For every \(q_1,q_3\), \(\sum_{q_4} G(q_1,q_3,q_4) \geq f_{1,2}(q_1,q_3) = \sum_{q_2} F_{1,2}(q_1,q_2,q_3)\).
	For all \(q_1,q_3\), if \(f_{1,2}(q_1,q_3) \ne \dummysymb\) then there is \(\tilde{q_2}\) such that \(F_{1,2}(q_1,\tilde{q_2},q_3) \ne \dummysymb\).
	Let \(F\) equal to \(F_{1,2}\) except that, for all \(q_1,q_3\) such that \(f_{1,2}(q_1,q_3) \ne \dummysymb\), we choose \(\tilde{q_2}\) such that \(F_{1,2}(q_1,\tilde{q_2},q_3) \ne \dummysymb\) and set \(F(q_1,\tilde{q_2}, q_3) \deff F_{1,2}(q_1,\tilde{q_2},q_3)+\sum_{q_4} G(q_1,q_3,q_4) - f_{1,2}(q_1,q_3)\).
	This way, \(F\) satisfies the same conditions as \(F_{1,2}\) related to \(f_1\) and \(f_2\) but also, for all \(q_1,q_3\), \(\sum_{q_2} F(q_1,q_2,q_3) = \sum_{q_4} G(q_1,q_3,q_4)\).
	To provide \(H\) that satisfies the conditions above, it suffices to build \(H: \states^4 \to \natsanddummy\) so that \(\sum_{q_4} H(q_1,q_2,q_3,q_4) = F(q_1,q_2,q_3)\) and \(\sum_{q_2} H(q_1,q_2,q_3,q_4) = G(q_1,q_3,q_4)\).
	Indeed, this would imply conditions \ref{associtivity-1} and \ref{associtivity-4} thanks to \(F\) and conditions \ref{associtivity-2} and \ref{associtivity-3} thanks to \(G\).

	We now prove the following statement:
	\begin{quote}
		For every \(F: \states^3 \to \natsanddummy\) and \(G: \states^3 \to \natsanddummy\), if \(\sum_{q_2} F(q_1,q_2,q_3) = \sum_{q_4} G(q_1,q_3,q_4)\) for every \(q_1,q_3\), then there is \(H: \states^4 \to \natsanddummy\) such that \(\sum_{q_4} H(q_1,q_2,q_3,q_4) = F(q_1,q_2,q_3)\) and \(\sum_{q_2} H(q_1,q_2,q_3,q_4) = G(q_1,q_3,q_4)\).
	\end{quote}

	First, if \(F\) and \(G\) are constant equal to \(\dummysymb\) then we set \(H\) constant equal to \(\dummysymb\).
	Suppose now that it is not the case; let \(n \deff \sum_{q_1,q_2,q_3} F(q_1,q_2,q_3) = \sum_{q_1,q_3,q_4} G(q_1,q_3,q_4) \in \nats\).
	We proceed by induction on \(n\).

	If \(n=0\) then all values in \(F\) and \(G\) are in \(\set{0,\dummysymb}\).
	We let \(H(q_1,q_2,q_3,q_4) \deff 0\) whenever both \(F(q_1,q_2,q_3) =0\) and \(G(q_1,q_3,q_4) =0\), and \(H(q_1,q_2,q_3,q_4) \deff \dummysymb\) otherwise.
	We claim that, for all \(q_1,q_2,q_3\), \(\sum_{q_4} H(q_1,q_2,q_3,q_4) = F(q_1,q_2,q_3)\).
	Let \(q_1,q_2,q_3 \in \states\); if \(F(q_1,q_2,q_3) = \dummysymb\) then \(H(q_1,q_2,q_3,q_4) = \dummysymb\) for all \(q_4\) hence \(\sum_{q_4} H(q_1,q_2,q_3,q_4) = \dummysymb\).
	Suppose now that \(F(q_1,q_2,q_3) = 0\).
	This implies \(\sum_{q_4} G(q_1,q_3,q_4) = 0\) therefore there is \(\tilde{q_4}\) such that \(G(q_1,q_3,\tilde{q_4}) =0\), so that \(H(q_1,q_2,q_3,\tilde{q_4})=0\) and \(\sum_{q_4} H(q_1,q_2,q_3,q_4) = 0\).
	Similarly, for every \(q_1,q_3,q_4\), if \(G(q_1,q_3,q_4) = \dummysymb\) then \(\sum_{q_2} H(q_1,q_2,q_3,q_4) = \dummysymb\) and if \(G(q_1,q_3,q_4) = 0\) then there is \(\tilde{q_2}\) such that \(F(q_1,\tilde{q_2},q_3) =0\) hence \(H(q_1,\tilde{q_2},q_3,q_4) =0\) and \(\sum_{q_2} H(q_1,q_2,q_3,q_4) =0\).

	Suppose now that \(n>0\).
	Let \(\tilde{q_1}, \tilde{q_3}\) such that \(\sum_{q_2} F(\tilde{q_1}, q_2,\tilde{q_3}) = \linebreak \sum_{q_4} G(\tilde{q_1},\tilde{q_3},q_4)>0\).
	Let \(\tilde{q_2}\) such that \(F(\tilde{q_1}, \tilde{q_2},\tilde{q_3})>0\) and \(\tilde{q_4}\) such that \(G(\tilde{q_1},\tilde{q_3},\tilde{q_4})>0\).
	Let \(F'\) equal to \(F\) except that \(F'(\tilde{q_1}, \tilde{q_2},\tilde{q_3}) \deff F(\tilde{q_1}, \tilde{q_2},\tilde{q_3})-1\) and let \(G'\) equal to \(G\) except that \(G'(\tilde{q_1},\tilde{q_3},\tilde{q_4}) \deff G(\tilde{q_1},\tilde{q_3},\tilde{q_4})-1\).
	We have \(\sum_{q_2} F'(q_1,q_2,q_3) = \sum_{q_4} G'(q_1,q_3,q_4)\) for all \(q_1\) and \(q_3\), and \(\sum_{q_1,q_2,q_3} F'(q_1,q_2,q_3) = \linebreak \sum_{q_1,q_2,q_3} F(q_1,q_2,q_3)-1 = n-1\).
	We apply the induction hypothesis on \(F'\) and \(G'\) to obtain \(H'\) such that \(\sum_{q_4} H'(q_1,q_2,q_3,q_4) = F'(q_1,q_2,q_3)\) for all \(q_1,q_2,q_3\) and \(\sum_{q_2} H'(q_1,q_2,q_3,q_4) = G'(q_1,q_3,q_4)\) for all \(q_1,q_3,q_4\).
	It suffices to let \(H\) equal to \(H'\) except that \(H(\tilde{q_1}, \tilde{q_2}, \tilde{q_3},\tilde{q_4}) = H'(\tilde{q_1}, \tilde{q_2}, \tilde{q_3},\tilde{q_4})+1\).
	Note that it could be that \(H'(\tilde{q_1}, \tilde{q_2}, \tilde{q_3},\tilde{q_4}) = \dummysymb\), in which case \(H(\tilde{q_1}, \tilde{q_2}, \tilde{q_3},\tilde{q_4}) = 1\).
	We know that \(F'(\tilde{q_1}, \tilde{q_2}, \tilde{q_3}) \ne \dummysymb\) therefore \(\sum_{q_4} H'(\tilde{q_1}, \tilde{q_2}, \tilde{q_3},q_4) \ne \dummysymb\) so that we indeed have \(\sum_{q_4} H(\tilde{q_1}, \tilde{q_2}, \tilde{q_3},q_4) = F'(\tilde{q_1}, \tilde{q_2}, \tilde{q_3})+1 = F(\tilde{q_1}, \tilde{q_2}, \tilde{q_3})\).
	With the same argument, \(\sum_{q_2} H'(\tilde{q_1}, q_2, \tilde{q_3},\tilde{q_4}) = G(\tilde{q_1}, \tilde{q_3},\tilde{q_4})\).
	This concludes the induction.

	We have proved that \((\atf_1 \tftimes \atf_2) \tftimes \atf_3 \subseteq T_{1,2,3}\).
	The fact that \(\atf_1 \tftimes (\atf_2 \tftimes \atf_3) \subseteq T_{1,2,3}\) follows by a symmetric argument.
	We claim that \(T_{1,2,3} \subseteq (\atf_1 \tftimes \atf_2) \tftimes \atf_3\).
	Indeed, let \(\atf \in T_{1,2,3}\) and let \(H: \states^4 \to \natsanddummy\) that satisfies conditions \ref{associtivity-1} to \ref{associtivity-4} for \(\atf\).
	Let \(f : (q_1,q_3) \mapsto \sum_{q_3,q_4} H(q_1,q_2,q_3,q_4)\), we have \((f,\controlloc_1,\controlloc_2') \in \atf_1 \tftimes \atf_2\) with \(F: (q_1,q_2,q_3) \mapsto \sum_{q_4} H(q_1,q_2,q_3,q_4)\) as product witness.
	Moreover, let \(g: (q_3,q_4) \mapsto \sum_{q_1,q_2} H(q_1,q_2,q_3,q_4)\); we have \(\atf_3 \letf (g,\controlloc_3, \controlloc_3')\).
	Finally, we have \(\atf \in (f,\controlloc_1, \controlloc_2') \tftimes (g,\controlloc_3, \controlloc_3')\) with \((q_1,q_3,q_4) \mapsto \sum_{q_2} H(q_1,q_2,q_3,q_4)\) as product witness, hence by \ref{basic2} we conclude that \(\atf \in (\atf_1 \tftimes \atf_2) \tftimes \atf_3\).
	This proves that \(T_{1,2,3} \subseteq (\atf_1 \tftimes \atf_2) \tftimes \atf_3\); a symmetric argument proves that \(T_{1,2,3} \subseteq \atf_1 \tftimes (\atf_2 \tftimes \atf_3)\).
	In the end, we obtain \((\atf_1 \tftimes \atf_2) \tftimes \atf_3 = \atf_1 \tftimes (\atf_2 \tftimes \atf_3) = T_{1,2,3}\).

	\textbf{Proof of \ref{basic4}.}
	Let \(\atf_1 = (f_1, \controlloc_1, \controlloc_2)\), \(\atf_2 = (f_2, \controlloc_2, \controlloc_3)\) and \(\atf= (f,\controlloc_1, \controlloc_3) \in \minwqo{\atf_1 \tftimes \atf_2}\); let \(H : \states^3 \to \natsanddummy\) be a product witness that \(\atf \in \atf_1 \tftimes \atf_2\).
	We know that \(\tfweight{\atf} = \sum_{q_1,q_3} f(q_1,q_3) = \sum_{q_1,q_2,q_3} H(q_1,q_2,q_3)\).
	We thus prove that \(\sum_{q_1,q_2,q_3} H(q_1,q_2,q_3) \leq \tfweight{\atf_1} + \tfweight{\atf_2}\).
	Suppose by contradiction that \(\sum_{q_1,q_2,q_3} H(q_1,q_2,q_3) > \tfweight{\atf_1} + \tfweight{\atf_2}\).
	We claim that there is \(H': \states^3 \to \natsanddummy\) such that:
	\begin{itemize}[noitemsep,topsep=0pt,parsep=0pt,partopsep=0pt]
		\item \(H' \leq H\),
		\item \(\sum_{q_1,q_2,q_3}
		      H'(q_1,q_2,q_3) < \sum_{q_1,q_2,q_3} H(q_1,q_2,q_3)\), \item \(\sum_{q_3} H'(q_1,q_2) \geq f_1(q_1,q_2)\) for all \(q_1,q_2\), \item \(\sum_{q_1} H'(q_1,q_2,q_3) \geq f_{2}(q_2,q_3)\) for all \(q_2,q_3\).
	\end{itemize}
	Indeed, if we have such a function \(H'\), then letting \(f' : (q_1,q_2) \mapsto H'(q_1,q_2,q_3)\), we would have \((f', \controlloc_1, \controlloc_3) \in \atf_1 \tftimes \atf_2\) and \((f', \controlloc_1, \controlloc_3) \letf \atf\), contradicting minimality of \(\atf\) in \(\atf_1 \tftimes \atf_2\).

	To build \(H'\), it suffices to prove the existence of \(\tilde{q_1},\tilde{q_2},\tilde{q_3}\) such that \linebreak\(\sum_{q_3} H(\tilde{q_1},\tilde{q_2}, q_3) > f_1(\tilde{q_1},\tilde{q_2})\) and \(\sum_{q_1} H(q_1,\tilde{q_2},\tilde{q_3}) > f_{2}(\tilde{q_2}, \tilde{q_3})\), so that we can set \(H'\) equal to \(H\) except that \(H'(\tilde{q_1},\tilde{q_2},\tilde{q_3}) = H(\tilde{q_1},\tilde{q_2},\tilde{q_3}) -1\).

	To find \(\tilde{q_1},\tilde{q_2}\) and \(\tilde{q_3}\), we prove the following statement:
	\begin{quote}
		For all \(h:\states^3 \to \natsanddummy\), \(g_1: \states^2 \to \natsanddummy\) and \(g_2: \states^2 \to \natsanddummy\) such that \(\sum_{q_3} h(q_1,q_2,q_3) \geq g_1(q_1,q_2)\) for all \(q_1,q_2\), \(\sum_{q_1} h(q_1,q_2,q_3) \geq g_2(q_2,q_3)\) for all \(q_2,q_3\) and \\\(\sum_{q_1,q_2,q_3} h(q_1,q_2,q_3) > \sum_{q_1,q_2} g_1(q_1,q_2) + \sum_{q_2,q_3} g_2(q_2,q_3)\), there are \(\tilde{q_1}, \tilde{q_2}\) and \(\tilde{q_3}\) such that \(\sum_{q_3} h(\tilde{q_1},\tilde{q_2},q_3) > g_1(\tilde{q_1},\tilde{q_2})\) and \(\sum_{q_1} h(q_1,\tilde{q_2}, \tilde{q_3}) > g_2(\tilde{q_2},\tilde{q_3})\).
	\end{quote}
	The proof is by induction on \(\sum_{q_1,q_2,q_3} h(q_1,q_2,q_3)\).
	The base case is when \linebreak \(\sum_{q_1,q_2,q_3} h(q_1,q_2,q_3) \allowbreak = 1\) and \(g_1\) and \(g_2\) only have value \(0\) and \(\dummysymb\), in which case it suffices to take \(\tilde{q_1},\tilde{q_2}, \tilde{q_3}\) such that \(h(\tilde{q_1}, \tilde{q_2}, \tilde{q_3}) = 1\).
	For the induction step, let \(r_1,r_2,r_3\) such that \(h(r_1,r_2,r_3) > 0\).
	This implies that \(g_1(r_1,r_2), g_2(r_2,r_3) \in \nats\).
	If \(g_1(r_1,r_2) =0\) and \(g_2(r_2,r_3) = 0\) then we let \((\tilde{q_1}, \tilde{q_2}, \tilde{q_3}) \deff (r_1,r_2,r_3)\) and we are done.
	Assume now that \(g_1(r_1,r_2) >0\) or \(g_2(r_2,r_3)>0\).
	Let \(h'\) equal to \(h\) except that \(h'(r_1,r_2,r_3) = h(r_1,r_2,r_3) -1\); let \(g_1'\) equal to \(g_1\) except if \(g_1(r_1,r_2)>0\) in which case \(g_1'(r_1,r_2) = g_1(r_1,r_2) - 1\); let \(g_2'\) equal to \(g_2\) except if \(g_2(r_2,r_3) >0\) in which case \(g_2'(r_2,r_3) = g_2(r_2,r_3) -1\).
	For every \((q_1,q_2) \ne (r_1,r_2)\), we have \(\sum_{q_3} h'(q_1,q_2,q_3) = \sum_{q_3} h(q_1,q_2,q_3) \geq g_1(q_1,q_2) = g_1'(q_1,q_2)\).
	Moreover, if \(g_1(r_1,r_2) = 0\) then \(g_1'(r_1,r_2) =0\) and \(\sum_{q_3} h'(r_1,r_2,q_3) \in \nats\) so that \(\sum_{q_3} h'(r_1,r_2,q_3) \geq 0 = g_1'(r_1,r_2)\).
	If \(g_1(r_1,r_2) > 0\) then \(g_1'(r_1,r_2) = g_1(r_1,r_2) -1\) and \(\sum_{q_3} h'(r_1,r_2,q_3) = \sum_{q_3} h(r_1,r_2,q_3)-1 \geq g_1(r_1,r_2)-1 = g_1'(r_1,r_2)\).
	Overall, we have proved that, for all \(q_1,q_2\), \(\sum_{q_3} h'(q_1,q_2) \geq g_1'(q_1,q_2)\).
	A similar argument proves that, for all \(q_2,q_3\), \(\sum_{q_1} h'(q_1,q_2,q_3) \geq g_2(q_2,q_3)\).
	Finally, by hypothesis, we have either \(g_1(r_1,r_2) >0\) or \(g_2(r_2,r_3)>0\) so that \(\sum_{q_1,q_2} g_1'(q_1,q_2) + \sum_{q_2,q_3} g_2'(q_2,q_3) \leq \sum_{q_1,q_2} g_1(q_1,q_2) + \sum_{q_2,q_3} g_2(q_2,q_3)-1\).
	Therefore, \(\sum_{q_1,q_2,q_3} h'(q_1,q_2,q_3) = \sum_{q_1,q_2,q_3} h(q_1,q_2,q_3)-1 \geq \sum_{q_1,q_2} g_1(q_1,q_2) + \sum_{q_2,q_3} g_2(q_2,q_3)-1 \geq \sum_{q_1,q_2} g_1'(q_1,q_2) + \sum_{q_2,q_3} g_2'(q_2,q_3)\).
	We have proved that we may apply the induction hypothesis on \(h'\), \(g_1'\) and \(g_2'\).
	By doing so, we obtain \(\tilde{q_1}, \tilde{q_2}, \tilde{q_3}\) such that \(\sum_{q_3} h'(\tilde{q_1}, \tilde{q_2}, q_3) > g_1'(\tilde{q_1}, \tilde{q_2})\) and \(\sum_{q_1} h'(q_1, \tilde{q_2}, \tilde{q_3}) > g_2'(\tilde{q_2}, \tilde{q_3})\).
	We prove that the same holds for \(h, g_1\) and \(g_2\).
	If \((\tilde{q_1}, \tilde{q_2}) \ne (r_1,r_2)\) then \(\sum_{q_3} h(\tilde{q_1}, \tilde{q_2}, q_3) = \sum_{q_3} h'(\tilde{q_1}, \tilde{q_2}, q_3) > g_1'(\tilde{q_1}, \tilde{q_2}) = g_1(\tilde{q_1}, \tilde{q_2})\).
	Moreover, if \((\tilde{q_1}, \tilde{q_2})= (r_1,r_2)\), we have \(\sum_{q_3} h(\tilde{q_1}, \tilde{q_2}, q_3) = \sum_{q_3} h'(\tilde{q_1}, \tilde{q_2}, q_3)+1 > g_1'(\tilde{q_1}, \tilde{q_2}) +1 \geq g_1(\tilde{q_1}, \tilde{q_2})\).
	Overall, this proves that \(\sum_{q_3} h(\tilde{q_1}, \tilde{q_2}, q_3) > g_1(\tilde{q_1}, \tilde{q_2})\); a similar argument proves that \(\sum_{q_1} h(q_1, \tilde{q_2}, \tilde{q_3}) > g_2(\tilde{q_2}, \tilde{q_3})\).
	This concludes the induction.

	Applying the property to \(H\), \(f_1\) and \(f_2\) allow to obtain \(\tilde{q_1},\tilde{q_2},\tilde{q_3}\) such that \(\sum_{q_3} H(\tilde{q_1},\tilde{q_2}, q_3) > f_1(\tilde{q_1},\tilde{q_2})\) and \(\sum_{q_1} H(q_1,\tilde{q_2},\tilde{q_3}) > f_{2}(\tilde{q_2}, \tilde{q_3})\), so that we can set \(H'\) equal to \(H\) except that \(H'(\tilde{q_1},\tilde{q_2},\tilde{q_3}) = H(\tilde{q_1},\tilde{q_2},\tilde{q_3}) -1\).
	We then let \(f' : (q_1,q_2) \mapsto H'(q_1,q_2,q_3)\); \(H'\) is a product witness that \((f', \controlloc_1, \controlloc_3) \in \atf_1 \tftimes \atf_2\), but \((f', \controlloc_1, \controlloc_3) \letf \atf\), contradicting minimality of \(\atf\) in \(\atf_1 \tftimes \atf_2\).
\end{appendixproof}

\begin{example}
	Consider \cref{fig:example-composition}.
	Let \(\atf_1 = (f_1, \controlloc_1, \controlloc_2)\) and \(\atf_2 = (f_2, \controlloc_2, \controlloc_3)\), with \(f_1(q_1,q_2) = 2\), \(f_2(q_2,q_3) = 3\), \(f_1(q,q) = f_2(q,q) = 0\) for all \(q\), \(f_2(q_2,q_1) = 0\) and all other values equal to \(\dummysymb\).
	Let \(\atf = (f, \controlloc_1, \controlloc_3)\), with \(f(q_1,q_1) = 1\), \(f(q_1,q_3) = 1\), \(f(q_2,q_3) = 2\), \(f(q_2,q_2) = f(q_3,q_3) = f(q_1,q_2) = f(q_2,q_1) = 0\) and \(f(q,q') = \dummysymb\) for all other \((q,q')\).
	We have \(\atf \in \atf_1 \tftimes \atf_2\).
	Indeed, we have a product witness \(H\) defined by \(H(q_1, q_2,q_1) = 1\), \(H(q_1,q_2,q_3) = 1\), \(H(q_2,q_2,q_3) = 2\), \(H(q_1,q_2,q_2) = H(q_2,q_2,q_1)\) \(= H(q_2,q_2,q_2) = H(q_3,q_3,q_3) = H(q_1,q_1,q_1) = 0\) and all other values equal to \(\dummysymb\).
	In fact, \(\atf\) is minimal for \(\letf\) in \(\atf_1 \tftimes \atf_2\).
\end{example}

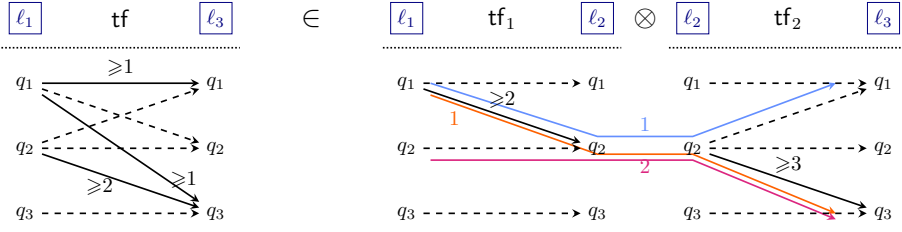
\begin{figure}
	\centering
	\resizebox{\linewidth}{!}{
		%
		\begin{tikzpicture}[xscale = 0.8]
			\node at (2,-0.5) {\large $\atf_1$};
			\node[controlpart] at (0,-0.5) (controlloc1) {$\controlloc_1$};
			\node[controlpart] at (4,-0.5) (controlloc2) {$\controlloc_2$};
			\foreach \i in {1,2,3}{
					\node (start\i) at (0, -0.5-1.1*\i) {$q_{\i}$};
					\node (end\i) at (4, -0.5-1.1*\i) {$q_{\i}$};
					\draw (start\i) edge[zerotransfer] (end\i);
				}
			\draw (start1) edge[transferedge] node[above] {${\geq} 2$} (end2);
			\draw[sepline] (-0.5, -1) -- (4.5,-1);
			\begin{scope}[xshift = 6cm]
				\node at (2,-0.5) {\large $\atf_2$};
				\node[controlpart] at (0,-0.5) (controlloc2prime) {$\controlloc_2$};
				\node[controlpart] at (4,-0.5) (controlloc3) {$\controlloc_3$};
				\foreach \i in {1,2,3}{
						\node (startbis\i) at (0, -0.5-1.1*\i) {$q_{\i}$};
						\node (endbis\i) at (4, -0.5-1.1*\i) {$q_{\i}$};
						\draw (startbis\i) edge[zerotransfer] (endbis\i);
					}
				\draw (startbis2) edge[zerotransfer]  (endbis1);
				\draw (startbis2) edge[transferedge] node[above] {${\geq} 3$} (endbis3);
				\draw[sepline] (-0.5, -1) -- (4.5,-1);
			\end{scope}
			\node at (5,-0.5) {\Large $\tftimes$}; %

			\begin{scope}[xshift = -8cm]
				\node at (2,-0.5) {\large $\atf$};
				\node[controlpart] at (0,-0.5) (controlloc1ter) {$\controlloc_1$};
				\node[controlpart] at (4,-0.5) (controlloc3ter) {$\controlloc_3$};
				\foreach \i in {2,3}{
						\node (startter\i) at (0, -0.5-1.1*\i) {$q_{\i}$};
						\node (endter\i) at (4, -0.5-1.1*\i) {$q_{\i}$};
						\draw (startter\i) edge[zerotransfer] (endter\i);
					}
				\node (startter1) at (0, -0.5-1.1) {$q_{1}$};
				\node (endter1) at (4, -0.5-1.1) {$q_{1}$};
				\draw (startter1) edge[transferedge] node[above] {${\geq} 1$} (endter1);
				\draw (startter1) edge[zerotransfer]  (endter2);
				\draw (startter2) edge[zerotransfer]  (endter1);
				\draw (startter2) edge[transferedge] node[below, xshift = -1em, yshift = 0.5em] {${\geq} 2$} (endter3);
				\draw (startter1) edge[transferedge] node[right, yshift = -1.5em, xshift = 2em] {${\geq} 1$} (endter3);
				\draw[sepline] (-0.5, -1) -- (4.5,-1);
			\end{scope}

			\node at (-2,-0.5) {\Large $\in$}; %

			\begin{scope}[on background layer]
				\draw[-stealth, color = color1, thick] (0.5,-1.6) -- (4,-2.5) -- (6,-2.5) -- (9,-1.6);
				\draw[-stealth, color = color2, thick] (0.5,-1.8) -- (4,-2.8) -- (6,-2.8) -- (9,-3.8);
				\draw[-stealth, color = color3, thick] (0.5,-2.9) -- (4,-2.9) -- (6,-2.9) -- (9,-3.9);
				\node[text = color1] at (5,-2.3) {$1$};
				\node[text = color2] at (1,-2.2) {$1$};
				\node[text = color3] at (5,-3.0) {$2$};
			\end{scope}

		\end{tikzpicture}
	}
	\caption{%
		Dashed arrows correspond to value \(0\), no arrow corresponds to \(\dummysymb\).
		The product witness \(H\) is represented with colored arrows.
		We do not depict \(H\) when its value is \(0\).
	}\label{fig:example-composition}
\end{figure}

Given a sequence \(\atrans_1 \dots \atrans_k\) of transitions, we let \(\transfersof{\atrans_1 \dots \atrans_k} \deff \transfersof{\atrans_1} \tftimes \transfersof{\atrans_2} \tftimes \dots \tftimes \transfersof{\atrans_k}\).
For the empty sequence \(\emptysequence\), we define \(\transfersof{\emptysequence}\) as the set of \((f,\bstate,\bstate')\) where \(\bstate = \bstate'\), \(f(q,q) \in \nats\) for all \(q\) and \(f(q,q') = \dummysymb\) for all \(q \ne q'\).
For all upward-closed sets \(F \subseteq \transferflows\), we have \(F \tftimes \transfersof{\emptysequence} = \transfersof{\emptysequence} \tftimes F = F\).
Observe that, for every \(\atrans_1 \dots \atrans_k\), all transfer flows \((f, \bstate, \bstate') \in \transfersof{\atrans_1 \dots \atrans_k}\) are such that \(f(q,q) \in \nats\) for all \(q\).

\begin{lemmarep}
	\label{sequence-transitions-product-tf}
	For all \(k \geq 0\), for all \(\atrans_1, \atrans_2 \dots, \atrans_k \in \transitions\) , and for all \(\prodconfig, \prodconfig' \in \prodconfigs\), \(\prodconfig \step{\atrans_1 \dots \atrans_k} \prodconfig' \) if and only if there exists \(\atf \in \transfersof{\atrans_1 \dots \atrans_k}\) such that \( \prodconfig \flowstep{\atf} \prodconfig'.
	\)
\end{lemmarep}
\begin{appendixproof}
	We first prove the following auxiliary lemma:
	\begin{lemma}
		\label{product-basic-property}
		Let \(\atf_1, \atf_2 \in \transferflows\) be two "transfer flows", \(\prodconfig_1, \prodconfig_3 \in \prodconfigs\).
		We have the following equivalence:
		\begin{equation*}
			( \exists \,  \atf \in \atf_1 \tftimes \atf_2, \prodconfig_1 \flowstep{\atf} \prodconfig_3 )
			\iff ( \exists \, \prodconfig_2 \in \prodconfigs, \prodconfig_1 \flowstep{\atf_1} \prodconfig_2 \flowstep{\atf_2} \prodconfig_3 ). \end{equation*}
	\end{lemma}
	\begin{proof}
		Let \(\atf_1 =: (f_1, \controlloc_1, \controlloc_2)\), \(\atf_2 =: (f_2, \controlloc_2', \controlloc_3)\).
		If we have \(\controlloc_2 \ne \controlloc_2'\) then \(\atf_1 \tftimes \atf_2 = \emptyset\), both assertions are false and the equivalence holds.
		Similarly, if the "control location" of \(\prodconfig_1\) is not equal to \(\controlloc_1\), then both assertions are false and the equivalence holds, and same for \(\prodconfig_3\) and \(\controlloc_3\).
		We now suppose that \(\prodconfig_1 =: \confpair{\config_1}{\controlloc_1}\) and \(\prodconfig =: \confpair{\config_3}{\controlloc_3}\).

		Assume first that there is \(\atf \in \atf_1 \tftimes \atf_2\) such that \(\prodconfig_1 \flowstep{\atf} \prodconfig_3\), let \(\atf =: (f, \controlloc_1, \controlloc_3)\).
		Let \(g \geq f\) witnessing that \(\prodconfig_1 \flowstep{\atf} \prodconfig_3\).
		By hypothesis, \(\atf \in \atf_1 \tftimes \atf_2\).
		By \cref{prod-upward-closed}, \(\atf_1 \tftimes \atf_2\) is upward-closed, therefore \(\atf' \deff (g,\controlloc_1,\controlloc_3) \in \atf_1 \tftimes \atf_2\).
		Let \(H : \states^3 \to \natsanddummy\) be a product witness of that.
		Let \(\config_2 : q_2 \in \states \mapsto \sum_{q_1,q_3} H(q_1,q_2,q_3)\), and let \(\prodconfig_2 \deff \confpair{\config_2}{\controlloc_2}\).
		Let \(h : (q_1,q_2) \in \states^2 \mapsto \sum_{q_3} H(q_1,q_2,q_3)\), we prove that \(h\) is a step witness that \(\prodconfig_1 \flowstep{\atf_1} \prodconfig_2\).
		By definition of \(H\), for all \(q_1,q_2\), \(\sum_{q_3} H(q_1,q_2,q_3) \geq f_1(q_1,q_2)\) hence \(h(q_1, q_2) \geq f_1(q_1,q_2)\), so that \(h \geq f_1\).
		By definition of \(H\), for all \(q_1\), \(\sum_{q_2} H(q_1,q_2,q_3) = g(q_1,q_3)\) and by definition of \(g\), \(\sum_{q_3} g(q_1,q_3) = \config_1(q_1)\).
		This gives, for all \(q_1\), \(\sum_{q_2} h(q_1,q_2) = \sum_{q_2,q_3} H(q_1,q_2,q_3) = \sum_{q_3} \sum_{q_2} H(q_1,q_2,q_3) = \sum_{q_3} g(q_1,g_3) = \config_1(q_1)\).
		Moreover, \(\sum_{q_1} h(q_1,q_2) = \sum_{q_1,q_3} H(q_1,q_2,q_3) = \config_2(q_2)\) by definition of \(\config_2\).
		This proves that \(\prodconfig_1 \flowstep{\atf_1} \prodconfig_2\); the proof that \(\prodconfig_2 \flowstep{\atf_2} \prodconfig_3\) is similar.

		Conversely, assume that there is \(\prodconfig_2\) such that \(\prodconfig_1 \flowstep{\atf_1} \prodconfig_2 \flowstep{\atf_2} \prodconfig_3\).
		Let \(g_1\geq f_1\) be a step witness that \(\prodconfig_1 \flowstep{\atf_1} \prodconfig_2\) and \(g_2 \geq f_2\) a step witness that \(\prodconfig_2 \flowstep{\atf_2} \prodconfig_3\).
		We build \(H : \states^3 \to \natsanddummy\) that satisfies the following conditions:
		\begin{enumerate}[label=(\roman*),ref=(\roman*),topsep=3pt]
			\item \label{comp-cond1} for all \(q_1,q_2\), \(\sum_{q_3}
			      H(q_1,q_2,q_3) = g_1(q_1,q_2)\), \item \label{comp-cond2} for all \(q_2,q_3\), \(\sum_{q_1} H(q_1,q_2,q_3) = g_2(q_2,q_3)\).
		\end{enumerate}
		Indeed, the existence of \(H\) would imply that, by letting \(h : (q_1,q_3) \mapsto \sum_{q_2} H(q_1,q_2,q_3)\) and \(\atf \deff (h,\controlloc_1,\controlloc_3)\), we have \(\atf \in \atf_1 \tftimes \atf_2\) (with \(H\) as product witness, because \(g_1 \geq f_1\) and \(g_2 \geq f_2\)) and \(\prodconfig_1 \flowstep{\atf} \prodconfig_3\) because \(\sum_{q_2} g_1(q_1,q_2) = \config_1(q_1)\) and \(\sum_{q_2} g_2(q_2,q_3) = \config_3(q_3)\).

		We now prove the following statement:
		\begin{quote}
			For every \(g_1: \states^2 \to \natsanddummy\) and \(g_2: \states^2 \to \natsanddummy\), if \(\sum_{q_1} g_1(q_1,q_2) = \sum_{q_3} g_2(q_2,q_3)\) for every \(q_2\), then there is \(H: \states^3 \to \natsanddummy\) such that \(\sum_{q_3} H(q_1,q_2,q_3) = g_1(q_1,q_2)\) and \(\sum_{q_1} H(q_1,q_2,q_3) = G(q_2,q_3)\).
		\end{quote}
		First, if \(F\) and \(G\) are constant equal to \(\dummysymb\) then we set \(H\) constant equal to \(\dummysymb\).
		Suppose that it is not the case; let \(n \deff \sum_{q_1,q_2} g_1(q_1,q_2) = \sum_{q_2,q_3} g_2(q_2,q_3) \in \nats\).
		We proceed by induction on \(n\).

		If \(n=0\) then all values in \(g_1\) and \(g_2\) are in \(\set{0,\dummysymb}\).
		For each \(q_1,q_2,q_3\), we let \(H(q_1,q_2,q_3) \deff 0\) whenever both \(g_1(q_1,q_2) =0\) and \(g_2(q_2,q_3) =0\), and \(H(q_1,q_2,q_3) \deff \dummysymb\) otherwise.
		We first prove that, for all \(q_1,q_2\), \(\sum_{q_3} H(q_1,q_2,q_3) = g_1(q_1,q_2)\).
		Let \(q_1,q_2 \in \states\); if \(g_1(q_1,q_2,q_3) = \dummysymb\) then \(H(q_1,q_2,q_3) = \dummysymb\) for all \(q_3\) hence \(\sum_{q_3} H(q_1,q_2,q_3) = \dummysymb\).
		Suppose now that \(g_1(q_1,q_2) = 0\).
		This implies that \(\sum_{q_3} g_2(q_2,q_3) = 0\) therefore there is \(\tilde{q_3}\) such that \(g_2(q_2,\tilde{q_3}) =0\), so that \(H(q_1,q_2,\tilde{q_3})=0\) and \(\sum_{q_3} H(q_1,q_2,q_3) = 0\).
		Similarly, for every \(q_2,q_3\), if \(g_2(q_2,q_3) = \dummysymb\) then \(\sum_{q_1} H(q_1,q_2,q_3) = \dummysymb\) and if \(g_2(q_2,q_3) = 0\) then there is \(\tilde{q_1}\) such that \(g_1(\tilde{q_1},q_2) =0\) hence \(H(\tilde{q_1},q_2,q_3) =0\) and \(\sum_{q_1} H(q_1,q_2,q_3) =0\).

		Suppose now that \(n>0\).
		There exists \(\tilde{q_2}\) such that \(\sum_{q_1} g_1(q_1, \tilde{q_2}) = \sum_{q_3} g_2(\tilde{q_2},q_3)>0\).
		Let \(\tilde{q_1}\) such that \(g_1(\tilde{q_1}, \tilde{q_2})>0\) and \(\tilde{q_3}\) such that \(g_2(\tilde{q_2},\tilde{q_3})>0\).
		Let \(g_1'\) equal to \(g_1\) except that \(g_1'(\tilde{q_1}, \tilde{q_2}) \deff g_1(\tilde{q_1}, \tilde{q_2})-1\) and let \(g_2'\) equal to \(g_2\) except that \(g_2'(\tilde{q_2},\tilde{q_3}) \deff g_2(\tilde{q_2},\tilde{q_3})-1\).
		We have \(\sum_{q_1} g_1'(q_1,q_2) = \sum_{q_3} g_2'(q_2,q_3)\) for all \(q_2\), and \(\sum_{q_1,q_2} g_1'(q_1,q_2) = \sum_{q_2,q_3} g_2(q_2,q_3)-1 = n-1\).
		We apply the induction hypothesis on \(g_1'\) and \(g_2'\) to obtain \(H'\) such that \(\sum_{q_3} H'(q_1,q_2,q_3) = g_1'(q_1,q_2)\) for all \(q_1,q_2\) and \(\sum_{q_1} H'(q_1,q_2,q_3) = g_2'(q_2,q_3)\) for all \(q_2,q_3\).
		It suffices to let \(H\) equal to \(H'\) except that \(H(\tilde{q_1}, \tilde{q_2}, \tilde{q_3}) \deff H'(\tilde{q_1}, \tilde{q_2}, \tilde{q_3})+1\).
		Note that it could be that \(H'(\tilde{q_1}, \tilde{q_2}, \tilde{q_3}) = \dummysymb\), in which case \(H(\tilde{q_1}, \tilde{q_2}, \tilde{q_3}) = 1\).
		We know that \(g_1'(\tilde{q_1}, \tilde{q_2}) \ne \dummysymb\) therefore \(\sum_{q_3} H'(\tilde{q_1}, \tilde{q_2}, q_3) \ne \dummysymb\) so that we indeed have \(\sum_{q_3} H(\tilde{q_1}, \tilde{q_2}, q_3) = g_1'(\tilde{q_1}, \tilde{q_2})+1 = g_1(\tilde{q_1}, \tilde{q_2})\).
		With the same argument, \(\sum_{q_1} H(q_1, \tilde{q_2},\tilde{q_3}) = g_2(\tilde{q_2}, \tilde{q_3})\).
		This concludes the induction.

		By letting \(h : (q_1,q_3) \mapsto \sum_{q_2} H(q_1,q_2,q_3)\) and \(\atf = (h, \controlloc_1, \controlloc_3)\), we have \(\atf \in \atf_1 \tftimes \atf_2\) and \(\prodconfig_1 \flowstep{\atf} \prodconfig_3\), concluding the proof.
	\end{proof}

	We now prove \cref{sequence-transitions-product-tf}.
	We proceed by induction on \(k\).
	The case \(k=0\) corresponds to the fact that \(\prodconfig = \prodconfig'\) if and only if there is \(\atf \in \transfersof{\emptysequence}\) such that \(\prodconfig \flowstep{\atf} \prodconfig'\).
	We assume that the property is true for sequences of length up to \(k\), and we prove it for sequence of length \(k+1\).
	Let \(\atrans_1, \dots, \atrans_{k+1} \in \transitions\).
	First, assume that \(\prodconfig \accstep{\atrans_1 \dots \atrans_{k+1}} \prodconfig'\); split this execution into \(\prodconfig = \prodconfig_0 \accstep{\atrans_1} \prodconfig_1 \accstep{\atrans_2} \dots \accstep{\atrans_{k+1}} \prodconfig_{k+1} = \prodconfig'\).
	By induction hypothesis, there is \(\atf \in \transfersof{\atrans_1 \dots \atrans_k}\) such that \(\prodconfig_0 \flowstep{\atf} \prodconfig_k\).
	By \cref{flow-one-trans}, there is \(\atf_{k+1} \in \transfersof{\atrans_{k+1}}\) such that \(\prodconfig_k \flowstep{\atf_{k+1}} \prodconfig_{k+1}\).
	We apply \cref{product-basic-property} to obtain the existence of \(\atf' \in \transfersof{\atrans_1 \dots \atrans_k} \tftimes \transfersof{\atrans_{k+1}} = \transfersof{\atrans_1 \dots \atrans_{k+1}}\) such that \(\prodconfig_0 \flowstep{\atf'} \prodconfig_{k+1}\).
	Conversely, assume that there is \(\atf' \in \transfersof{\atrans_1 \dots \atrans_{k+1}}\) such that \(\prodconfig \flowstep{\atf'} \prodconfig'\).
	We have \(\atf' \in \transfersof{\atrans_1 \dots \atrans_k} \tftimes \transfersof{\atrans_{k+1}}\), hence by \cref{product-basic-property} there is \(\prodconfig_k\), \(\atf \in \transfersof{\atrans_1 \dots \atrans_k}\) and \(\atf_{k+1} \in \transfersof{\atrans_{k+1}}\) such that \(\prodconfig \flowstep{\atf} \prodconfig_k \flowstep{\atf_{k+1}} \prodconfig'\).
	We conclude by applying the induction hypothesis to \(\prodconfig \flowstep{\atf} \prodconfig_k\) and by applying \cref{flow-one-trans} to \(\prodconfig_k \flowstep{\atf_{k+1}} \prodconfig'\).
\end{appendixproof}

Given \(T \subseteq \transitions^*\), we let \(\transfersof{T} \deff \bigcup_{w \in T} \transfersof{w}\).
For all \(k \geq 0\), we denote by \(\transitions^{\leq k} \subseteq \transitions^*\) the set of sequences of length at most \(k\).
Let \(m = \size{\states}\) and \(M = \size{\bstates}\).
We prove \cref{thm:k-blind-post} using the following theorem, which proved in \cref{subsec:stuctural-bound}.

\begin{theorem}[Structural theorem]
	\label{structural-theorem-flows}
	Let \(\structuralbound \deff (M+1)^{3^{m^2+2} \cdot 2(\log(m^2+2) +1) m^2}\).
	We have \(\transfersof{\transitions^{\leq B}} = \transfersof{\transitions^*}\) and elements of \(\minwqo{\transfersof{\transitions^*}}\) have norm at most \(2B\).
\end{theorem}
%

%
\subsection{Proof of Theorem~\ref{thm:k-blind-post}}

Again, we write \(m = \size{\states}\) and \(M = \size{\bstates}\).
Let \(K' \geq 0\), \(K \deff m^2 \max(K', 2 B)\) and \(\cSet\) a \(K'\)-blind set.
We prove that \(\poststar{\cSet}\) is \(K\)-blind; the proof for \(\prestar{\cSet}\) is similar.
We start with the following observation.

\begin{lemma}
	\label{poststar-chacterization}
	A configuration \(\prodconfig\) is in \(\poststar{\cSet}\) if and only if there are \(\prodconfig_\cSet \in \cSet\) and \(\atf \in \transfersof{\transitions^*}\) such that \(\prodconfig_\cSet \flowstep{\atf} \prodconfig\) and \(\tfweight{\atf} \leq 2B\).
\end{lemma}
\begin{proof}
	By \cref{sequence-transitions-product-tf}, if we have such \(\prodconfig_\cSet\) and \(\atf\) then \(\prodconfig \in \poststar{\cSet}\).
	Conversely, if \(\prodconfig = (\config,\bstate) \in \poststar{\cSet}\), there are \(\prodconfig_\cSet = (\config_\cSet, \bstate_\cSet) \in \cSet\) and \(w \in \transitions^*\) such that \(\config_\cSet \step{w} \config\).
	By \cref{sequence-transitions-product-tf}, there is \(\atf = (f, \bstate_\cSet, \bstate) \in \transfersof{w} \subseteq \transfersof{\transitions^*}\) such that \(\prodconfig_\cSet \flowstep{\atf} \prodconfig\); by \cref{structural-theorem-flows}, one may assume that \(\tfweight{\atf} \leq 2B\).
\end{proof}

Let \(\prodconfig = (\config, \bstate) \in \prodconfigs\) and \(q \in \states\) such that \(\config(q) \geq K\); we show that \((\config, \bstate) \in \poststar{\cSet}\) if and only if \((\config + \vec{q}, \bstate) \in \poststar{\cSet}\).
First, suppose that \(\prodconfig = (\config,\bstate) \in \poststar{\cSet}\).
Let \(\atf, \prodconfig_\cSet = (\config_\cSet, \bstate_\cSet)\) obtained thanks to \cref{poststar-chacterization}.
Let \(g : \states^2 \to \natsanddummy\) be a step witness that \(\prodconfig_\cSet \flowstep{\atf} \prodconfig\).
We have \(\sum_{r \in \states} g(r,q) = \config(q) \geq K\).
By the pigeonhole principle, there is \(r\) such that \(g(r,q) \geq \frac{K}{m^2} \geq K'\) therefore \(\config_\cSet(r) \geq K'\).
Let \(g'\) such that \(g'(q_1,q_2) = g(q_1,q_2)\) for all \((q_1,q_2) \ne (r,q)\) and \(g'(r,q) = g(r,q)+1\); \(g'\) is a witness that \((\config_\cSet+ \vec{r},\bstate_\cSet) \flowstep{\atf} (\config +\vec{q}, \bstate)\).
Thanks to \cref{sequence-transitions-product-tf}, this proves that \((\config_\cSet+ \vec{r},\bstate_\cSet) \step{*} (\config +\vec{q}, \bstate)\).
Because \(\cSet\) is \(K'\)-blind, we conclude that \((\config +\vec{q}, \bstate) \in \poststar{\cSet}\).
Conversely, suppose that \((\prodconfig + \vec{q}, \bstate) \in \poststar{\cSet}\).
With the same reasoning as above, we obtain \(\prodconfig_\cSet = (\config_\cSet, \bstate_\cSet)\in \cSet, \atf = (f, \bstate_\cSet, \bstate), g, r\) such that \(g\) is a witness that \(\prodconfig_\cSet \flowstep{\atf} \prodconfig\).
By the pigeonhole principle, there is \(r\) such that \(g(r,q) \geq K'+1\) and \(g(r,q) \geq 2B+1 > f(r,q)\).
Because \(\cSet\) is \(K'\)-blind and \(\config_\cSet(r) \geq g(r,q) \geq K' +1\), we have \((\config_\cSet- \vec{r},\bstate_\cSet) \in \cSet\).
Let \(g'(q_1,q_2) = g(q_1,q_2)\) for all \((q_1,q_2) \ne (r,q)\) and \(g'(r,q) = g(r,q)-1\).
Because \(g' \geq f\), \(g'\) is a step witness that \((\config_\cSet-\vec{r},\bstate_\cSet) \flowstep{\atf} (\config, \bstate)\).
Thanks to \cref{sequence-transitions-product-tf}, this proves that \((\config, \bstate) \in \poststar{\cSet}\).
\subsection{Proving the Structural Theorem with Descending Chains}
\label{subsec:stuctural-bound}

To prove \cref{structural-theorem-flows}, we use a result bounding the length of descending chains in \(\nats^d\) from \cite{LS21,SS23}.
We recall the result and some definitions.
Let \(d \geq 1\).
Given \(\vec{v}\) of \(\nats^d\) and \(i \in \nset{1}{d}\), we denote by \(\vec{v}(i)\) its \(i\)-th component.
Let \(\leprod\) be the order over \(\nats^d\) such that \(\vec{u} \leprod \vec{v}\) if and only if, for all \(i \in \nset{1}{d}\), \(\vec{u}(i) \leq \vec{v}(i)\).
The obtained \((\nats^d, \leprod)\) is a well-quasi-order (Dickson's lemma \cite{Dickson}).
A ""descending chain"" is a sequence \(D_0 \supsetneq D_1 \supsetneq D_2 \dots\) of sets \(D_k \subseteq \nats^d\) that are downward-closed for \(\leprod\).
Because \((\nats^d,\leprod)\) is a "well-quasi-order", all "descending chains" have finite length, \emph{i.e.}, are of the form \(D_0, \dots, D_\ell\) with \(\ell \in \nats\).
To bound the length of "descending chains" \cite{LS21,SS23} we need the sequence to be \emph{controlled} and \emph{\(\omega\)-monotone}.

We extend \(\nats\) to \(\omeganats \deff \nats \cup \set{\omega}\) with \(n < \omega\) for all \(n \in \nats\).
Given \(\vec{v} \in \omeganats^d\), its ""norm@@vector"" \(\norm{\vec{v}}\) is the largest \(\vec{v}(i)\) that is not \(\omega\).
An ""ideal"" \(I\) is the downward-closure in \(\nats^d\) of a vector \(\vec{v} \in \omeganats^d\), \emph{i.e.}, \(I = \downwardclosure{\{\vec{v}\}}\cap\nats^d\); its ""norm@@ideal"" \(\idealnorm{I}\) is \(\norm{\vec{v}}\).
A downward-closed set \(D \subseteq \nats^d\) is canonically represented as a finite union of "ideals"; its ""norm@@dcset"" \(\norm{D}\) is the maximum of the norms of its ideals.
Given \(N > 0\) and a "descending chain" \((D_k)\), we call \((D_k)\) ""\(N\)-controlled"" when, for all \(k\), \(\norm{D_k} \leq (k+1) N\).
In a "descending chain" \((D_k)\), an ideal \(I\) is proper at step \(k\) if \(I\) is in the canonical representation of \(D_k\) but \(I \nsubseteq D_{k+1}\).
The sequence \((D_k)\) is ""\(\omega\)-monotone"" if, when an ideal \(I_{k+1}\) represented by some vector \(\vec{v}_{k+1}\) is "proper" at step \(k+1\), there is \(I_k\) that is "proper" at step \(k\) and that is represented by \(\vec{v}_k\) such that, for all \(i \in \nset{1}{d}\), if \(\vec{v}_{k+1}(i) = \omega\) then \(\vec{v}_k(i) = \omega\).

\begin{theorem}[{\cite{SS23}}]
	\label{bound-descending-chains}
	Let \(d, n > 0\).
	Every "descending chain" \((D_k)\) of \(\nats^d\) that is "\(n\)-controlled" and "\(\omega\)-monotone" has length at most \(n^{3^{d} (\log(d)+1)}\).
\end{theorem}

We now use this bound to prove \cref{structural-theorem-flows}.
Recall that we write \(m = \size{\states}\) and \(M = \size{\bstates}\).
Let \(d \deff m^2 + 2\) and \(N \deff M^2 \cdot 2^{m^2} = \size{\bstates^2 \times 2^{\states^2}}\).
We fix two arbitrary bijective mappings \(\finitemapping: \bstates^2 \times 2^{\states^2} \to \nset{1}{N}\) and \(\indexof:\states^2 \to \nset{1}{m^2}\).
We map "transfer flows" to sets of elements of \(\nats^d\) with \(\mapencoding: \transferflows \to 2^{\nats^d}\).
Let \(\atf = (f,\bstate, \bstate') \in \transferflows\) and \(S \deff \set{(q,q') \mid f(q,q') = \dummysymb}\).
A vector \(\vec{v} \in \nats^{d}\) is in \(\mapencoding(\atf)\) when:
\begin{itemize}[noitemsep,topsep=0pt,parsep=0pt,partopsep=0pt]
	\item for all \((q,q')\) such that \(f(q,q') \ne \dummysymb\), \(\vec{v}(\indexof(q,q')) = f(q,q')\);
	\item \(\vec{v}(m^2+1) = \finitemapping(\bstate, \bstate', S)\);
	\item \(\vec{v}(m^2+2) = N+1 -  \finitemapping(\bstate, \bstate', S)\).
\end{itemize}
Note that there is no restriction to \(\vec{v}(i)\) when the corresponding pair \((q,q') = \indexof^{-1}(i)\) is such that \(f(q,q') = \dummysymb\).
Also, if \(\vec{v} \in \mapencoding(\atf)\) and \(\vec{u} \in \mapencoding(\atf')\) are such that \(\vec{v} \leprod \vec{u}\), then \(\vec{u}(m^2+1) = \vec{v}(m^2+1)\) and \(\vec{u}(m^2+2) = \vec{v}(m^2+2)\), so that \(\atf\) and \(\atf'\) have the same states of \(\bstates\) and the same \(\dummysymb\) components.
For \(\atf \ne \atf'\), we have \(\mapencoding(\atf), \mapencoding(\atf') \ne \emptyset\) but \(\mapencoding(\atf) \cap \mapencoding(\atf') = \emptyset\), a property that we call ""strong injectivity"" of \(\mapencoding\).
The vectors of \(\nats^d \cap \mapencoding(\transferflows)\) are exactly those whose last two components are strictly positive and sum to \(N+1\).
We build a "decreasing chain" \((D_k)\) such that \(D_k \cap \mapencoding(\transferflows) = \mapencoding(\transferflows \setminus \transfersof{\transitions^{\leq k}})\).

Let \(\zerovectors\) denote the set of vectors \(\vec{v}\) such that either \((\vec{v}(m^2+1),\vec{v}(m^2+2))= (N+1,0)\) or \((\vec{v}(m^2+1),\vec{v}(m^2+2))= (0,N+1)\).
For technical reasons (related to \(\omega\)-monotonicity), we will enforce that \(D_k \cap \zerovectors = \emptyset\) for every \(k\).
Note that \(\zerovectors \cap \relevantvectors = \emptyset\): vectors in \(\zerovectors\) have no relevance in terms of "transfer flows".
For all \(k \geq 0\), let \(U_k\deff \upwardclosure{\mapencoding(\transfersof{\transitions^{\leq k}}) \cup \zerovectors}\), and let \(D_k = \nats^d \setminus U_k\); \((D_k)\) is a "decreasing chain" because all \(D_k\) are downward-closed and \(\transfersof{\transitions^{\leq k}} \subseteq \transfersof{\transitions^{\leq k+1}}\) for all \(k\).

\begin{lemmarep}
	\label{upward-closure-relevant-vectors}
	For all \(k\), \(U_k \cap \relevantvectors = \mapencoding(\transfersof{\transitions^{\leq k}})\) and \(D_k \cap \relevantvectors = \mapencoding(\transferflows \setminus \transfersof{\transitions^{\leq k}})\).
\end{lemmarep}
\begin{appendixproof}
	It suffices to prove the claim for \(U_k\).
	Trivially, \(\mapencoding(\transfersof{\transitions^{\leq k}}) \subseteq U_k \cap \relevantvectors\).
	Conversely, let \(\vec{v} \in U_k \cap \relevantvectors\).
	There exists \(\vec{u} \in \mapencoding(\transfersof{\transitions^{\leq k}}) \cup \zerovectors\) such that \(\vec{u} \leprod \vec{v}\).
	Since \(\vec{v} \in \relevantvectors\), the last two components of \(\vec{v}\) sum to \(N\) and same for \(\vec{u}\), so that \(\vec{u}(m^2+1) = \vec{v}(m^2+1)\) and \(\vec{u}(m^2+2) = \vec{v}(m^2+2)\).
	This proves that \(\vec{u} \notin \zerovectors\) because \(\vec{v} \in \relevantvectors\), therefore \(\vec{u} \in \mapencoding(\transfersof{\transitions^{\leq k}})\).
	Let \(\atf_u = (f_u,\bstate_u, \bstate_u') \in \transfersof{\transitions^{\leq k}}\) such that \(\vec{u} \in \mapencoding(\atf_u)\); let \(\atf_v = (f_v, \bstate_v, \bstate_v') \in \transferflows\) such that \( \vec{v} \in \mapencoding(\atf_v)\).
	Because \(\vec{u}\) and \(\vec{v}\) coincide on the last two component, we have \(\bstate_u = \bstate_v\), \(\bstate_u' = \bstate_v'\) and \(f_u(q,q') = \dummysymb\) whenever \(f_v(q,q') = \dummysymb\).
	When \(f_u(q,q'), f_v(q,q') \ne \dummysymb\), we have \(f_v(q,q') \leq f_u(q,q')\), hence \(\atf_v \in \transfersof{\transitions^{\leq k}}\) because \(\transfersof{\transitions^{\leq k}}\) is upward-closed for \(\letf\).
\end{appendixproof}

Note that if \(D_{k+1} = D_k\) then, by \cref{upward-closure-relevant-vectors}, \(\mapencoding(\transfersof{\transitions^{\leq k+1}}) = \mapencoding(\transfersof{\transitions^{\leq k}})\) and, by injectivity of \(\mapencoding\), \(\transfersof{\transitions^{\leq k+1}} = \transfersof{\transitions^{\leq k}}\).
This means that if \(D_{k+1} = D_k\) then \(\transfersof{\transitions^{\leq k}}\) is stable under product by \(\transfersof{\atrans}\) for all \(\atrans\), hence that \(\transfersof{\transitions^*} = \transfersof{\transitions^{\leq k}}\).
Let \(L\) be the smallest \(k \in \nats\) such that \(D_{k} \ne D_{k-1}\); it exists because \((\nats^d, \leprod)\) is a well-quasi-order.
To prove \cref{structural-theorem-flows}, we want \(L \leq N^{3^{d} (\log(d)+1)}\).
To apply \cref{bound-descending-chains}, we need to prove that \((D_k)\) is "\((N+1)\)-controlled" and "\(\omega\)-monotone".

Transfer flows in \(\minwqo{\T}\) have weight bounded by \(2\).
Let \(\atf \in \minwqo{\Tleq{k}}\), there are \(\ell \leq k\) and \(\tftrans_1,\dots,\tftrans_\ell \in \minwqo{\transfersof{\transitions}}\) such that \(\atf \in \tftrans_1 \tftimes \ldots \tftimes \tftrans_\ell\).
A straightforward induction using \ref{basic2} proves that \(\tfweight{\atf} \leq 2 \ell \leq 2 k\).
This proves that minimal elements of \(\Tleq{k}\) have weight bounded by \(2k\).
In turn, this bounds the norm of minimal elements of \(U_k\) by \(\max(N+1,2k)\).
Because \(D_k = \nats^d \setminus U_k\), this last bound applies to the norm of \(D_k\).

\begin{toappendix}
	\begin{lemma}
		\label{min-Uk-min-Tleqk}
		Minimal vectors of \(\mapencoding(\Tleq{k})\) are in \(\mapencoding(\minwqo{\Tleq{k}})\).
	\end{lemma}
	\begin{proof}
		Let \(\vec{v}\) minimal in \(\mapencoding(\Tleq{k})\).
		In particular, \(\vec{v} \in \mapencoding(\Tleq{k})\); let \(\atf = (f, \controlloc_1, \controlloc_2) \in \Tleq{k}\) such that \(\vec{v} \in \mapencoding(\atf)\).
		Our aim is to prove that \(\atf \in \minwqo{\Tleq{k}}\).
		Let \(S \deff \set{(q,q') \mid f(q,q') = \dummysymb}\).
		Let \(\atf' = (f', \controlloc_1', \controlloc_2) \letf \atf\); we prove that \(\atf' = \atf\).
		Because \(\atf' \letf \atf\), by letting \(S' \deff \set{(q,q') \mid f(q,q') = \dummysymb}\), we have \(S' = S\).
		Therefore, there exists \(\vec{u} \in \mapencoding(\atf')\) such that \(\vec{u}(i) = 0\) for all \(i \in \indexof^{-1}(S)\).
		We claim that \(\vec{u} \leprod \vec{v}\).
		We have \(\vec{u}(m^2+1) = \vec{v}(m^2+1)\) and \(\vec{u}(m^2+2) = \vec{v}(m^2+2)\); for all \(i \in \indexof^{-1}(S)\), \(\vec{u}(i) = 0 \leq \vec{v}(i)\); for all \(i \notin \indexof^{-1}(S)\), by letting \((q,q') \deff \indexof^{-1}(i)\), we have \(\vec{u}(i) = f'(q,q') \leq f(q,q') = \vec{v}(i)\).
		We have therefore \(\vec{u} \in \mapencoding(\Tleq{k})\) and \(\vec{u} \leprod \vec{v}\), but \(\vec{v}\) is minimal in \(\minwqo{\mapencoding(\Tleq{k})}\) therefore \(\vec{u} = \vec{v}\).
		This implies that \(f = f'\) hence that \(\atf = \atf'\).
	\end{proof}

	\begin{lemma}
		\label{bound-minwqo-Tleqk}
		For all \(k \geq 0\), for all \(\atf \in \minwqo{\Tleq{k}}\), \(\tfweight{\atf} \leq 2 k\).
	\end{lemma}
	\begin{proof}
		The proof is by induction on \(k\) and relies on the bound from \ref{basic4}.
		For \(k=0\), \(\Tleq{0} = \neutraltransferflows\) and "transfer flows" in \(\minwqo{\neutraltransferflows}\) only have values \(0\) and \(\dummysymb\), so that they have "weight@@tf" \(0\).
		Suppose that the statement is true for \(k\), and prove it for \(k+1\).
		Let \(\atf \in \minwqo{\Tleq{k+1}}\).
		If \(\atf \in \Tleq{k}\) then, because \(\Tleq{k} \subseteq \Tleq{k+1}\), \(\atf \in \minwqo{\Tleq{k}}\) and it suffices to apply the induction hypothesis on \(\atf\).
		Otherwise, there is \(\atf_k \in \Tleq{k}\), \(\tftrans \in \Tmin\) such that \(\atf \in \atf_k \tftimes \tftrans\).
		By \ref{basic2}, we may assume that \(\atf_k \in \minwqo{\Tleq{k}}\).
		By applying the induction hypothesis, we have \(\tfweight{\atf_k} \leq \maxweight k\); by construction of \(\T\), we have \(\tfweight{\tftrans} \leq \maxweight\).
		By \ref{basic4}, we obtain that \(\tfweight{\atf} \leq \maxweight k + \maxweight = \maxweight(k+1)\), concluding the induction.
	\end{proof}
\end{toappendix}

\begin{lemmarep}
	\label{Dk-controlled}
	\((D_k)\) is \((N+1)\)-controlled and \(\omega\)-monotone.
\end{lemmarep}
\begin{appendixproof}
	Towards proving \cref{Dk-controlled}, we start by bounding the norm of minimal elements of \(U_k\) which is the result of \cref{min-Uk-min-Tleqk}.
	For all \(k\), \(U_k\) is "upward-closed" for \(\leprod\) hence it has a finite "basis" \(\minwqo{U_k}\).

	Note that, because \(\mapencoding(\Tleq{k})\) is not "upward-closed", we cannot write that minimal vectors of \(\mapencoding(\Tleq{k})\) are in the "basis" of the set.
	The remaining task is to bound the values of "transfer flows" in \(\minwqo{\Tleq{k}}\), which is the result of \cref{bound-minwqo-Tleqk}.

	Because \(U_k\) is the upward-closure of \(\mapencoding(\Tleq{k}) \cup \zerovectors\), we have \(\minwqo{U_k} \subseteq {\mapencoding(\Tleq{k})} \cup {\zerovectors}\).
	A vector \(\vec{v} \in \minwqo{\zerovectors}\) is such that \(\vec{v}(i) = 0\) for all \(i \in \nset{1}{m^2}\), and \(\max(\vec{v}(m^2+1), \vec{v}(m^2+2)) =N+1\), so that \(\norm{\vec{v}} = N+1\).
	We now consider vectors in \(\minwqo{U_k} \cap {\mapencoding(\Tleq{k})}\); such vectors must be minimal in \({\mapencoding(\Tleq{k})}\).
	We now conclude the proof of \cref{Dk-controlled}.
	Let \(k \in \nats\), and let \(\vec{v} \in \omeganats^d\) be the representing vector of some ideal of \(D_k\).
	First, we argue that \(\vec{v}(i) \leq N+1\) for \(i \in \set{m^2+1, m^2+2}\).
	Indeed, we would otherwise have a vector \(\vec{u} \leprod \vec{v}\) such that \(\vec{u} \in \zerovectors\), which contradicts the fact that \(\zerovectors \subseteq U_k\).
	Let \(i \in \nset{1}{m^2}\) such that \(\vec{v}(i) \ne \omega\).
	Let \(\vec{u}\) denote the vector equal to \(\vec{v}\) except that \(\vec{u}(i) \deff \vec{v}(i) + 1\).
	Because \(\vec{v}\) is maximal in \(D_k\), \(\downwardclosure{\set{\vec{u}}} \nsubseteq D_k\).
	Therefore, there is a vector \(\vec{u}_m \in \minwqo{U_k}\) such that \(\vec{u}_m\leprod \vec{u}\).
	We must have \(\vec{u}_m(i) = \vec{v}(i)+1\) because we would otherwise have \(\vec{u}_m \leprod \vec{v}\), which would imply that \(\vec{u}_m \in D_k\) and would contradict \(\vec{u}_m \in U_k\).
	By definition of \(U_k\), we have \(\vec{u}_m \in \zerovectors \cup \mapencoding(\Tleq{k})\); but \(\vec{u}_m(i)>0\), hence \(\vec{u}_m \notin \zerovectors\) therefore \(\vec{u}_m \in \mapencoding(\Tleq{k})\).
	Moreover, because \(\vec{u}_m \in \minwqo{U_k}\), by \cref{min-Uk-min-Tleqk}, there is \(\atf_m = (f_m, \controlloc, \controlloc') \in \minwqo{\Tleq{k}}\) such that \(\vec{u}_m \in \mapencoding(\atf_m)\).
	By \cref{bound-minwqo-Tleqk}, we have \(\tfweight{\atf_m} \leq \maxweight k\) so that \(f_m(q,q') \in \nset{0}{\maxweight k} \cup \set{\dummysymb}\) for all \(q,q'\).
	This proves in particular that \(\vec{v}(i) \leq \vec{u}_m(i) \leq \maxweight k\).
	Overall, we have proved that, for all \(i \in \nset{1}{m^2}\) such that \(\vec{v}(i) \ne \omega\), \(\vec{v}(i) \leq \maxweight k\), and that \(\vec{v}(m^2+1) \leq N+1\) and \(\vec{v}(m^2+2) \leq N+1\), so that \(\norm{\vec{v}} \leq \max(\maxweight k, N+1)\).
	Because \(N+1 \geq 2\), the norm of \(D_k\) is bounded by \((N+1) (k+1)\), concluding the part of the proof that the sequence is \((N+1)\)-controlled.

	Let us now turn to the \(\omega\)-monotonicity which is very technical.

	Let \(k \geq 0\), let \(I_{k+1} \subseteq D_{k+1}\) be a "proper ideal at step \(k+1\)".
	Let \(\vec{v}_{k+1} \in \omeganats^d\) be the vector representing \(I_{k+1}\).
	Because \(I_{k+1}\) is proper, \(I_{k+1} \nsubseteq D_{k+2}\).

	\begin{lemma}
		\label{ikplisone-intersects}
		\(I_{k+1} \cap (\mapencoding(\Tpow{k+2}) \setminus \mapencoding(\Tleq{k+1})) \ne \emptyset\).
	\end{lemma}
	\begin{proof}
		We know that \(I_{k+1} \cap (U_{k+2} \setminus U_{k+1}) \ne \emptyset\) because \(I_{k+1}\) is a "proper ideal at step \(k+1\)".
		Let \(\vec{v} \in I_{k+1} \cap (U_{k+2} \setminus U_{k+1})\).
		Because \(\vec{v} \in U_{k+2}\), there is \(\vec{u} \in \mapencoding(\Tleq{k+2}) \cup \zerovectors\) such that \(\vec{u} \leprod \vec{v}\).
		We have \(\vec{u} \notin \zerovectors\) as it would otherwise imply that \(\vec{v} \in U_{k+1}\), therefore \(\vec{u} \in \mapencoding(\Tleq{k+2})\).
		Also, \(\vec{u} \notin \mapencoding(\Tleq{k+1})\) as it would otherwise imply that \(\vec{v} \in U_{k+1}\).
		Because \(I_{k+1}\) is downward-closed, \(\vec{u} \in I_{k+1}\), so that \(\vec{u} \in I_{k+1} \cap (\mapencoding(\Tpow{k+2}) \setminus \mapencoding(\Tleq{k+1}))\).
	\end{proof}

	Given a set \(I \subseteq \nats^d\) of vectors and \(\tftrans \in \Tmin\), let \(\preflows{\tftrans}{I}\) be the set of vectors \(\vec{v}\) such that there are "transfer flows" \(\atf_{I} \in \transferflows\), \(\atf_{\vec{v}} \in \atf_I \tftimes \tftrans \) with \(\vec{v} \in \mapencoding(\atf_{\vec{v}})\) and \(\mapencoding(\atf_I) \subseteq I\).

	\begin{lemma}
		For all \(\tftrans \in \Tmin\), \(\preflows{\tftrans}{I_{k+1}} \subseteq D_k\).
	\end{lemma}
	\begin{proof}
		Suppose by contradiction that we have \(\vec{v} \in \preflows{\tftrans}{I_{k+1}} \cap U_k\).
		There are \(\atf_{k}, \atf_{k+1}\) such that \(\vec{v} \in \mapencoding(\atf_k)\) and \(\mapencoding(\atf_{k+1}) \subseteq I_{k+1}\).
		In particular \(\vec{v} \in \relevantvectors \cap U_k\) hence \(\vec{v} \in \mapencoding(\Tleq{k})\) by \cref{upward-closure-relevant-vectors}.
		By "strong injectivity", this implies that \(\atf_k \in \Tleq{k}\), so that \(\atf_{k+1} \in \Tleq{k+1}\).
		Therefore, \(\mapencoding(\atf_{k+1}) \cap I_{k+1} \ne \emptyset\) but \(\mapencoding(\atf_{k+1}) \subseteq \mapencoding(\Tleq{k+1}) \subseteq U_{k+1}\), which contradicts \(I_{k+1} \subseteq D_{k+1}\).
	\end{proof}

	\begin{lemma}
		\label{pret-intersects}
		There is \(\tftrans \in \Tmin\) such that \(\preflows{\tftrans}{I_{k+1}} \cap U_{k+1} \ne \emptyset\).
	\end{lemma}
	\begin{proof}
		By \cref{ikplisone-intersects}, there is \(\vec{v} \in I_{k+1} \cap (\mapencoding(\Tpow{k+2}) \setminus \mapencoding(\Tleq{k+1}))\).
		Let \(\atf = (f, \controlloc, \controlloc') \in \Tpow{k+2}\) such that \(\vec{v} \in \mapencoding(\atf)\).
		Because \(\vec{v}\notin \mapencoding(\Tleq{k+1})\), \(\atf \notin \Tleq{k+1}\).
		Also, \(\mapencoding(\atf) \subseteq I_{k+1}\).
		Indeed, by \cref{upward-closure-relevant-vectors}, \(\mapencoding(\atf) \cap U_{k+1} \subseteq \mapencoding(\Tleq{k+1})\) but \(\mapencoding(\atf) \cap \mapencoding(\Tleq{k+1}) = \emptyset\) by "strong injectivity", so that \(\mapencoding(\atf) \subseteq D_{k+1}\).
		This means that the representing vector of \(I_{k+1}\) must have value \(\omega\) on every \(i\) such that \(f(\indexof^{-1}(i)) = \dummysymb\) by maximality of \(I_{k+1}\) in \(D_{k+1}\), so that \(\vec{v} \in I_{k+1}\) implies that \(\vec{u} \in I_{k+1}\) for every \(\vec{u} \in \mapencoding(\atf)\).
		Overall, we have proved that \(\mapencoding(\atf) \subseteq I_{k+1}\).

		Because \(\atf \in \Tpow{k+2}\), there is \(\tftrans \in \Tmin\), \(\atf' \in \Tpow{k+1}\) such that \(\atf \in \atf' \tftimes \tftrans\).
		By definition of \(\preflows{\tftrans}{I_{k+1}}\), \(\mapencoding(\atf') \subseteq \preflows{\tftrans}{I_{k+1}}\).
		Also, \(\mapencoding(\atf') \subseteq \mapencoding(\Tleq{k+1}) \subseteq U_{k+1}\), so that \(\mapencoding(\atf') \subseteq \preflows{\tftrans}{I_{k+1}} \cap U_{k+1}\).
		By strong injectivity, \(\mapencoding(\atf') \ne \emptyset\) therefore \(\preflows{\tftrans}{I_{k+1}} \cap U_{k+1}\ne \emptyset\).
	\end{proof}

	In all the following, we fix \(\tftrans \in \Tmin\) such that \(\preflows{\tftrans}{I_{k+1}} \cap U_{k+1} \ne \emptyset\).
	By applying the definition, there are \(\atf_k, \atf_{k+1}\) such that \(\atf_{k+1} \in \atf_k \tftimes \tftrans\), \(\mapencoding(\atf_k) \subseteq \preflows{\tftrans}{I_{k+1}} \cap U_{k+1}\) and \(\mapencoding(\atf_{k+1}) \subseteq I_{k+1}\).
	We write \(\atf_{k+1} = (f_{k+1}, \controlloc_{k+1}, \controlloc_{k+1}')\).
	Also, let \(E \subseteq \nset{1}{d}\) be the set of components at which the representing vector of \(I_{k+1}\) is equal to \(\omega\).
	We know that \(m^2+1, m^2+2 \notin E\) as it would otherwise imply that \(\zerovectors \cap D_k \ne \emptyset\), which contradicts definition of \(U_k\)\footnote{In fact, this argument is the reason why we enforced that \(\zerovectors \subseteq U_k\).
	}.
	This means that \(E \subseteq \nset{1}{m^2}\).
	Let \(S \deff \indexof^{-1}(E)\); \(S\) is the set of pairs of states \((q,q')\) such that \(\indexof(q,q') \in E\).
	For every \(j \in \nats\), for every "transfer flow" \(\atf = (f , \controlloc, \controlloc')\), we denote by \(\atf^{(j)}\) the "transfer flow" \((f^{(j)},\controlloc, \controlloc')\) where \(f^{(j)}\) is such that, for all \(q,q' \in \states\):
	\begin{itemize}[noitemsep,topsep=0pt,parsep=0pt,partopsep=0pt]
		\item if \((q,q') \notin S\) then \(f^{(j)}(q,q') = f(q,q')\);
		\item if \((q,q') \in S\) and \(f(q,q') \ne \dummysymb\) then \(f^{(j)}(q,q') = \max(f(q,q'), j)\);
		\item if \((q,q') \in S\) and \(f(q,q') = \dummysymb\) then \(f^{(j)}(q,q') = \dummysymb\).
	\end{itemize}

	Intuitively, \(\atf^{(j)}\) is equal to \(\atf\) except that, in all components in \(E\) where \(\atf\) does not have value \(\dummysymb\), the values will tend to infinity as \(j\) grows.
	We define a similar notion for vectors.
	For every \(j \in \nats\), for every vector \(\vec{v}\), we let \(\vec{v}^{(j)}\) be the vector defined by, for all \(i \in \nset{1}{d}\):
	\begin{itemize}[noitemsep,topsep=0pt,parsep=0pt,partopsep=0pt]
		\item if \(i \in E\) and \(\vec{v}^{(j)}(i) = \max(\vec{v}(i),j)\)
		\item if \(i \notin E\) then \(\vec{v}^{(j)}(i) = \vec{v}(i)\).
	\end{itemize}

	We connect the definition of \(\vec{v}^{(j)}\) and of \(\atf^{(j)}\) with the following lemma:

	\begin{lemma}
		\label{vj-tfj}
		Let \(\atf \in \transferflows\) and \(\vec{v} \in \mapencoding(\atf)\).
		For all \(j \in \nats\), \(\vec{v}^{(j)} \in \mapencoding(\atf^{(j)})\).
	\end{lemma}
	\begin{proof}
		Let \(j \in \nats\), \(\atf = (f, \controlloc, \controlloc')\) and \(\vec{v} \in \mapencoding(\atf)\).
		Let \(i \in \nset{1}{m^2}\) and \((q,q') \deff \indexof^{-1}(i)\).
		First, if \(i \notin E\) then \((q,q') \notin S\), so that \(\vec{v}^{(j)}(i) = \vec{v}(i) = f(q,q') = f^{(j)}(q,q')\).
		Suppose now that \(i \in E\).
		If \(f(q,q') = \dummysymb\) then \(f^{(j)}(q,q') = \dummysymb\) and the value at component \(i\) in \(\vec{v}^{(j)}\) plays no role in whether \(\vec{v}^{(j)} \in \mapencoding(\atf^{(j)})\).
		If \(f(q,q') \ne \dummysymb\) then \(f(q,q') = \vec{v}(i)\) so that \(\vec{v}^{(j)}(i) = \max(\vec{v}(i),j) = \max(f(q,q'),j) = f^{(j)}(q,q')\), concluding the proof.
	\end{proof}

	By applying this construction to \(\atf_{k+1}\), we obtain a sequence that remains in \(\mapencoding^{-1}(I_{k+1})\):

	\begin{lemma}
		\label{atfj-remains-in-Ikplusone}
		For all \(j\), we have \(\mapencoding(\atf_{k+1}^{(j)}) \subseteq I_{k+1}\).
	\end{lemma}
	\begin{proof}
		By definition of \(\atf_{k+1}\), \(\mapencoding(\atf_{k+1}) \subseteq I_{k+1}\).
		Let \(\vec{v}_j \in \mapencoding(\atf_{k+1}^{(j)})\).
		Let \(\vec{v}\) equal to \(\vec{v}_j\) except that \(\vec{v}(\indexof(q,q')) = f(q,q')\) for all \(q,q' \in S\) such that \(f(q,q') \ne \dummysymb\).
		We obtain that \(\vec{v} \in \mapencoding(\atf_{k+1})\) so that \(\vec{v}\in I_{k+1}\).
		Observe that, for such values of \(q\) and \(q'\), \(\vec{v}_j(q,q') = \max(j, f(q,q'))\), so that \(\vec{v}\leprod \vec{v}_j\).
		Moreover, \(\vec{v}_j\) is equal to \(\vec{v}\) on components that are not in \(E\), because by definition \(E = \indexof(S)\).
		By definition of \(E\), the representing vector of \(I_{k+1}\) is equal to \(\omega\) on all components in \(E\), so that membership in \(E\) is not sensitive to the values at these components; this proves that \(\vec{v}_j \in I_{k+1}\).
	\end{proof}

	The following lemma is where the magic really happens:

	\begin{lemma}
		\label{there-is-m}
		For all \(j \in \nats\), there is \(p \in \nats\) such that \(\atf_{k+1}^{(p)} \in \atf_k^{(j)} \tftimes \tftrans\).
	\end{lemma}
	\begin{proof}
		We proceed by induction on \(j\).
		For \(j=0\), \(\atf_k^{(0)} = \atf_k\), \(\atf_{k+1}^{(0)} = \atf_{k+1}\) and indeed \(\atf_{k+1} \in \atf_k \tftimes \tftrans\) so that the property holds by letting \(p=0\).

		We suppose that the property is true for \(j\), and we prove it for \(j+1\).
		By induction hypothesis, there is \(p\) such that \(\atf_{k+1}^{(p)} \in \atf_k^{(j)} \tftimes \tftrans\); let \(H_j : \states^3 \to \natsanddummy\) be a product witness of that.
		We build a function \(H_{j+1}: \states^3 \to \natsanddummy\) as follows.
		For every \((q_1,q_2) \in \states^2\):
		\begin{itemize}[noitemsep,topsep=0pt,parsep=0pt,partopsep=0pt]
			\item if \(f_{k}^{(j)}(q_1,q_2) \ne j\) or \(f_{k}^{(j+1)}(q_1,q_2) \ne j+1\), we set \(H_{j+1}(q_1,q_2,q_3) \deff H_j(q_1,q_2,q_3)\) for all \(q_3\);
			\item if \(f_{k}^{(j)}(q_1,q_2)= j\) and \(f_{k}^{(j+1)}(q_1,q_2)= j+1\), we set \(H_{j+1}(q_1,q_2,q_2) \deff H_j(q_1,q_2,q_2) +1\) and, for all \(q_3 \ne q_2\), \(H_{j+1}(q_1,q_2,q_3) \deff H_j(q_1,q_2,q_3)\).
		\end{itemize}
		There is a subtlety in the second case: it could be that \(H_j(q_1,q_2,q_2) = \dummysymb\), in which case we set \(H_{j+1}(q_1,q_2,q_2) = 1\) (recall that \(\dummysymb +1 = 1\)).
		We therefore do not always have \(H_{j+1} \geq H_j\).
		Let \(f: (q_1,q_3) \mapsto \sum_{q_2} H_{j+1}(q_1,q_2,q_3)\).
		We claim that \((f, \controlloc_{k+1}, \controlloc_{k+1}') \in \atf_k^{(j+1)} \tftimes \tftrans\), with \(H_{j+1}\) as product witness.

		First, we prove that, for all \(q_1,q_2\), \(\sum_{q_3} H_{j+1}(q_1,q_2,q_3) \geq f_k^{(j+1)}(q_1,q_2)\).
		Let \(q_1,q_2 \in \states\).
		Because \(H_j\) is a product witness that \(\atf_{k+1}^{(p)} \in \atf_k^{(j)} \tftimes \tftrans\), we have \(\sum_{q_3} H_j(q_1,q_2,q_3) \geq f_{k}^{(j)}(q_1,q_2)\).
		If \(f_{k}^{(j)}(q_1,q_2) \ne j\) or \(f_{k}^{(j+1)}(q_1,q_2) \ne j+1\), we have \(f_k^{(j+1)}(q_1,q_2) = f_k^{(j)}(q_1,q_2)\) and \(H_{j+1}(q_1,q_2,q_3) = H_j(q_1,q_2,q_3)\) for all \(q_3\), so that \(\sum_{q_3} H_{j+1}(q_1,q_2,q_3) = \sum_{q_3} H_{j}(q_1,q_2,q_3) \geq f_k^{(j)}(q_1,q_2) = f_k^{(j+1)}(q_1,q_2)\).
		If \(f_{k}^{(j)}(q_1,q_2) = j\) and \(f_{k}^{(j+1)}(q_1,q_2) = j+1\), then \(\sum_{q_3} H_j(q_1,q_2,q_3) \geq j\) and \(\sum_{q_3} H_{j+1}(q_1,q_2,q_3) = \sum_{q_3} H_j(q_1,q_2,q_3) +1 \geq j+1\).

		Let \(\tftrans =: (f_{\tftrans}, \controlloc_{\tftrans}, \controlloc'_\tftrans)\).
		We now prove that, for all \(q_2,q_3\), \(\sum_{q_1} H_{j+1}(q_1,q_2,q_3) \geq f_{\tftrans}(q_2,q_3)\).
		We know that this is true for \(H_j\).
		For every \(q_2 \ne q_3\), we have \(\sum_{q_1} H_{j+1}(q_2,q_3) = \sum_{q_1} H_j(q_2,q_3) \geq f_{\tftrans}(q_2,q_3)\).
		For every \(q_2\), we have either \(\sum_{q_1} H_{j+1}(q_1,q_2,q_2) = \sum_{q_1} H_{j}(q_1,q_2,q_2)\) or \(\sum_{q_1} H_{j+1}(q_1,q_2,q_2) = \sum_{q_1} H_{j}(q_1,q_2,q_2)+1\).
		By "idle-compliance", we have \(f_{\tftrans}(q_2,q_2) \ne \dummysymb\) so that \(\sum_{q_1} H_{j}(q_1,q_2,q_2) \ne \dummysymb\), therefore \(\sum_{q_1} H_{j+1}(q_1,q_2,q_2) \geq \sum_{q_1} H_{j}(q_1,q_2,q_2) \geq f_{\tftrans}(q_2,q_2)\).
		Note that we need "idle-compliance" here: if we had \(f_{\tftrans}(q_2,q_2) = \dummysymb\) then we would have \(\sum_{q_1} H_{j}(q_1,q_2,q_2) = \dummysymb\) but \(\sum_{q_1} H_{j+1}(q_1,q_2,q_2)= 1\) so that \(\sum_{q_1} H_{j+1}(q_1,q_2,q_2)\) would be incomparable with \(f_{\tftrans}(q_2,q_2)\).

		It remains to prove that \(H_{j+1}\) is a product witness that \(\atf_{k+1}{(p')} \in \atf_k^{(j+1)} \tftimes \tftrans\) for some \(m'\).
		To do that, let \(f : (q_1,q_3) \mapsto \sum_{q_2} H_{j+1}(q_1,q_2,q_3)\) and \(\atf_{k+1}' \deff (f, \controlloc_{k+1}, \controlloc_{k+1}')\).
		By the above arguments, we know that \(\atf_{k+1}' \in \atf_{k}^{(j+1)} \tftimes \tftrans\).
		It therefore suffices to prove that there exists \(p'\) such that \((f, \controlloc_{k+1}, \controlloc_{k+1}') \letf (f_{k+1}^{(p')}, \controlloc_{k+1}, \controlloc_{k+1}') = \atf_{k+1}^{(p')}\).
		Indeed, this would imply that \(\atf_{k+1}^{(p')} \in \atf_k^{(j+1)} \tftimes \tftrans\) by \ref{basic1}.

		We now prove that there is \(p' \in \nats\) such that \(\atf_{k+1}' \letf \atf_{k+1}^{(p')}\).
		Let \(q_1,q_3 \in \states\).
		If \((q_1,q_3) \notin S\) then we must have \(f_k^{(j)}(q_1,q_2) = f_k^{(j+1)}(q_1,q_2)\) so that, by definition of \(H_{j+1}\), \(\sum_{q_2} H_{j+1}(q_1,q_2,q_3) = \sum_{q_2} H_j(q_1,q_2,q_3) = f_{k+1}^{(p)}(q_1,q_3) = f_{k+1}^{(p')}(q_1,q_3)\) for all \(p'\).
		Suppose now that \((q_1,q_3) \in S\).
		There are two cases: \(f_{k+1}(q_1,q_3)= \dummysymb\) and \(f_{k+1}(q_1,q_3) \ne \dummysymb\).

		Let \((q_1,q_3) \in S\) such that \(f_{k+1}(q_1,q_3) = \dummysymb\).
		We must prove that \(f(q_1,q_3) = \dummysymb\).
		Recall that we have \(\atf_{k+1} \in \atf_k \tftimes \tftrans\).
		Upon defining the "compositional product" \(\tftimes\), we have made sure that \(\atf_{k+1} \in \atf_k \tftimes {\tftrans}\) implies that, for all \(q,q'\), if \(f_{k+1}(q,q') = \dummysymb\) then \(f_k(q,q') = \dummysymb\).
		This proves that \(f_k(q_1,q_3) = \dummysymb\).
		By definition of \(f_k^{(j+1)}\), this implies that \(f_k^{(j+1)}(q_1,q_3) = \dummysymb\).
		In particular, \(f_k^{(j+1)}(q_1,q_3) \ne j+1\) so that, by definition of \(H_{j+1}\), for all \(q_2\), we have \(H_{j+1}(q_1,q_2,q_3) = H_j(q_1,q_2,q_3)\) for all \(q_2\).
		However, \(\sum_{q_2} H_{j}(q_1,q_2,q_3) = f_{k+1}^{(p)}(q_1,q_3) = \dummysymb\), so that \(f(q_1,q_3) = \sum_{q_2} H_{j+1}(q_1,q_2,q_3) = \sum_{q_2} H_j(q_1,q_2,q_3) = \dummysymb\).

		Let \((q_1,q_2) \in S\) such that \(f_{k+1}(q_1,q_3) \ne \dummysymb\); we have \(f_{k+1}^{(p')}(q_1,q_3) = \max(f_{k+1}(q_1,q_3), p')\) for all \(p'\).
		Also, \(f_{k+1}^{(p)}(q_1,q_3) \ne \dummysymb\) therefore \(\sum_{q_2} H_j(q_1,q_2,q_3) \ne \dummysymb\).
		By definition of \(H_{j+1}\), this also implies \(\sum_{q_2} H_{j+1}(q_1,q_2,q_3) \ne \dummysymb\) so that \(f(q_1,q_3) \ne \dummysymb\).
		For \(p' \geq f(q_1,q_3)\), we have \(f(q_1,q_3) \leq f_{k+1}^{(p')}(q_1,q_3)\).

		Let \(p'\) large enough so that, for every \(q_1,q_3\) such that \(f(q_1,q_3) \ne \dummysymb\), \(p' \geq f(q_1,q_3)\).
		We have proved that \(\atf_{k+1}' \letf \atf_{k+1}^{(p')}\).
		This implies that \(\atf_{k+1}^{(p')} \in \atf_k^{(j+1)} \tftimes \tftrans\), concluding the induction.
	\end{proof}

	We now claim that, for all \(j \in \nats\), \(\mapencoding(\atf_k^{(j)}) \cap U_k = \emptyset\).
	Indeed, suppose by contradiction that this is not the case: let \(j\) such that \(\mapencoding(\atf_k^{(j)}) \cap U_k \ne \emptyset\).
	By \cref{upward-closure-relevant-vectors} and by strong injectivity, this implies that \(\atf_k^{(j)} \in \Tleq{k}\).
	We apply \cref{there-is-m} to obtain \(p\) such that \(\atf_{k+1}^{(p)} \in \atf_k^{(j)} \tftimes \tftrans\), so that \(\atf_{k+1}^{(p)} \in \Tleq{k+1}\).
	We have \(\mapencoding(\atf_{k+1}^{(p)}) \subseteq U_{k+1}\), but by \cref{atfj-remains-in-Ikplusone} \(\mapencoding(\atf_{k+1}^{(p)}) \subseteq I_{k+1}\).
	Because \(\mapencoding(\atf_{k+1})^{(p)} \) is not empty by strong injectivity, this contradicts that \(I_{k+1} \subseteq D_{k+1}\).

	We have proved that \(\mapencoding(\atf_k^{(j)}) \subseteq D_k\) for every \(j\).
	Recall that our objective is to exhibit an ideal \(I_k\) "proper at step \(k\)" whose representing vector is equal to \(\omega\) at any component in \(E\).
	Let \(\vec{v}_{k+1}\) be an arbitrary vector in \(\mapencoding(\atf_{k+1})\), and \(\vec{v}_k\) be an arbitrary vector of \(\mapencoding(\atf_k)\).

	\begin{lemma}
		\label{vk-vkplusone-sets} For every \(j \in \nats\):
		\begin{itemize}[noitemsep,topsep=0pt,parsep=0pt,partopsep=0pt]
			\item \(\vec{v}_{k+1}^{(j)} \in I_{k+1}\),
			\item \(\vec{v}_k^{(j)} \in U_{k+1} \setminus U_k\).
		\end{itemize}
	\end{lemma}
	\begin{proof}
		Let \(j \in \nats\).
		By \cref{vj-tfj}, \(\vec{v}_k^{(j)} \in \mapencoding(\atf_k^{(j)})\) and \(\vec{v}_{k+1}^{(j)} \in \mapencoding(\atf_{k+1}^{(j)})\).
		By \cref{atfj-remains-in-Ikplusone}, this directly proves that \(\vec{v}_{k+1}^{(j)} \in I_{k+1}\).

		We have \(\mapencoding(\atf_k) \subseteq U_{k+1}\) therefore \(v_k \in U_{k+1}\), thus \(\vec{v}_{k}^{(j)} \in U_{k+1}\) because \(U_{k+1}\) is "upward-closed".
		Suppose by contradiction that we have \(\vec{v}_{k}^{(j)} \in U_k\).
		This would imply that \(\mapencoding(\atf_k^{(j)}) \cap U_k \ne \emptyset\); by \cref{upward-closure-relevant-vectors} and by strong injectivity, this implies that \(\atf_k^{(j)} \in \Tleq{k}\).
		By \cref{there-is-m}, there is \(p\) such that \(\atf_{k+1}^{(p)} \in \atf_k^{(j)} \tftimes \tftrans\), so that \(\atf_{k+1}^{(p)} \in \Tleq{k+1}\).
		This implies that \(\mapencoding(\atf_{k+1}^{(p)}) \subseteq U_{k+1}\), which contradicts \cref{atfj-remains-in-Ikplusone} since \(I_{k+1}\subseteq D_{k+1}\).
	\end{proof}

	Let \(\vec{u}_k\) be the vector such that, for all \(i \in \nset{1}{d}\), \(\vec{u}_k(i) \deff \omega\) if \(i \in E\) and \(\vec{u}_k(i) \deff \vec{v}_k(i)\) if \(i \notin E\).
	Let \(J\) be the ideal represented by \(\vec{u}_k\), \emph{i.e.}, \(J := \set{\vec{u} \in \nats^d \mid \vec{u} \leprod \vec{u}_k}\).
	In particular, \(J\) contains vector \(\vec{v}_k^{(j)}\) for every \(j \in \nats\), which are all in \(U_{k+1} \setminus U_k\) by \cref{vk-vkplusone-sets}.
	This implies that \(J \nsubseteq D_{k+1}\).
	We now prove that \(J \subseteq D_k\).
	Let \(\vec{u} \in J\), and let \(j \deff \norm{\vec{u}}\).
	We have \(\vec{u} \leprod \vec{v}_k^{(j)}\): for all \(i \notin E\), \(\vec{u}(i) \leq \vec{u}_k(i) = \vec{v}_k^{(j)}(i)\), and for \(i \in E\), \(\vec{u}_k(i) \leq \norm{\vec{u}} = j \leq \vec{v}_k^{(j)}(i)\).
	Because \(\vec{v}_k^{(j)} \in D_k\) by \cref{vk-vkplusone-sets}, we have proved that \(\vec{v} \in D_k\).
	That being true for all \(\vec{v} \in J\), this proves that \(J \subseteq D_k\).
	In particular, \(J\) is contained in some ideal \(I_k\) in the "decomposition" of \(D_k\); because \(J \nsubseteq D_{k+1}\), \(I_k\) is "proper at step \(k\)".
	The representing vector of \(J\) is equal to \(\omega\) on every \(i \in E\), therefore the same is true for the representing vector of \(I_k\), concluding the proof of \cref{Dk-controlled}.
\end{appendixproof}

We apply \cref{bound-descending-chains} on \((D_k)\) to prove that \((D_k)\) and \((U_k)\) stabilize at index at most \((N+1)^{3^{d} (\log(d)+1)} \leq (M+1)^{3^{m^2+2} \cdot 2(\log(m^2+2) +1) m^2} = B\), so that \(\transfersof{\transitions^*} = \transfersof{\transitions^{\leq B}}\).
By above, transfer flows in \(\minwqo{\Tleq{k}}\) have weight bounded by \(k\), therefore transfer flows of \(\minwqo{\transfersof{\transitions^*}} = \minwqo{\Tleq{B}}\) have weight at most \(2B\).
This concludes the proof of \cref{structural-theorem-flows}.
%
%
\section{Conclusion}
When compared to the \NEXPTIME result for LTL\(\setminus \ltlnext\) verification of shared-memory systems with pushdown machines \cite{FortinMW17}, our 2-\EXPSPACE LTL result may seem weak.
However, their techniques are quite specific, while ours are generic, enabling us to go from LTL to monadic HyperLTL with little extra work.
Additionally, we believe transfer flows, \(K\)-blind sets and the results thereof apply to other problems and systems, such as reconfigurable broadcast networks \cite{DSTZ12} or asynchronous shared-memory systems \cite{EGM13}, which enjoy a similar monotonicity property to IOPP.
Most problems considered in this paper are undecidable; this was to be expected for infinite-state systems.
However, our decidability result (\cref{thm:hyperltl-iopp}) sheds light on a decidable fragment, suggesting that further research on verification of hyperproperties for infinite-state systems should be pursued.
%

%
%
%
%
%
%
%
%
%
%
%
%
%
%
%
%
%
%
\bibliographystyle{plain}
\end{document}